\documentclass[conference]{IEEEtran}

\newif\ifarxiv%
\arxivtrue%

\newif\ifreview%

\newif\ifanonymous%

\usepackage[Tol]{colorblind}
\usepackage[dvipsnames, table]{xcolor}
\usepackage{hyperref}
\usepackage{graphicx}
\usepackage{subcaption}
\usepackage{orcidlink}

\usepackage{booktabs}
\usepackage[inline]{enumitem}
\usepackage{multirow, bigdelim}
\usepackage{comment}

\usepackage{tikz}
\usetikzlibrary{
  arrows.meta,
  calc,
  fit,
  positioning
}

\hypersetup{
  colorlinks,
  linkcolor={T-Q-MC6},
  citecolor={T-Q-MC7},
  urlcolor={T-Q-MC7}
}

\colorlet{HighlightColor}{T-Q-MC7}
\colorlet{AlertColor}{T-Q-MC5}

\NewDocumentCommand{\Paragraph}{m}{%
  {\par\noindent\normalfont\normalsize\bfseries#1}%
}

\usepackage[numbers, sort]{natbib}
\NewDocumentCommand{\CiteAuthor}{m}{%
  {\protect\NoHyper\citeauthor{#1}\protect\endNoHyper}}

\newif\ifcomment%

\ifcomment
  \newcommand{\LS}[1]{{\color{T-Q-V3}#1}}
  \newcommand{\NF}[1]{{\color{T-Q-V4}#1}}
  \newcommand{\MS}[1]{{\color{T-Q-V5}#1}}

  \newcommand{\LuigiSays}[2][All]{\LS{\textbf{Luigi $\to$ #1:} #2}}
  \newcommand{\TashSays}[2][All]{\NF{\textbf{NF $\to$ #1:} #2}}
  \newcommand{\MarioSays}[2][All]{\MS{\textbf{Mário $\to$ #1:} #2}}
\else
  \newcommand{\LS}[1]{#1}
  \newcommand{\NF}[1]{#1}
  \newcommand{\MS}[1]{#1}

  \makeatletter
  \newcommand{\LuigiSays}[2][All]{\@bsphack\@esphack}
  \newcommand{\TashSays}[2][All]{\@bsphack\@esphack}
  \newcommand{\MarioSays}[2][All]{\@bsphack\@esphack}
  \makeatother
\fi

\ifreview
  \NewDocumentCommand{\Review}{>{\TrimSpaces}+m}{{\color{HighlightColor}#1}}
\else
  \NewDocumentCommand{\Review}{+m}{#1}
\fi

\usepackage{amsmath}
\usepackage{amsfonts}
\usepackage{amssymb}
\usepackage{amsthm}
\usepackage{blkarray}
\usepackage{environ}
\usepackage[scr=rsfso]{mathalfa}
\usepackage{mathtools}
\usepackage{stmaryrd}
\usepackage{thmtools, thm-restate}
\usepackage{xfrac}

\allowdisplaybreaks%

\newcommand*{\Adv}{\mathscr{A}}
\newcommand*{\Target}{Alex}
\newcommand*{\TrueAdv}{\Adv^{\Chan{P}}}
\newcommand*{\NoiseAdv}{\Adv^{\Chan{S}}}

\newcommand{\qm}[1]{``#1''} 

\NewDocumentCommand{\Comment}{m}{\tag*{%
    \footnotesize%
    \begin{tabular}{@{} >{\# } l @{}}%
      #1%
    \end{tabular}%
}}

\DeclareMathOperator*{\argmax}{argmax}
\DeclareMathOperator*{\argmin}{argmin}
\DeclareMathOperator{\FnComp}{\circ}

\newcommand*{\CalK}{\mathcal{K}}
\newcommand*{\CalL}{\mathcal{L}}
\newcommand*{\CalQ}{\mathcal{Q}}

\newcommand*{\CalW}{\mathcal{W}}
\newcommand*{\CalX}{\mathcal{X}}
\newcommand*{\CalY}{\mathcal{Y}}
\newcommand*{\CalZ}{\mathcal{Z}}
\newcommand*{\EqDef}{\stackrel{\text{def}}{=}}

\newcommand*{\Reals}{\mathbb{R}}
\newcommand*{\Supp}{\mathrm{supp}}

\newcommand*{\Dist}{\mathbb{D}}
\NewDocumentCommand{\PointDist}{m}{[#1]}

\NewDocumentCommand{\Card}{m}{|#1|}
\NewDocumentCommand{\Map}{m}{\{#1\}}
\NewDocumentCommand{\MapEntry}{m m}{#1\colon #2}
\NewDocumentCommand{\Set}{m}{\{#1\}}
\NewDocumentCommand{\OrdSet}{s m}{%
  \IfBooleanTF{#1}%
  {\left\{#2\right\}}%
  {\{#2\}}%
}
\NewDocumentCommand{\Tuple}{s m}{%
  \IfBooleanTF{#1}%
  {\left\langle#2\right\rangle}%
  {\langle#2\rangle}%
}

\DeclareMathAlphabet{\mathbbold}{U}{bbold}{m}{n}

\DeclareMathOperator{\JointOp}{\kern-0.6pt\smalltriangleright\kern.3pt}
\DeclareMathOperator{\RefinedBy}{\sqsubseteq}
\DeclareMathOperator{\SubsOp}{\parallel}

\NewDocumentCommand{\Chan}{m +o}{%
  \IfNoValueTF{#2}
  {\mathrm{#1}}
  {\mathrm{#1}^{\textnormal{#2}}}%
}
\NewDocumentCommand{\Prior}{m +o}{%
  \IfNoValueTF{#2}
  {#1}
  {#1^{\textnormal{#2}}}%
}

\newcommand*{\ChanLeaksNothing}{\mathbbold{1}}

\NewDocumentCommand{\ChanLDelim}{m}{\ldelim[{#1}{4mm}}
\NewDocumentCommand{\ChanRDelim}{m}{\rdelim]{#1}{4mm}}
\NewEnviron{ChanEnv}[3]{%
  \begin{blockarray}{c@{\quad}>{\;\;}*{#1}{c}<{\;\;}}
    #2 & #3 \\[.5ex]
    \begin{block}{c@{\quad}[>{\;\;}*{#1}{c}<{\;\;}]}
      \BODY \\
    \end{block}
  \end{blockarray}
}

\NewDocumentCommand{\Gain}{O{}}{g_{\textnormal{#1}}}
\NewDocumentCommand{\Loss}{O{}}{\ell_{\textnormal{#1}}}
\newcommand*{\GainBayes}{\mathbf{1}}
\newcommand*{\LossShannon}{\Loss[$H$]}

\NewDocumentCommand{\Joint}{s m m O{X, Y}}{%
  \IfBooleanTF{#1}%
  {\left(#2 \JointOp #3\right)_{#4}}%
  {(#2 \JointOp #3)_{#4}}%
}

\NewDocumentCommand{\Outer}{s m m O{Y}}{%
  \IfBooleanTF{#1}%
  {\Joint*{#2}{#3}[#4]}%
  {\Joint{#2}{#3}[#4]}%
}

\NewDocumentCommand{\Posterior}{s m m O{X} O{Y}}{%
  \IfBooleanTF{#1}%
  {\Joint*{#2}{#3}[#4 \mid #5]}%
  {\Joint{#2}{#3}[#4 \mid #5]}%
}

\NewDocumentCommand{\Hyper}{s m m}{%
  \IfBooleanTF{#1}%
  {\left[#2 \JointOp #3\right]}%
  {[#2 \JointOp #3]}%
}

\NewDocumentCommand{\Vg}{O{g}}{V_{#1}}
\NewDocumentCommand{\VgMax}{O{g}}{\Vg[#1]^{\max}}
\NewDocumentCommand{\Ul}{O{\ell}}{U_{#1}}
\NewDocumentCommand{\UlMin}{O{\ell}}{U_{#1}^{\min}}
\NewDocumentCommand{\LeakFn}{O{g}}{\CalL_{#1}}
\NewDocumentCommand{\Leak}{O{g} m m}{\LeakFn[#1](#3, #2)}
\NewDocumentCommand{\LeakMax}{O{g} m m}{\LeakFn[#1]^{\max}(#3, #2)}
\NewDocumentCommand{\LeakMin}{O{g} m m}{\LeakFn[#1]^{\min}(#3, #2)}
\NewDocumentCommand{\LeakDyn}{O{g} m m}{\LeakFn[#1](#2 \SubsOp #3)}

\NewDocumentCommand{\St}{O{g}}{\mathrm{st}_{#1}}
\NewDocumentCommand{\StVgFn}{O{g}}{\Vg[#1]^{\mathrm{st}}}
\NewDocumentCommand{\StVg}{O{g} m m}{\StVgFn[#1](#2 \SubsOp #3)}
\NewDocumentCommand{\StUl}{O{\Loss} m m}{\Ul[#1]^{\mathrm{st}}(#2 \SubsOp #3)}
\NewDocumentCommand{\StUlFn}{O{\Loss}}{\Ul[#1]^{\mathrm{st}}}
\NewDocumentCommand{\StLeakFn}{O{g}}{\CalL_{#1}^{\mathrm{st}}}
\NewDocumentCommand{\StLeak}{O{g} m m}{\StLeakFn[#1](#2 \SubsOp #3)}
\NewDocumentCommand{\SuppSt}{O{g}}{\Supp \FnComp \St}

\newcommand*{\Cascade}{\mathbin{;}}
\newcommand*{\ParComp}{\circledparallel}
\NewDocumentCommand{\IChoice}{O{p}}{\mathbin{{}_{#1}\oplus}}

\DeclareFontFamily{U}{mathb}{\hyphenchar\font45}
\DeclareFontShape{U}{mathb}{m}{n}{
	<5> <6> <7> <8> <9> <10> gen * mathb
	<10.95> mathb10 <12> <14.4> <17.28> <20.74> <24.88> mathb12
}{}
\DeclareSymbolFont{mathb}{U}{mathb}{m}{n}
\DeclareMathSymbol{\smalltriangleright}{2}{mathb}{"9B} 

\DeclareFontEncoding{LS1}{}{}
\DeclareFontSubstitution{LS1}{stix2}{m}{n}
\DeclareSymbolFont{stix2-frak}{LS1}{stix2frak}{m}{n}
\DeclareMathSymbol{\circledparallel}{\mathbin}{stix2-frak}{"A6}

\theoremstyle{plain}
\newtheorem{theorem}{Theorem}
\newtheorem{lemma}{Lemma}
\theoremstyle{definition}
\newtheorem{definition}{Definition}
\newtheorem{remark}{Remark}

\theoremstyle{definition}
\newtheorem{example}{Example} 

\NewDocumentEnvironment{Example}{m}{%
  \pushQED{\qed}%
  \example%
  \label{#1}
}{%
  \popQED%
  \endexample%
}

\makeatletter
\NewDocumentCommand{\exContinued@makeGroup}{m}{%
  \newcounter{exContinued@group@#1}%
  \setcounter{exContinued@group@#1}{\theexample}%
}

\NewDocumentCommand{\exContinued@makeSub}{m}{%
  \newcounter{exContinued@sub@#1}%
  \setcounter{exContinued@sub@#1}{1}%
}

\NewDocumentCommand{\exContinued@initGroupRef}{m}{%
  \exContinued@makeGroup{#1}%
  \expandafter\gdef\csname exContinued@groupRef@#1\endcsname{%
    exContinued@group@#1%
  }%
}

\NewDocumentCommand{\exContinued@initSubRef}{m}{%
  \exContinued@makeSub{#1}%
  \expandafter\gdef\csname exContinued@subRef@#1\endcsname{%
    exContinued@sub@#1%
  }%
}

\NewDocumentCommand{\exContinued@getGroup}{m}{%
  \ifcsname exContinued@groupRef@#1\endcsname%
    \@nameuse{exContinued@groupRef@#1}%
  \else%
    \PackageError{ExampleContinued}{%
      Label ``#1'' possibly used with ExampleContinued
      but has not been introduced with ExampleHead%
    }{%
      Define the example first, e.g.,\MessageBreak
      \string\begin{ExampleHead}\MessageBreak
      \string\label{...}\MessageBreak
        ...\MessageBreak
      \string\end{ExampleHead}%
    }
  \fi%
}

\NewDocumentCommand{\exContinued@getSub}{m}{%
  \ifcsname exContinued@subRef@#1\endcsname%
    \@nameuse{exContinued@subRef@#1}%
  \else%
    \PackageError{ExampleContinued}{%
      Label ``#1'' possibly used with ExampleContinued
      but has not been introduced with ExampleHead%
    }{%
      Define the example first, e.g.,\MessageBreak
      \string\begin{ExampleHead}\MessageBreak
      \string\label{...}\MessageBreak
        ...\MessageBreak
      \string\end{ExampleHead}%
    }
  \fi%
}

\NewDocumentCommand{\exContinued@addGroupRef}{m m}{%
  \expandafter\gdef\csname exContinued@groupRef@#1\endcsname{%
    \exContinued@getGroup{#2}%
  }%
}

\NewDocumentCommand{\exContinued@addSubRef}{m m}{%
  \expandafter\gdef\csname exContinued@subRef@#1\endcsname{%
    \exContinued@getSub{#2}%
  }%
}

\theoremstyle{definition}
\newtheorem{exampleGroup}{Example} 

\NewDocumentEnvironment{ExampleHead}{m}{%
  \pushQED{\qed}%
  \refstepcounter{example}%
  \exContinued@initSubRef{#1}%
  \exContinued@initGroupRef{#1}%
  \renewcommand{\theexampleGroup}{%
    \arabic{\exContinued@getGroup{#1}}\alph{\exContinued@getSub{#1}}%
  }%
  \exampleGroup%
  \label{#1}
}{%
    \popQED%
    \endexampleGroup%
    \stepcounter{\exContinued@getSub{#1}}%
}

\NewDocumentCommand{\ThmContinues}{m}{%
  \hyperref[#1]{continuing} from page~\pageref{#1}%
}

\NewDocumentEnvironment{ExampleContinued}{m m}{%
  \pushQED{\qed}%
  \exContinued@addSubRef{#2}{#1}%
  \exContinued@addGroupRef{#2}{#1}%
  \renewcommand{\theexampleGroup}{%
    \arabic{\exContinued@getGroup{#2}}\alph{\exContinued@getSub{#2}}%
  }%
  \exampleGroup[\ThmContinues{#1}]%
  \label{#2}
}{%
  \popQED%
  \endexampleGroup%
  \stepcounter{\exContinued@getSub{#2}}%
}

\makeatother

\title{A new measure for dynamic leakage based on quantitative information flow}

\newif\ifauthorlong%

\ifanonymous%
  \author{\IEEEauthorblockN{Anonymous Authors}}
\else \ifauthorlong%
    \author{
      \IEEEauthorblockN{
        Luigi~D. C.~Soares\IEEEauthorrefmark{1}\IEEEauthorrefmark{2}%
        \orcidlink{0000-0002-9579-8427},
        Mário S.~Alvim\IEEEauthorrefmark{1}\orcidlink{0000-0002-4196-7467} and
        Natasha Fernandes\IEEEauthorrefmark{2}\orcidlink{0000-0002-9212-7839}
      } 
      \IEEEauthorblockA{
        \IEEEauthorrefmark{1}%
        \textit{Department of Computer Science},
        \textit{Federal University of Minas Gerais} ---
        Belo Horizonte, Brazil 
        \\
        \IEEEauthorrefmark{2}%
        \textit{School of Computing},
        \textit{Macquarie University} ---
        Sydney, Australia
      }
    }
  \else
    \author{
      \IEEEauthorblockN{Luigi~D. C.~Soares\,\orcidlink{0000-0002-9579-8427}}
      \IEEEauthorblockA{
        \textit{UFMG \& Macquarie University} \\
        Belo Horizonte, Brazil / Sydney, Australia \\
      }
      \and
      \IEEEauthorblockN{Mário S.~Alvim\,\orcidlink{0000-0002-4196-7467}} 
      \IEEEauthorblockA{
        \textit{UFMG} \\
        Belo Horizonte, Brazil \\
      }
      \and
      \IEEEauthorblockN{Natasha Fernandes\,\orcidlink{0000-0002-9212-7839}} 
      \IEEEauthorblockA{
        \textit{Macquarie University} \\
        Sydney, Australia \\
      }
    }
  \fi
\fi

\begin{document}
\maketitle

\begin{abstract}
  \emph{Quantitative information flow (QIF)} is concerned with assessing the leakage of information in computational systems.
  In QIF there are two main perspectives for the quantification of leakage.
  On one hand,
  the \emph{static perspective} considers all possible runs of the system in the computation of information flow,
  and is usually employed when preemptively deciding whether or not to run the system.
  On the other hand,
  the \emph{dynamic perspective} considers only a specific, concrete run of the system that has been realised,
  while ignoring all other runs.
  The dynamic perspective is relevant for, e.g., system monitors and trackers,
  especially when deciding whether to continue 
  or to abort a particular run based on how much leakage has occurred up to a certain point.
  Although the static perspective of leakage is well-developed in the literature,
  the dynamic perspective still lacks the same level of theoretical maturity.
  In this paper we take steps towards bridging this gap with the following key contributions: 
  \begin{enumerate*}[label=(\roman*)]
    \item we provide a novel definition of dynamic leakage 
      that decouples the adversary's \emph{belief} about the secret value from a \emph{baseline} distribution on secrets 
      against which the success of the attack is measured;
    \item we demonstrate that our formalisation satisfies relevant information-theoretic axioms,
      including non-interference and relaxed versions of monotonicity and the data-processing inequality (DPI); 
    \item we identify under what kind of analysis strong versions of the axioms of monotonicity and the DPI might not hold,
      and explain the implications of this (perhaps counter-intuitive) outcome;
    \item we show that our definition of dynamic leakage is compatible with the well-established static perspective; and
    \item we exemplify the use of our definition 
      on the formalisation of attacks against privacy-preserving data releases.
  \end{enumerate*}
\end{abstract}


\section{Introduction}\label{sec:intro}

The field of \emph{quantitative information flow (QIF)} is concerned with assessing the amount of information that flows through a system.
This assessment is relevant both in cases in which the flow of information through the system is desirable 
(e.g., in a census dataset which produces output statistics reflecting trends in the population data fed as input)
and (more commonly) in cases where the flow of information is undesirable 
(e.g., in electronic voting protocols in which the final tally should not reveal any undue information about the ballot of each voter). 
QIF has been successfully applied to a wide variety of scenarios,
including privacy and security~\cite{%
  alvim2024privacy,alvim:23:CCS,jurado2021formal,fernandes2018processing,Alvim:15:JCS,Chatzi:2019,
  OpenBanking:Soares2025%
}, 
machine learning~\cite{Silva:22:IS,Viegas:20:IS},
robotics~\cite{Pimentel:18:JIRS},
recommendation systems~\cite{Moraes:17:ICTIR},
fairness~\cite{Alvim:23:CADE}, 
among others.

In the QIF framework a system is modelled as an information-theoretic channel taking in some secret input and producing some observable output. 
Attackers are Bayesian, modelled using a prior knowledge and a gain function describing their goals and capabilities.  
By observing the channel's output, 
the adversary can update their prior distribution on secret values to a collection of posterior, revised distributions, 
each occurring with its own probability.
The information leakage is defined as the difference between the ``vulnerability'' of (i.e., the ``information'' about) the secret,
from the point of view of the attacker, before and after passing through the channel 
(prior and posterior vulnerabilities, respectively).
These ideas are captured in the \emph{$g$-leakage} framework~\cite{Alvim:12:CSF}.

It is known~\cite{alvim2019axiomatization,QIF-Book:Alvim2020} that 
the vulnerability measures in the $g$-leakage framework satisfy some fundamental information-theoretic axioms.
Among these, three are of particular interest to this paper, as they deal with the conservation of information in a system.
The first, \emph{monotonicity}, says that, by observing the behaviour of the system,
the adversary can never end up with less information about the system's input than they had before.\footnote{%
  This is akin to Shannon's \emph{\qm{information can't hurt}} property for entropies,
  which means conditioning cannot increase entropy: $H(X \mid Y) \leq H(X)$~\cite{coverthomas:book}.%
}
The second, known as the \emph{data-processing inequality (DPI)},
says that the post-processing of a system's output cannot increase the information the adversary has about the system's input.\footnote{%
  This is a generalisation of Shannon's property of DPI that states that if $X~\rightarrow~Y~\rightarrow~Z$ form a Markov chain, 
  then $H(X \mid Y) \leq H(X \mid Z)$~\cite{coverthomas:book}.%
}
Lastly, \emph{non-interference} says that a system that never produces any useful output cannot leak information.

These axioms are satisfied under the \emph{static perspective} of the $g$-leakage framework,
which considers \emph{all possible outputs} of the system, under either \emph{expectation} 
(by averaging the vulnerability over all posteriors)
or \emph{max-case} (by taking the worst-case vulnerability over all posteriors)~\cite{alvim2019axiomatization,QIF-Book:Alvim2020}.
Static leakage is useful for analysing a system before it is executed, e.g., when preemptively deciding whether or not to run it.

However, there are situations in which it is desirable to measure leakage with respect to a single run of the system,
also known as the \emph{dynamic} perspective.
In this case only the posterior corresponding to the realised output is considered,
and all posteriors corresponding to non-realised outputs are ignored.
The dynamic perspective is particularly useful in leakage trackers or monitors,
which follow in real time the information leaked by a system 
and may allow or abort further execution based on the leakage that has occurred so far.

Although of significant interest, 
the literature lacks a sound definition for dynamic leakage that enjoys the properties of static leakage identified earlier.
More precisely, the ``traditional'' definition of dynamic leakage~\cite{DynLeak:Bielova2016,QIF-Book:Alvim2020},
motivated by the static definition of leakage,
simply compares the posterior vulnerability of a secret (computed from a specific observation) with the prior vulnerability.
However, this definition may lead to unexpected and incorrect conclusions, as we cover in Section~\ref{sec:motivation}.
Indeed, the QIF community is aware of the need for a more principled definition of dynamic leakage~\cite{DynLeak:Bielova2016}.

\Paragraph{Contributions.}
In this paper we propose a definition of dynamic information leakage that fills this gap in the QIF literature.
Our work makes the following key contributions:
\begin{enumerate}
\item We provide a novel definition of dynamic leakage that is parameterised by two states of knowledge (probability distributions):
  one modelling the adversary's \emph{belief} about the secret value ---~which is used to guide their actions~--- 
  and another one modelling the \emph{baseline} distribution on secrets ---~against which the success of the adversary's actions are measured.
  Such a distinction between states of knowledge has already been made by \CiteAuthor{Belief:Clarkson2009}~\cite{Belief:Clarkson2009}, 
  who defined information leakage in terms of the Kullback-Leibler (KL) divergence.
  We extend their work to more general adversarial models, via the $g$-leakage framework~\cite{g-leakage:Alvim2012}, 
  and generalise the notion of \qm{baseline} from \qm{strict truth} to a state of knowledge that is the \qm{closest to the truth} available.
  
\item We show that the proposed formalisation satisfies the fundamental axiom of non-interference,
  and satisfies relaxed versions (which we call ``single-step'') of the axioms of monotonicity and the data-processing inequality.
  We also show that additional post-processing steps might have the opposite effect to what is expected,
  increasing vulnerability (dually, decreasing uncertainty).
  This (perhaps) counter-intuitive result is also not always captured using the traditional definition of dynamic leakage.

\item We identify under what kind of analysis strong versions 
  (which we call ``multi-step'')
  of the axioms of monotonicity and the data-processing inequality might not hold,
  noting that results obtained under such scenarios do not invalidate 
  ---~and are in fact not comparable to~---
  the guarantees of single-step monotonicity and DPI.
  We also explain why, from the analyst's perspective,
  a break of monotonicity (i.e., a negative amount of leakage) naturally corresponds to privacy improvement.

\item We show that the proposed definition of dynamic measures is consistent with the static perspective.
  That is, the expectation of dynamic vulnerability (or uncertainty)
  corresponds to the traditional definition of expected vulnerability (uncertainty),
  and similar for the max-case.
  
\item We formalise attacks against privacy-preserving data-release systems using QIF 
  and study the application of the new definition of dynamic vulnerability 
  to evaluate the effect of adding a privacy mechanism to the system.
\end{enumerate}

\Paragraph{Outline of the paper.}
In Section~\ref{sec:motivation} we explore the failures of monotonicity
and DPI in the traditional definition of dynamic leakage,
motivating our novel definition.
In Section~\ref{sec:prelim} we provide the current definitions of leakage measures using the QIF framework,
both in the static and dynamic perspectives.
In Section~\ref{sec:formalisation} we propose a new definition of dynamic leakage,
show the necessary steps to adopt it and its limitations,
demonstrate that it satisfies non-interference,
and show that it partially satisfies monotonicity and the DPI.
We also show that this formalisation is consistent with the static perspective.
In Section~\ref{sec:case-studies} we analyse a real-world application of privacy-preserving data releases,
in which we compare two systems satisfying the DPI and show that the traditional definition of dynamic leakage is not adequate,
while our novel formalisation is.
In Section~\ref{sec:multi-step} we discuss the difference between single- and multi-step analysis,
and the implications of additional dynamic steps, which might break the axioms of monotonicity and the DPI.
In Section~\ref{sec:rw} we discuss related work.
Finally, in Section~\ref{sec:conclusion} we conclude and discuss future work.
\ifarxiv%
  Full proofs of lemmas and theorems are provided in Appendix~\ref{sec:appendix-proofs}.%
\fi%

\section{Motivating Examples}\label{sec:motivation}

In this section we exemplify how the traditional definition of dynamic leakage might lead to misleading conclusions.
For that, consider a Bayesian adversary with access to the system of interest, including implementation details, 
who starts with some prior knowledge about a secret value fed as input to the system and,
after observing the system's behaviour, updates their knowledge using Bayes' rule.
The amount of information leaked by the system is customarily defined as
the increase in the adversary's information about the secret, captured as the change in the vulnerability of the secret
---~that is, the amount of useful information to perform an attack~--- before and after the system's execution
(or, dually, the decrease in the adversary's uncertainty about the secret).

\subsection{The problem of monotonicity: The traditional definition of
dynamic leakage can produce incorrect negative values}%
\label{sec:motivation:mono}

Intuitively, we expect vulnerability/uncertainty measures to be \emph{monotonic}, 
meaning that observing an output should never decrease the attacker's
information, and so leakage should always be non-negative.
It has already been shown that monotonicity is indeed always preserved in the static perspective,
for all measures satisfying a few fundamental information-theoretic axioms~\cite[Thm 5.8]{QIF-Book:Alvim2020}.
However, as shown by \CiteAuthor{DynLeak:Bielova2016}~\cite{DynLeak:Bielova2016},
it is possible to obtain negative leakage when considering a traditional definition of dynamic leakage 
that simply subtracts the posterior uncertainty from the prior uncertainty.
To illustrate, we review \CiteAuthor{DynLeak:Bielova2016}'s Example 1:

\begin{ExampleHead}{ex:motivation:shannon}
  Consider a (deterministic) program 
  that takes a 2-bit secret $X \in \OrdSet{\textnormal{00}, \textnormal{10}, \textnormal{11}}$ 
  and produces a value $Y \in \OrdSet{a, b}$:
  \begin{center}
    if $X = \textnormal{00}$ then $Y = a$ else $Y = b$
  \end{center}
  Now, suppose that the adversary's prior knowledge about the secret is modelled by the distribution 
  $p_{X} = \OrdSet{\sfrac{7}{8}, \sfrac{1}{16}, \sfrac{1}{16}}.$
  (We write $p_{X} = \OrdSet{p_{x_{i}}}_{i = 0}^{n - 1}$, matching the indices in $\OrdSet{x_{i}}_{i = 0}^{n - 1}$.)
  Then, let us say that the program is executed and produces an output $b$ that rules out secret value $00$ with certainty,
  leading the adversary to a posterior knowledge of $p_{X \mid b} = \OrdSet{0, \sfrac{1}{2}, \sfrac{1}{2}}.$
  
  A popular measure of uncertainty is the Shannon entropy $H(p)$ of a probability distribution $p : \Dist\mathcal{X}$,
  defined as~\cite{thomas2006elements}
  \begin{equation}\label{eq:intro:shannon-entropy}
    H(p) \EqDef - \sum_{x \in \CalX} p_{x}\, \log_{2} p_{x}
  \end{equation}
  Recall that we expect a \emph{decrease} of the adversary's uncertainty once an output is observed.
  Hence, in order to have positive values of leakage representing some flow of information,
  we would typically define the (dynamic) Shannon leakage as
  \begin{equation}\label{eq:intro:shannon-leak}
    H(p_{X}) - H(p_{X \mid y})
  \end{equation}
  Therefore, going back to our example, the dynamic leakage caused by observing output $b$ is
  \begin{equation*}
    H(p_{X}) - H(p_{X \mid b}) \approx 0.67 - 1 = -0.33
  \end{equation*}
  implying that the attacker's uncertainty has \emph{increased} after observing that particular output,
  even though the adversary has effectively learned the first bit of the secret.
\end{ExampleHead}

Example~\ref{ex:motivation:shannon} 
can be captured by the \emph{Belief Tracking} framework of \CiteAuthor{Belief:Clarkson2009}~\cite{Belief:Clarkson2009}.
Although for deterministic systems that approach does not depend on the concrete value of the secret~\cite[Thm 1]{DynLeak:Bielova2016},
in general (for probabilistic systems) the framework requires a concrete secret value,
and it measures uncertainty with the KL divergence.
Our approach allows for more general adversarial models via the $g$-leakage framework.

\subsection{The problem of DPI: The traditional definition of dynamic
  leakage is not well-suited for comparing systems}%
\label{sec:motivation:dpi}

It is often the case that we want to compare two systems $\Chan{C}$ and $\Chan{D}$ in terms of their information leakage.
This is especially useful when $\Chan{C}$ and $\Chan{D}$ are related 
---~for example, when $\Chan{D}$ is exactly the program $\Chan{C}$ with an additional code snippet,
which can be modelled as a post-processing step.
One might include a post-processing step to increase security/privacy;
this intuition relies on the \emph{data-processing inequality} (DPI),
which says that post-processing cannot \emph{increase} information leakage,
and therefore $\Chan{D}$ must be \emph{at least as safe as} $\Chan{C}$. 
In QIF we call $\Chan{D}$ a \emph{refinement} of $\Chan{C}$, written $\Chan{C} \sqsubseteq \Chan{D}$,
whenever $\Chan{D}$ is a post-processing of $\Chan{C}$.\footnote{%
  The definition is a little more nuanced than this,
  but these notions have been shown to be equivalent.%
}
It has been shown that the static perspective of leakage preserves the DPI, and thus refinement.

One might wonder whether this notion of refinement, 
i.e., of being able to safely replace a system $\Chan{C}$ with $\Chan{D}$ 
if system $\Chan{D}$ can be proven to be a post-processing of $\Chan{C}$,
also holds under dynamic scenarios.
Unfortunately, by employing current definitions of dynamic information leakage, this fails in many ways.
To illustrate that,
Example~\ref{ex:motivation:query} below shows a case in which the results suggest that such a replacement could be performed,
but in practice we can verify that it is not necessarily safe.
A second example will be detailed in Section~\ref{sec:case-studies}.
  
\begin{ExampleHead}{ex:motivation:query}
  Consider a company that collects health data and provides public access to it via queries,
  such as if there is a case of a particular disease in the database.
  Due to the sensitive nature of the data, to ensure individuals' privacy,
  the data engineering team develops the following program $\Chan{P}$ which produces as output a value $y \in \OrdSet{\textsc{no}, \textsc{yes}}$,
  where ($a\, \IChoice[p]\, b$) returns $a$ with probability $p$ and $b$ with $1 - p$:
  \begin{center}
    if answer = \textsc{no} 
    then ($\textsc{no} \IChoice[\sfrac{2}{3}] \textsc{yes}$)
    else ($\textsc{no} \IChoice[\sfrac{1}{3}] \textsc{yes}$)
  \end{center}
  That is, this program reports the true answer with \sfrac{2}{3} probability and the wrong answer with \sfrac{1}{3} probability.

  To test $\Chan{P}$,
  the data engineering team runs a query from the perspective of an attacker $\TrueAdv$ 
  and measures the chance of correctly inferring the original answer.
  Assume that the adversary starts with no prior knowledge,
  modelled as a uniform distribution $p_{X} = \OrdSet{\sfrac{1}{2}, \sfrac{1}{2}}$.
  Then the value \textsc{yes} is observed by $\TrueAdv$,
  who uses Bayes' rule to obtain the posterior $p_{X \mid \textsc{yes}}^{\TrueAdv} = \OrdSet{ \sfrac{1}{3}, \sfrac{2}{3}}.$
  Thus, $\TrueAdv$'s best action is to guess \textsc{yes}.

  To measure the secret's vulnerability with respect to $\TrueAdv$'s action,
  we use the Bayes vulnerability, which chooses the secret with the highest probability,
  and corresponds to the adversary's probability of guessing the secret in one try:
  \begin{equation}\label{eq:motivation:bayes-vuln}
    \Vg[\GainBayes](p) \EqDef \max_{x \in \CalX} p_{x}
  \end{equation}
  Hence, the (dynamic) vulnerability is $\Vg[\GainBayes](p_{X \mid \textsc{yes}}^{\TrueAdv}) = \sfrac{2}{3}$.

  Now, let us say that the engineering team considers this vulnerability to be too large,
  and decides to add a second sanitisation step $\Chan{S}$ to the pipeline,
  implemented as follows:
  \begin{center}
    if (perturbed) answer = \textsc{no} then (\textsc{no}) else
    ($\textsc{no} \IChoice[\sfrac{1}{2}] \textsc{yes}$)
  \end{center}
  That is, $\Chan{S}$ always passes on input \textsc{no} accurately,
  whereas input \textsc{yes} has $\sfrac{1}{2}$ probability of being flipped.

  \begin{figure}[tb]
   \centering
   \begin{tikzpicture}

  \node[] (input) at (0, 0) {\textsc{no}};
  \node[draw, rounded corners, right = 2em of input, inner sep = 5pt] (P) {$\Chan{P}$};
  \node[draw, rounded corners, right = 4em of P, inner sep = 5pt] (S) {$\Chan{S}$};
  \node[right = 4em of S] (output) {};

  \draw[-Stealth, shorten > = 2pt] (input) to (P);
  
  \draw[-Stealth, shorten < = 2pt, shorten > = 2pt]
  ($(P.east) + (0, .15)$) to node[above] {\textsc{no} $\sfrac{2}{3}$} ($(S.west) + (0, .15)$);
  
  \draw[-Stealth, dashed, shorten < = 2pt, shorten > = 2pt, HighlightColor]
  ($(P.east) - (0, .15)$) to node[below] {\textsc{yes} $\sfrac{1}{3}$} ($(S.west) - (0, .15)$);
  
  \draw[-Stealth, shorten < = 2pt, shorten > = 2pt]
  ($(S.east) + (0, .15)$) to node[above] {\textsc{no} $1$} ($(output.west) + (0, .15)$);
  
  \draw[-Stealth, dashed, shorten < = 2pt, shorten > = 2pt, HighlightColor]
  ($(S.east) - (0, .15)$) to node[below] {\textsc{no} $\sfrac{1}{2}$} ($(output.west) - (0, .15)$);
\end{tikzpicture}
   \caption{Pipeline $\Chan{P}\Cascade\Chan{S}$ for
     Example~\ref{ex:motivation:query},
     considering as input the answer \textsc{no} and as final output the answer \textsc{no}. 
     The outgoing black arrows correspond to the case where the mechanism preserved the answer,
     whereas the outgoing (blue) dashed arrows indicate that the answer was flipped. 
     The probability of an output \textsc{no} when its input is \textsc{no} is
     $\sfrac{2}{3} \cdot 1 + \sfrac{1}{3} \cdot \sfrac{1}{2} =
     \sfrac{5}{6}$. \label{fig:query-pipeline}}
  \end{figure}
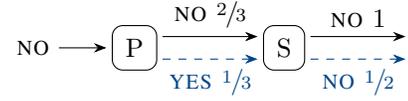

  The new post-processing mechanism $\Chan{S}$ is then applied to the output produced by $\Chan{P}$.
  In our example, 
  let us say that the output \textsc{yes} produced by $\Chan{P}$ is fed as input to $\Chan{S}$,
  which produces an output \textsc{no},
  which is observed by a second adversary $\NoiseAdv$ (see Figure~\ref{fig:query-pipeline} for an illustration).
  We can compute $\NoiseAdv$'s posterior knowledge via Bayes' rule as $p_{X \mid \textsc{no}}^{\NoiseAdv} = \OrdSet{\sfrac{5}{9}, \sfrac{4}{9}}.$

  Therefore, $\NoiseAdv$'s best action is to guess \textsc{no},
  with a vulnerability of $\Vg[\GainBayes](p_{X \mid \textsc{no}}^{\NoiseAdv}) = \sfrac{5}{9}$,
  which is \emph{smaller} than $\TrueAdv$'s chance of correctly guessing the secret.
  So, by employing the traditional definition of dynamic vulnerability,
  it seems that the engineering team has achieved their goal:
  to come up with a system (the composition $\Chan{P}\Cascade\Chan{S}$) that is \emph{always} safer than $\Chan{P}$.
 
  However, a more careful analysis shows that this conclusion can be misleading.
  Suppose that the actual input to $\Chan{P}$ was \textsc{no}.
  This is the real secret, which is known to the data engineering team.
  Notice that $\TrueAdv$'s action of guessing \textsc{yes} would not match the secret value in this scenario,
  whereas $\NoiseAdv$'s action of guessing \textsc{no} would.
  Thus, \emph{in reality} $\NoiseAdv$ would guess correctly while $\TrueAdv$ guesses incorrectly,
  even though we computed that $\NoiseAdv$'s chance of success was smaller than $\TrueAdv$'s chance.
\end{ExampleHead}

The example above shows that the traditional definition of information leakage can lead to erroneous conclusions under dynamic scenarios.
Our goal is then to provide a new definition that captures dynamic scenarios correctly, 
and also to identify when the data-processing inequality is satisfied.

\subsection{%
    A brief explanation on why the traditional
    definition of dynamic leakage leads to unexpected results%
}%
\label{sec:explain-traditional}

Our key insight is that, in the traditional definition of leakage,
the adversary's knowledge about the secret plays two distinct roles at the same time:
it represents the adversary's \emph{belief} about the secret value 
---~which is used to guide their actions~--- 
but it is also the \emph{baseline} distribution on secrets 
---~against which the success of the adversary's actions are measured.
This is a problem in the dynamic perspective because,
once a concrete execution of the system is observed,
the prior distribution on secrets is updated to a posterior distribution
that better captures what values the secret could have taken \emph{in that particular run}.
Therefore, 
although the adversary's prior distribution remains a good reflection of their \emph{belief} 
about the secret values before the system was run,
it is not necessarily an accurate reflection of the \emph{baseline} 
against which their success should be measured \emph{in that particular run}.

As such, we argue that we should divide the states of knowledge between baseline and belief. 
As we shall formalise in this paper, the absence of such a distinction
did not pose a problem under the static perspective because the expectation of all posteriors 
---~which act as both the adversary's posterior beliefs and also the baselines~---
is indeed always equal to the prior (belief) and the effect is, therefore, not detectable.

\begin{figure}[tb]
  \centering
  \begin{subfigure}{\columnwidth}
    \begin{tikzpicture}
  \node[] (P/input) at (0, 0) {$\Dist\mathcal{X}$};
  \node[draw, rounded corners, inner sep = 6pt, right = 2em of P/input]
  (P/baseline) {$\Chan{B}$};

  \draw[-Stealth, shorten < = 2pt, shorten > = 2pt]
  ($(P/input.north east |- P/baseline.north west) + (0, -1ex)$)
  to ($(P/baseline.north west) + (0, -1ex)$);
  
  \draw[-Stealth, shorten < = 2pt, shorten > = 2pt]
  (P/input) to (P/baseline);

  \draw[-Stealth, shorten < = 2pt, shorten > = 2pt]
  ($(P/input.south east |- P/baseline.south west) + (0, +1ex)$)
  to ($(P/baseline.south west) + (0, +1ex)$);

  \node[draw, rounded corners, inner sep = 6pt, right = 5em of P/baseline]
  (P/P) {$\Chan{P}$};

  \draw[-Stealth, shorten < = 2pt, shorten > = 2pt]
  ($(P/baseline.north east |- P/P.north west) + (0, -1ex)$)
  to node[midway, fill=white] {\phantom{$\Dist\mathcal{O}$}}
  ($(P/P.north west) + (0, -1ex)$);
  
  \draw[-Stealth, shorten < = 2pt, shorten > = 2pt]
  ($(P/baseline.south east |- P/P.south west) + (0, +1ex)$)
  to node[midway, fill=white] {\phantom{$\Dist\mathcal{O}$}}
  ($(P/P.south west) + (0, +1ex)$);
  
  \draw[-Stealth, shorten < = 2pt, shorten > = 2pt]
  (P/baseline) to node[midway, fill=white] {$\Dist\mathcal{O}$} (P/P);

  \node[inner sep = 6pt, right = 5em of P/P] (P/output)
  {$\phantom{\Chan{P}}$};
  
  \draw[-Stealth, shorten < = 2pt, shorten > = 2pt]
  ($(P/P.north east |- P/output.north west) + (0, -1ex)$)
  to node[midway, fill=white] {\phantom{$\Dist\mathcal{Y}$}}
  ($(P/output.north west) + (0, -1ex)$);
  
  \draw[-Stealth, shorten < = 2pt, shorten > = 2pt]
  ($(P/P.south east |- P/output.south west) + (0, +1ex)$)
  to node[midway, fill=white] {\phantom{$\Dist\mathcal{Y}$}}
  ($(P/output.south west) + (0, +1ex)$);
  
  \draw[-Stealth, shorten < = 2pt, shorten > = 2pt]
  (P/P) to node[midway, fill=white] {$\Dist\mathcal{Y}$} (P/output);

  \node[below = 3ex of P/input] (S/input) {$\Dist\mathcal{X}$};
  \node[draw, rounded corners, inner sep = 6pt, right = 2em of S/input]
  (S/baseline) {$\Chan{B}$};

  \draw[-Stealth, shorten < = 2pt, shorten > = 2pt]
  ($(S/input.north east |- S/baseline.north west) + (0, -1ex)$)
  to ($(S/baseline.north west) + (0, -1ex)$);
  
  \draw[-Stealth, shorten < = 2pt, shorten > = 2pt]
  (S/input) to (S/baseline);

  \draw[-Stealth, shorten < = 2pt, shorten > = 2pt]
  ($(S/input.south east |- S/baseline.south west) + (0, +1ex)$)
  to ($(S/baseline.south west) + (0, +1ex)$);

  \node[draw, rounded corners, inner sep = 6pt, right = 5em of S/baseline]
  (S/P) {$\Chan{P}$};

  \draw[-Stealth, shorten < = 2pt, shorten > = 2pt]
  ($(S/baseline.north east |- S/P.north west) + (0, -1ex)$)
  to node[midway, fill=white] {\phantom{$\Dist\mathcal{O}$}}
  ($(S/P.north west) + (0, -1ex)$);
  
  \draw[-Stealth, shorten < = 2pt, shorten > = 2pt]
  ($(S/baseline.south east |- S/P.south west) + (0, +1ex)$)
  to node[midway, fill=white] {\phantom{$\Dist\mathcal{O}$}}
  ($(S/P.south west) + (0, +1ex)$);
  
  \draw[-Stealth, shorten < = 2pt, shorten > = 2pt]
  (S/baseline) to node[midway, fill=white] {$\Dist\mathcal{O}$} (S/P);

  \node[draw, rounded corners, inner sep = 6pt, right = 5em of S/P]
  (S/S) {$\Chan{S}$};
  
  \draw[-Stealth, shorten < = 2pt, shorten > = 2pt]
  ($(S/P.north east |- S/S.north west) + (0, -1ex)$)
  to node[midway, fill=white] {\phantom{$\Dist\mathcal{Y}$}}
  ($(S/S.north west) + (0, -1ex)$);
  
  \draw[-Stealth, shorten < = 2pt, shorten > = 2pt]
  ($(S/P.south east |- S/S.south west) + (0, +1ex)$)
  to node[midway, fill=white] {\phantom{$\Dist\mathcal{Y}$}}
  ($(S/S.south west) + (0, +1ex)$);
  
  \draw[-Stealth, shorten < = 2pt, shorten > = 2pt]
  (S/P) to node[midway, fill=white] {$\Dist\mathcal{Y}$} (S/S);
  
  \node[inner sep = 6pt, right = 5em of S/S] (S/output)
  {$\phantom{\Chan{P}}$};
  
  \draw[-Stealth, shorten < = 2pt, shorten > = 2pt]
  ($(S/S.north east |- S/output.north west) + (0, -1ex)$)
  to node[midway, fill=white] {\phantom{$\Dist\mathcal{Z}$}}
  ($(S/output.north west) + (0, -1ex)$);
  
  \draw[-Stealth, shorten < = 2pt, shorten > = 2pt]
  ($(S/S.south east |- S/output.south west) + (0, +1ex)$)
  to node[midway, fill=white] {\phantom{$\Dist\mathcal{Z}$}}
  ($(S/output.south west) + (0, +1ex)$);
  
  \draw[-Stealth, shorten < = 2pt, shorten > = 2pt]
  (S/S) to node[midway, fill=white] {$\Dist\mathcal{Z}$} (S/output);

  \draw[dotted, color = HighlightColor, very thick]
  ($(P/output.north west) + (-1em, +1ex)$) to ($(S/S.south west) + (-1em, -1ex)$);
\end{tikzpicture}%
    \vspace{-3ex}%
    \caption{Static perspective.}%
    \label{fig:static-diagram}%
    \hspace*{\fill}%
  \end{subfigure}

  \begin{subfigure}{\columnwidth}
    \begin{tikzpicture}
  \node[] (P/input) at (0, 0) {$\Dist\mathcal{X}$};
  \node[draw, rounded corners, inner sep = 6pt, right = 2em of P/input]
  (P/baseline) {$\Chan{B}$};

  \draw[-Stealth, shorten < = 2pt, shorten > = 2pt]
  ($(P/input.north east |- P/baseline.north west) + (0, -1ex)$)
  to ($(P/baseline.north west) + (0, -1ex)$);
  
  \draw[-Stealth, shorten < = 2pt, shorten > = 2pt]
  (P/input) to (P/baseline);

  \draw[-Stealth, shorten < = 2pt, shorten > = 2pt]
  ($(P/input.south east |- P/baseline.south west) + (0, +1ex)$)
  to ($(P/baseline.south west) + (0, +1ex)$);

  \node[draw, rounded corners, inner sep = 6pt, right = 5em of P/baseline]
  (P/P) {$\Chan{P}$};

  \draw[-Stealth, shorten < = 2pt, shorten > = 2pt]
  ($(P/baseline.north east |- P/P.north west) + (0, -1ex)$)
  to node[midway, fill=white] {\phantom{$\Dist\mathcal{O}$}}
  ($(P/P.north west) + (0, -1ex)$);
  
  \draw[-Stealth, shorten < = 2pt, shorten > = 2pt]
  ($(P/baseline.south east |- P/P.south west) + (0, +1ex)$)
  to node[midway, fill=white] {\phantom{$\Dist\mathcal{O}$}}
  ($(P/P.south west) + (0, +1ex)$);
  
  \draw[-Stealth, shorten < = 2pt, shorten > = 2pt]
  (P/baseline) to node[midway, fill=white] {$\Dist\mathcal{O}$} (P/P);

  \node[inner sep = 6pt, right = 5em of P/P] (P/output)
  {$\phantom{\Chan{P}}$};
  
  \draw[-Stealth, shorten < = 2pt, shorten > = 2pt]
  (P/P) to node[midway, fill=white] {\textsc{yes}} (P/output);

  \node[below = 3ex of P/input] (S/input) {$\Dist\mathcal{X}$};
  \node[draw, rounded corners, inner sep = 6pt, right = 2em of S/input]
  (S/baseline) {$\Chan{B}$};

  \draw[-Stealth, shorten < = 2pt, shorten > = 2pt]
  ($(S/input.north east |- S/baseline.north west) + (0, -1ex)$)
  to ($(S/baseline.north west) + (0, -1ex)$);
  
  \draw[-Stealth, shorten < = 2pt, shorten > = 2pt]
  (S/input) to (S/baseline);

  \draw[-Stealth, shorten < = 2pt, shorten > = 2pt]
  ($(S/input.south east |- S/baseline.south west) + (0, +1ex)$)
  to ($(S/baseline.south west) + (0, +1ex)$);

  \node[draw, rounded corners, inner sep = 6pt, right = 5em of S/baseline]
  (S/P) {$\Chan{P}$};

  \draw[-Stealth, shorten < = 2pt, shorten > = 2pt]
  ($(S/baseline.north east |- S/P.north west) + (0, -1ex)$)
  to node[midway, fill=white] {\phantom{$\Dist\mathcal{O}$}}
  ($(S/P.north west) + (0, -1ex)$);
  
  \draw[-Stealth, shorten < = 2pt, shorten > = 2pt]
  ($(S/baseline.south east |- S/P.south west) + (0, +1ex)$)
  to node[midway, fill=white] {\phantom{$\Dist\mathcal{O}$}}
  ($(S/P.south west) + (0, +1ex)$);
  
  \draw[-Stealth, shorten < = 2pt, shorten > = 2pt]
  (S/baseline) to node[midway, fill=white] {$\Dist\mathcal{O}$} (S/P);

  \node[draw, rounded corners, inner sep = 6pt, right = 5em of S/P]
  (S/S) {$\Chan{S}$};
  
  \draw[-Stealth, shorten < = 2pt, shorten > = 2pt]
  ($(S/P.north east |- S/S.north west) + (0, -1ex)$)
  to node[midway, fill=white] {\phantom{$\Dist\mathcal{Y}$}}
  ($(S/S.north west) + (0, -1ex)$);
  
  \draw[-Stealth, shorten < = 2pt, shorten > = 2pt]
  ($(S/P.south east |- S/S.south west) + (0, +1ex)$)
  to node[midway, fill=white] {\phantom{$\Dist\mathcal{Y}$}}
  ($(S/S.south west) + (0, +1ex)$);
  
  \draw[-Stealth, shorten < = 2pt, shorten > = 2pt]
    (S/P) to node[midway, fill=white] {$\color{AlertColor}\Dist\mathcal{Y}$} (S/S);
  
  \node[inner sep = 6pt, right = 5em of S/S] (S/output)
  {$\phantom{\Chan{P}}$};

  \draw[-Stealth, shorten < = 2pt, shorten > = 2pt]
  (S/S) to node[midway, fill=white] {\textsc{no}} (S/output);

  \draw[dotted, color = HighlightColor, very thick]
  ($(P/output.north west) + (-1em, +1ex)$) to ($(S/S.south west) + (-1em, -1ex)$);
\end{tikzpicture}%
    \vspace{-3ex}%
    \caption{%
        Traditional dynamic perspective, when the output observed by adversary $\TrueAdv$ is \textsc{yes} 
        and the output observed by $\NoiseAdv$ is \textsc{no}.%
    }%
    \label{fig:traditional-dynamic-diagram}%
    \hspace*{\fill}%
  \end{subfigure}

  \begin{subfigure}{\columnwidth}
    \begin{tikzpicture}
  \node[] (P/input) at (0, 0) {$\Dist\mathcal{X}$};
  \node[draw, rounded corners, inner sep = 6pt, right = 2em of P/input]
  (P/baseline) {$\Chan{B}$};

  \draw[-Stealth, shorten < = 2pt, shorten > = 2pt]
  ($(P/input.north east |- P/baseline.north west) + (0, -1ex)$)
  to ($(P/baseline.north west) + (0, -1ex)$);
  
  \draw[-Stealth, shorten < = 2pt, shorten > = 2pt]
  (P/input) to (P/baseline);

  \draw[-Stealth, shorten < = 2pt, shorten > = 2pt]
  ($(P/input.south east |- P/baseline.south west) + (0, +1ex)$)
  to ($(P/baseline.south west) + (0, +1ex)$);

  \node[draw, rounded corners, inner sep = 6pt, right = 5em of P/baseline]
  (P/P) {$\Chan{P}$};

%
  
  \draw[-Stealth, shorten < = 2pt, shorten > = 2pt]
    (P/baseline) to node[midway, fill=white] {\textsc{no}} (P/P);

  \node[inner sep = 6pt, right = 5em of P/P] (P/output)
  {$\phantom{\Chan{P}}$};
  
  \draw[-Stealth, shorten < = 2pt, shorten > = 2pt]
  (P/P) to node[midway, fill=white] {\textsc{yes}} (P/output);

  \node[below = 3ex of P/input] (S/input) {$\Dist\mathcal{X}$};
  \node[draw, rounded corners, inner sep = 6pt, right = 2em of S/input]
  (S/baseline) {$\Chan{B}$};

  \draw[-Stealth, shorten < = 2pt, shorten > = 2pt]
  ($(S/input.north east |- S/baseline.north west) + (0, -1ex)$)
  to ($(S/baseline.north west) + (0, -1ex)$);
  
  \draw[-Stealth, shorten < = 2pt, shorten > = 2pt]
  (S/input) to (S/baseline);

  \draw[-Stealth, shorten < = 2pt, shorten > = 2pt]
  ($(S/input.south east |- S/baseline.south west) + (0, +1ex)$)
  to ($(S/baseline.south west) + (0, +1ex)$);

  \node[draw, rounded corners, inner sep = 6pt, right = 5em of S/baseline]
  (S/P) {$\Chan{P}$};

  %
  
  \draw[-Stealth, shorten < = 2pt, shorten > = 2pt]
    (S/baseline) to node[midway, fill=white] {\textsc{no}} (S/P);

  \node[draw, rounded corners, inner sep = 6pt, right = 5em of S/P]
  (S/S) {$\Chan{S}$};
  
  %
  
  \draw[-Stealth, shorten < = 2pt, shorten > = 2pt]
  (S/P) to node[midway, fill=white] {\textsc{yes}} (S/S);
  
  \node[inner sep = 6pt, right = 5em of S/S] (S/output)
  {$\phantom{\Chan{P}}$};

  \draw[-Stealth, shorten < = 2pt, shorten > = 2pt]
  (S/S) to node[midway, fill=white] {\textsc{no}} (S/output);

  \draw[dotted, color = HighlightColor, very thick]
  ($(P/output.north west) + (-1em, +1ex)$) to ($(S/S.south west) + (-1em, -1ex)$);
\end{tikzpicture}%
    \vspace{-3ex}%
    \caption{%
        (Multi-step) dynamic perspective, from the point of view of the data engineering team, 
        when the real, secret value is \textsc{no}.%
    }%
    \label{fig:dynamic-diagram-ex2}%
    \hspace*{\fill}%
  \end{subfigure}
  \caption{%
      The semantics of analyses related to adversaries 
      $\TrueAdv$ (upper diagram in each figure) and
      $\NoiseAdv$ (lower diagram) from Example~\ref{ex:motivation:query}.
      In each figure, the vertical dotted line delimits the section of the pipeline that is common to both adversaries.
      Multiple arrows indicate a ``static'' input  (respectively, output), 
      meaning that the system takes as input (produces as output) a probability distribution over all possible values.
      A single arrow indicates a concrete, single execution of the system.%
  }%
  \label{fig:example-2}
\end{figure}
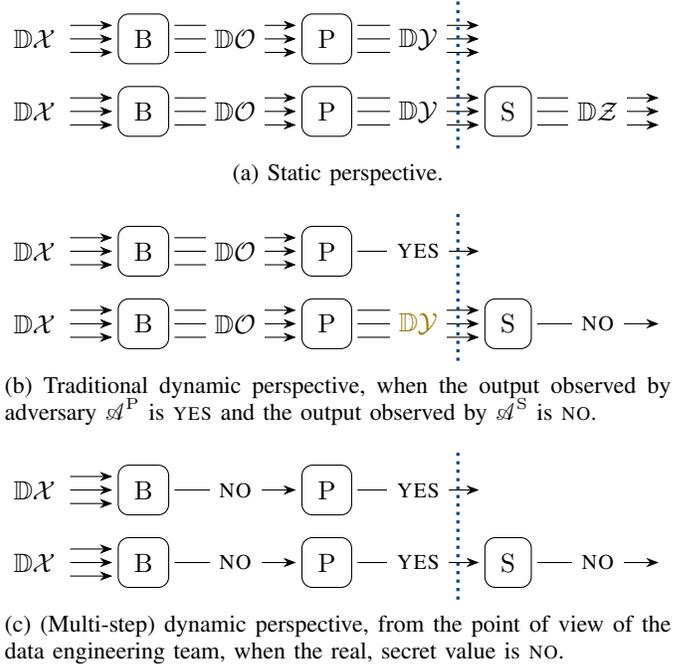
  
To explore in more detail why the traditional definition of dynamic leakage is not always suitable for comparing systems,
let us contrast it to the static perspective. 
Figures~\ref{fig:static-diagram}~and~\ref{fig:traditional-dynamic-diagram} 
illustrate what the semantics of, respectively, the static and (traditional) dynamic analysis would be
with respect to the adversaries from Example~\ref{ex:motivation:query}.
Recall that adversary $\TrueAdv$ observes the output of the system after the first post-processing step $\Chan{P}$ 
(represented in the upper diagram in each figure),
while adversary $\NoiseAdv$ observes the output of the system after the second post-processing step $\Chan{S}$ (lower diagram).
The first component of the system in both diagrams is $\Chan{B}$,
taking as input a secret value from a set $\mathcal{X}$ and producing as output a value from a set $\mathcal{O}$,
which could, in principle, have been observed by the data engineering team.
In this case the component $\Chan{B}$ is simply the program that reports the original, unaltered answer.

Notice that, in the static perspective (Figure~\ref{fig:static-diagram}),
the section from $\Dist\mathcal{X}$ up to $\Dist\mathcal{Y}$ in the lower diagram (adversary $\NoiseAdv$)
exactly matches the upper diagram (adversary $\TrueAdv$).
That is, the semantics of both analyses up to the output of system $\Chan{P}$ 
---~the section common to both analyses~---
takes into account every possible execution path.
This indicates that the results of the two analyses in the static perspective are comparable.

In contrast, in the (traditional) dynamic perspective (Figure~\ref{fig:traditional-dynamic-diagram}),
the section from $\Dist\mathcal{X}$ up to $\Dist\mathcal{Y}$ in the lower diagram ($\NoiseAdv$)
does \emph{not} match with the upper diagram ($\TrueAdv$).
That is, 
while the analysis of adversary $\TrueAdv$ (upper diagram) takes into account a single output \textsc{yes} from system $\Chan{P}$,
the analysis of $\NoiseAdv$ (lower diagram) considers both outputs \textsc{yes} and \textsc{no} from $\Chan{P}$.
Therefore, we argue that these results are \emph{not} comparable.

The intuition in this case is that 
the traditional definition of dynamic vulnerability is parameterised by only one explicit state of knowledge:
the adversary's belief after observing the final output.
However, implicitly, this definition evaluates the effect of the actions taken by the adversary,
derived from their belief, on every possible baseline
---~i.e., every possible state of knowledge 
that could have been obtained in any intermediate step prior to the output observed by the adversary~--- 
as an expectation.
The problem with this is then that 
the set of states of knowledge considered implicitly in the analysis of $\NoiseAdv$
may include states of knowledge that were not considered in the analysis of $\TrueAdv$.
Indeed, this is what happens in Example~\ref{ex:motivation:query},
when considering the traditional definition of dynamic vulnerability:
the result obtained for adversary $\NoiseAdv$ 
considers the intermediate state of knowledge obtained given output \textsc{no} from $\Chan{P}$ as one of the possible baselines,
whereas this state of knowledge is not considered in the result of $\TrueAdv$.

The reason why no issues arise in the static perspective (Figure~\ref{fig:static-diagram})
is that the semantics of the static perspective not only takes into account (implicitly) every possible baseline,
it also considers (explicitly) every possible belief that the adversary could construct.
Consequently, 
the set of states of knowledge considered in the analysis without component $\Chan{S}$ (upper diagram)
is exactly the same used in the analysis that includes $\Chan{S}$ (lower diagram),
up to the output of component $\Chan{P}$.%

A similar argument follows for Example~\ref{ex:motivation:shannon},
in which we obtained negative leakage.
In that case we would model $\Chan{P}$ as the deterministic program mapping the 2-bit secret to either $a$ or $b$,
and represent the adversary's prior knowledge as a post-processing step $\ChanLeaksNothing$ 
that maps every output from $\Chan{P}$ to the same value $\bot$ that does not reveal anything.
Then, there is only one posterior knowledge obtained from $\ChanLeaksNothing$, which is exactly the prior knowledge.
By doing this, we can see that the set of states of knowledge considered in the analysis a priori 
includes states of knowledge not considered in the analysis a posteriori.

Finally, we concluded Example~\ref{ex:motivation:query} 
by noting that there is at least one dynamic path $\textsc{no} \xrightarrow{\Chan{P}} \textsc{yes} \xrightarrow{\Chan{S}} \textsc{no}$
that shows that it is not necessarily safe to compose systems $\Chan{P}\Cascade\Chan{S}$. 
This comparison is represented by the diagrams in Figure~\ref{fig:dynamic-diagram-ex2}.
Notice that this time 
the prefix of the lower diagram ($\NoiseAdv$) matches with the upper diagram ($\TrueAdv$),
and so the results can be compared.%

\section{Preliminaries}\label{sec:prelim}

In this paper we employ the mathematical framework of \emph{quantitative information flow} (QIF)~\cite{QIF-Book:Alvim2020}
to model (probabilistic) systems and characterise Bayesian adversarial attacks against them.
In the rest of this section we formalise systems,
adversaries with different goals and capabilities,
and vulnerability/uncertainty measures under both the static and the dynamic perspectives,
using the framework of QIF.

\subsection{Modelling systems and knowledge}
In QIF a system is modelled as an \emph{information-theoretic channel} $\Chan{C} : \CalX \to \Dist\CalY$ 
that maps a (secret) input $x \in \CalX$ to a (public) output $y \in \CalY$ according to some distribution in $\Dist\CalY$. 
(Given a set $\CalY$, we denote by $\Dist\CalY$ the set of all probability distributions over $\CalY$.)
We usually assume that $\CalX$ and $\CalY$ are discrete,
in which case we write $\Chan{C}$ as a stochastic matrix with rows labeled by the possible inputs
and columns labeled by the possible outputs,
so that $\Chan{C}_{x, y}$ is the probability of channel $\Chan{C}$ outputing value $y \in \CalY$ when its input is $x \in \CalX$.

By combining a \emph{prior} knowledge $\Prior{\pi} : \Dist\CalX$ 
with the channel $\Chan{C}$ that models the system of interest,
one can compute a joint distribution $\Joint{\Prior{\pi}}{\Chan{C}} : \Dist(\CalX \times \CalY)$ as
\begin{equation}\label{eq:joint}
  \Joint{\Prior{\pi}}{\Chan{C}}[x, y] \EqDef \Prior{\pi}_{x} \Chan{C}_{x, y}
  \quad\forall x \in \CalX \textnormal{ and } y \in \CalY
\end{equation}
Then, from the joint distribution,
one can derive a marginal distribution $\Outer{\Prior{\pi}}{\Chan{C}} : \Dist\mathcal{Y}$ on the possible observations $y \in \CalY$,
which we call the \emph{outer} distribution, as
\begin{equation}\label{eq:outer}
  \Outer{\Prior{\pi}}{\Chan{C}}[y] \EqDef
  \sum_{x \in \CalX} \Joint{\Prior{\pi}}{\Chan{C}}[x, y]
  \quad\forall y \in \CalY
\end{equation}
From the joint and outer distributions,
one can construct posterior distributions $\Posterior{\Prior{\pi}}{\Chan{C}}[X][y] : \Dist\mathcal{X}$ on the values $\CalX$ 
given each $y \in \CalY$, 
which we also call the \emph{inner} distributions, as
\begin{equation}\label{eq:posterior}
  \Posterior{\Prior{\pi}}{\Chan{C}}[x][y] \EqDef
  \frac
  {\Joint{\Prior{\pi}}{\Chan{C}}[x, y]}
  {\Outer{\Prior{\pi}}{\Chan{C}}[y]}
  \quad\forall x \in \CalX
\end{equation}
Finally, when considering a static perspective,
we write $\Hyper{\Prior{\pi}}{\Chan{C}}$ to denote a \emph{hyper-distribution},
a distribution on distributions on $\CalX$, (customarily) abstracting away the actual output labels.

Channels can be composed in many ways. 
Relevant to this paper is the sequential composition of two channels $\Chan{C}$ and $\Chan{D}$,
termed \emph{cascade} and written $\Chan{C}\Cascade\Chan{D}$.
The cascade of a channel $\Chan{C}$ followed by a channel $\Chan{D}$ is defined as follows:

\begin{definition}[{Cascade~\cite[Definition 4.18]{QIF-Book:Alvim2020}}]%
  \label{def:cascade}
  Given two channel (matrices) $\Chan{C} : \CalX \to \Dist\CalY$ and $\Chan{D} : \CalY \to \Dist\CalZ$,
  the cascade $\Chan{C}\Cascade\Chan{D}$ of channels $\Chan{C}$ and $\Chan{D}$ is the channel matrix of type $\CalX \to \Dist\CalZ$,
  obtained by ordinary matrix multiplication.
\end{definition}

\subsection{Modelling adversarial goals}

Adversaries can have different goals:
they may be interested in guessing the exact secret,
guessing the secret within a certain number of tries,
they may deem some secret values more valuable than others, etc.
Each goal can be captured in the $g$-leakage framework~\cite{g-leakage:Alvim2012} as a suitable gain function $\Gain : \CalW \times \CalX \to \Reals$,
which measures the benefit to the adversary when they take an \emph{action} $w \in \CalW$ and the secret value is $x \in \CalX$.
For instance, an adversary who wants to guess the secret in one try is modelled by the identity gain function
\begin{equation}\label{eq:gain-bayes}
  \GainBayes(w, x) \EqDef 1 \textnormal{ if } w = x \textnormal{ else } 0
\end{equation}

As we shall see (Section~\ref{sec:vuln}),
gain functions are employed in the definition of vulnerability measures.
Dually, entropy measures can be defined in terms of loss functions of type $\ell : \CalW \times \CalX \to \Reals_{\geq 0} \cup \{\infty\}$, 
modelling the attacker's loss upon taking a particular action.  
For instance, Shannon entropy corresponds to an adversary who is trying to guess the \qm{true} distribution on $\CalX$ 
(that is, each action is a distribution).
Therefore, it can be modelled by the following loss function:
\begin{equation}\label{eq:loss-shannon}
  \LossShannon(w, x) \EqDef - \log_{2} w_{x}
\end{equation}
defining $-\log_{2} 0 = \infty$,\footnote{%
  For more examples, see \cite[\S3.2]{QIF-Book:Alvim2020}.%
}
noting that $\LossShannon$ corresponds to Shannon's definition of information content, measured in bits.

\begin{remark}\label{rmk:bound}
  We note that, albeit in~\cite{QIF-Book:Alvim2020} loss functions are defined disallowing $\infty$,
  we find it more useful to include $\infty$, to permit results on Shannon entropy.
  This is not a problem, as we still require uncertainty measures to be bounded.
\end{remark}

\subsection{Current definitions of vulnerability, uncertainty, and
  information leakage under static and dynamic perspectives}%
\label{sec:vuln}

The threat to a secret caused by an adversary is given by a \emph{$\Gain$-vulnerability measure $\Vg(p)$},
defined as the expected gain to the adversary by taking an optimal action:
\begin{equation}\label{eq:vg-dist}
  \Vg(p) \EqDef \max_{w \in \CalW} \sum_{x \in \CalX} p_{x}\, \Gain(w, x)
\end{equation}
replacing $\max$ with $\sup$ if $\CalW$ is infinite.
Dually, the adversary's uncertainty about the secret is given by an \emph{$\Loss$-uncertainty measure $\Ul(p)$},
parameterised by a loss function $\Loss$:
\begin{equation}\label{eq:ul-dist}
  \Ul(p) \EqDef \min_{w \in \CalW} \sum_{x \in \CalX} p_{x}\, \Loss(w, x)
\end{equation}
replacing $\min$ with $\inf$ when $\CalW$ is infinite.
For example, Shannon entropy can be obtained from \eqref{eq:loss-shannon} as
\begin{equation}\label{eqn:ul_shannon}
    \Ul[\LossShannon](p) = \inf_{w \in \CalW} \sum_{x \in \CalX} p_{x}\, (- \log_{2} w_{x})
\end{equation}
which is minimised exactly when $w = p$, and equals $H(p)$.

We then define the amount of information available to the adversary before and after the system is run,
and the amount of information leaked from the system's execution as follows:

\begin{definition}[Prior measurements]
  \label{def:prior-measures}
  Given a prior knowledge modelled by a probability distribution $\Prior{\pi} : \Dist\CalX$,
  the \emph{prior} $\Gain$-vulnerability is $\Vg(\Prior{\pi})$.
  Dually, the \emph{prior} $\Loss$-uncertainty is $\Ul(\Prior{\pi})$.
\end{definition}
  
\begin{definition}[Traditional dynamic posterior measurements]
  \label{def:post-dyn-measures}
  The \emph{dynamic posterior} $\Gain$-vulnerability 
  with respect to the adversary's posterior knowledge $\Posterior{\Prior{\pi}}{\Chan{C}}[X][y]$,
  constructed by observing an execution of the system $\Chan{C}$,
  is $\Vg(\Posterior{\Prior{\pi}}{\Chan{C}}[X][y])$.
  Dually, the \emph{dynamic posterior} $\Loss$-uncertainty is $\Ul(\Posterior{\Prior{\pi}}{\Chan{C}}[X][y])$.
\end{definition}

\begin{definition}[Traditional dynamic leakage]
  \label{def:dyn-leak}
  Let $\Prior{\pi}$ be the adversary's prior knowledge.
  The amount of information that leaks from the execution of a system $\Chan{C} : \mathcal{X} \to \Dist\mathcal{Y}$ 
  upon yielding an output $y \in \mathcal{Y}$ 
  is defined (additively) in terms of prior $\Gain$-vulnerability and posterior $\Gain$-vulnerability as
  \begin{equation*}
    \LeakDyn{\Posterior{\Prior{\pi}}{\Chan{C}}[X][y])}{\Prior{\pi}}
    \EqDef
    \Vg(\Posterior{\Prior{\pi}}{\Chan{C}}[X][y]) -
    \Vg(\Prior{\pi})
  \end{equation*}
  When considering the attacker's $\Loss$-uncertainty about the secret,
  we flip the roles of the prior and posterior measures. That is,
  \begin{equation*}
    \LeakDyn[\Loss]{\Posterior{\Prior{\pi}}{\Chan{C}}[X][y])}{\Prior{\pi}}
    \EqDef
    \Ul(\Prior{\pi}) -
    \Ul(\Posterior{\Prior{\pi}}{\Chan{C}}[X][y])
  \end{equation*}
\end{definition}

Now, under a static perspective, 
one might be interested in either the expected or worst case scenarios.
This is defined (overloading $\Vg$ and $\Ul$) as follows,
taking into account all possible outputs that could be produced from the system:

\begin{definition}[Static posterior measurements]
  \label{def:post-static-measures}
  Let $\Prior{\pi} : \Dist\mathcal{X}$ be the adversary's prior knowledge 
  and $\Chan{C} : \mathcal{X} \to \Dist\mathcal{Y}$ the channel modelling a system.
  The static posterior $\Gain$-vulnerability is defined in the expected- and max-case as, respectively,
  \begin{align*}
    \Vg\Hyper{\Prior{\pi}}{\Chan{C}}
    &\EqDef
      \sum_{y \in \CalY} \Outer{\Prior{\pi}}{\Chan{C}}[y]\,
      \Vg(\Posterior{\Prior{\pi}}{\Chan{C}}[X][y])
    \\
    \VgMax\Hyper{\Prior{\pi}}{\Chan{C}}
    &\EqDef
      \max_{y \in \CalY} \Vg(\Posterior{\Prior{\pi}}{\Chan{C}}[X][y])
  \end{align*}
  Dually, the static posterior $\Loss$-uncertainty is defined in the expected- and min-case as, respectively,
  \begin{align*}
    \Ul\Hyper{\Prior{\pi}}{\Chan{C}}
    &\EqDef
      \sum_{y \in \CalY} \Outer{\Prior{\pi}}{\Chan{C}}[y]\,
      \Ul(\Posterior{\Prior{\pi}}{\Chan{C}}[X][y])
    \\
    \UlMin\Hyper{\Prior{\pi}}{\Chan{C}}
    &\EqDef
      \min_{y \in \CalY} \Ul(\Posterior{\Prior{\pi}}{\Chan{C}}[X][y])
  \end{align*}
\end{definition}

\begin{definition}[Static leakage]
  \label{def:static-leak}
  The static $\Gain$-leakage is defined in the expected- and max-case as, respectively,
  \begin{align*}
    \Leak{\Chan{C}}{\Prior{\pi}}
    &\EqDef
    \Vg\Hyper{\Prior{\pi}}{\Chan{C}} - \Vg(\Prior{\pi})
    \\
    \LeakMax{\Chan{C}}{\Prior{\pi}}
    &\EqDef
      \VgMax\Hyper{\Prior{\pi}}{\Chan{C}} - \Vg(\Prior{\pi})
  \end{align*}
  When considering the attacker's $\Loss$-uncertainty about the secret,
  we flip the roles of the prior and posterior measures. That is,
  \begin{align*}
    \Leak[\Loss]{\Chan{C}}{\Prior{\pi}}
    &\EqDef \Ul(\Prior{\pi}) - \Ul\Hyper{\Prior{\pi}}{\Chan{C}} 
    \\
    \LeakMin[\Loss]{\Chan{C}}{\Prior{\pi}}
    &\EqDef \Ul(\Prior{\pi}) - \UlMin\Hyper{\Prior{\pi}}{\Chan{C}}
  \end{align*}
\end{definition}

Finally, we recall the notion of refinement, 
which tells us when one channel is at least as safe (i.e., leaks no more) than another,
and coincides with the cascading of the original channel and a post-processing channel (Definition~\ref{def:cascade}):

\begin{definition}[Refinement]
  \label{def:refinement}
  Given channels $\Chan{C} : \CalX \to \Dist\CalY$ and $\Chan{D} : \CalX \to \Dist\CalZ$,
  we say that $\Chan{C}$ is refined by $\Chan{D}$, written $\Chan{C} \RefinedBy \Chan{D}$,
  whenever $\Vg\Hyper{\Prior{\pi}}{\Chan{C}} \geq \Vg\Hyper{\Prior{\pi}}{\Chan{D}}$ 
  for all priors $\pi$ and gain functions $g$;
  equivalently, whenever $\Chan{D}$ can be written as a post-processing of $\Chan{C}$ 
  (i.e., as $\Chan{C}\Cascade\Chan{R}$ for some channel $\Chan{R}$).
\end{definition}

\section{Strategy-Based Formalisation of Leakage}\label{sec:formalisation}

The examples from Sections~\ref{sec:motivation:mono} and \ref{sec:motivation:dpi}
show that the traditional definition of dynamic leakage (Definition~\ref{def:dyn-leak}) fails in significant ways.
In particular,
it might result in negative leakage even when it is clear that information has flowed (Example~\ref{ex:motivation:shannon});
or it might lead to the conclusion that post-processing is safe,
even though there are reasonable circumstances where this does not hold (Example~\ref{ex:motivation:query}); 
it might also lead to the opposite conclusion: that post-processing is not safe,
even though a more careful analysis proves this to be incorrect
(we illustrate this case later in Section~\ref{sec:case-studies}).
 
As identified in Section~\ref{sec:motivation},
these correspond to the properties of monotonicity and DPI, respectively.
It turns out that, in general, monotonicity and DPI cannot be guaranteed when considering a dynamic perspective 
(as Example~\ref{ex:motivation:query} indicates),
but there are particular scenarios in which these axioms always hold.
With this in mind, 
our goal is to provide a proper formalisation of dynamic leakage that satisfies these fundamental properties when possible,
and that yields sound results when such properties are not guaranteed to hold. 
In addition,
we require that our formalisation recovers the static perspective when taking into account all possible observations.

In what follows we present a novel formalisation of leakage which allows for both dynamic and static leakage analysis.  
We provide intuition for our measure in Section~\ref{sec:st-knowledge1} 
and in Section~\ref{sec:st-leakage} we redefine $\Gain$-vulnerability,
$\Loss$-uncertainty and information leakage in terms of adversarial strategies guiding attackers' decisions.  
In Section~\ref{sec:motivation:contd} we revisit the motivating examples from Section~\ref{sec:motivation},
through the lens of our new formalisation.
In Section~\ref{sec:knowledge-ord} we discuss the limitations of our formalisation,
in that it requires a notion of ordering between states of knowledge.
In Section~\ref{sec:soundness-dyn} 
we demonstrate desirable properties of the strategy-based formulation that hold under a dynamic perspective.
Finally, in Section~\ref{sec:soundness-static}
we show that our formalisation coincides with the traditional definitions adopted in QIF when considering a static perspective.

\subsection{Defining leakage via actions: A thought experiment}\label{sec:st-knowledge1}

In Example \ref{ex:motivation:shannon} 
the dynamic leakage computed using the traditional formalisation of Definition~\ref{def:dyn-leak},
with the attacker modelled by the loss function $\LossShannon$~\eqref{eq:loss-shannon},
produces a negative value even though we know that information flow has occurred.
In that example the prior was $\OrdSet{\sfrac{7}{8}, \sfrac{1}{16}, \sfrac{1}{16}}$,
and the computed posteriors were $p_{X \mid a} = \OrdSet{1, 0, 0}$ and $p_{X \mid b} = \OrdSet{0, \sfrac{1}{2}, \sfrac{1}{2}}$, 
occurring with probability $\sfrac{7}{8}$ and $\sfrac{1}{8}$, respectively.
The first posterior (observation $a$) has all weight on secret $x_{0}$,
meaning that the adversary is completely certain,
and so it leads to a positive dynamic leakage of $0.67$. 
In contrast, the second posterior (observation $b$),
although less likely to be observed in practice,
yields a negative dynamic leakage value of $-0.33$.

Recall from Section~\ref{sec:explain-traditional} that the issue in this case is that
the set of states of knowledge considered in the analyses a priori and a posteriori are not compatible.
To resolve this issue, consider the following \textbf{thought experiment}:
imagine that there are two colluding adversaries:
Bob, whose belief about the secret is formed before any observation is made (a prior),
and Alice, whose belief results from a Bayesian update after observing a run of the system (a posterior),
both of whom have the same objective (a shared gain function). 
Bob can compute his maximum expected gain (with respect to his belief) using the shared gain function; 
this results in a set of actions which are optimal from his point of view.
Bob communicates his optimal actions to the more knowledgeable adversary, Alice, 
who fixes the posterior she computed as the baseline distribution on secret values according to which gain will be measured.
She then has two options:
gauge the gain using Bob's actions (chosen according to \emph{his} belief)
or using her actions (chosen according to \emph{her} belief, i.e., the baseline).

The difference between the expected gain of Bob's actions versus Alice's actions,
both computed against the most refined knowledge available as a baseline (Alice's posterior),
represents the advantage of choosing optimal actions according to more refined beliefs 
versus choosing them according to less refined ones.
It turns out that this difference is always non-negative
(Bob's actions will never yield a higher vulnerability than Alice's when the baseline is Alice's posterior)
and recovers static leakage under both expectation and max-case.

\subsection{Strategy-based dynamic g-leakage}\label{sec:st-leakage}

We now provide our definition for a dynamic leakage measure 
based on probability distributions over actions, which we term \emph{strategies}.  
Our strategy-based leakage measure relies on our knowledge of the order in which information flow occurs
(as when Alice has more knowledge than Bob in the thought experiment from Section~\ref{sec:st-knowledge1}),
capturing the idea that information flow corresponds to the difference in gain obtained by choosing actions 
---~i.e., defining a \emph{strategy}~---
based on a more refined state of knowledge about the secret value than on a less refined one.
The gain of such actions is then measured against the most refined state of knowledge about secrets available 
---~i.e., the \emph{baseline}.
We start by defining an adversarial strategy, obtained from the adversary's belief:

\begin{definition}[Adversarial strategy]\label{def:strategy}
  Given a gain function $\Gain : \CalW \times \CalX \to \Reals$
  modelling the adversary's goal and a probability distribution
  $q : \Dist\CalX$ representing the adversary's 
  belief about the
  secret, an optimal strategy $s$ is defined as a probability
  distribution over the actions $\CalW$ such that $s_{w^{*}} >
  0$ only if $w^{*}$ is an optimal action in the sense that it
  maximises the adversary's gain with respect to $\Gain$ and $q$; that
  is,
  \begin{equation*}
    w^{*} \in \argmax_{w \in \CalW} \sum_{x \in \CalX} q_{x}\, \Gain(w, x)
  \end{equation*}
  Analogously, for loss functions, replace $\argmax$ with $\argmin$:
  \begin{equation*}
    w^{*} \in \argmin_{w \in \CalW} \sum_{x \in \CalX} q_{x}\, \Loss(w, x)
  \end{equation*}
\end{definition}

\begin{remark}%
  \label{rmk:optimal-actions}%
  In this paper we assume that the set of optimal actions is finite 
  (but not necessarily the set of actions $\CalW$).%
  \footnote{We believe this is reasonable,
  as this is true for every case we have studied.}
\end{remark}

\Review{We note that there might be multiple strategies that are optimal according to the adversary's belief.
We write $\St(q)$ for the \emph{set} of all optimal strategies with respect to a belief $q$ and gain function $g$.
Then, taking $p : \Dist\mathcal{X}$ as the baseline,
the $\Gain$-vulnerability that results from any $s \in \St(q)$
can be defined as the expected gain of the optimal actions selected by $s$,
when secrets are distributed according to the baseline $p$:
\begin{equation}%
  \label{eq:st-vg/any}
  \sum_{w \in \mathcal{W}} s_{w} \sum_{x \in \mathcal{X}} p_{x}\, \Gain(w, x)
\end{equation}}

\Review{The adversary is, in principle, allowed to pick any strategy in $\St(q)$.
Nevertheless, we argue that, since all strategies in $\St(q)$ are equally optimal,
the adversary has no reason to favour one over another.
As such,
it seems reasonable to define dynamic $\Gain$-vulnerability (dually, $\Loss$-uncertainty) 
as an average over $s \in \St(q)$,
assigning the same probability to each $s$.
A caveat of this definition is that $\Card{\St(q)}$ may be infinite.
%
This is nonetheless solved by observing that 
the average over $\St(q)$ converges to the strategy-based $g$-vulnerability of a uniform strategy 
(i.e., uniform over all optimal actions), denoted $\St^{u}(q)$.
The next example illustrates this phenomenon:%
}

\begin{Example}{ex:avg-st}
  \Review{Let Bob's belief be $q = \OrdSet{\sfrac{1}{3}, \sfrac{1}{3}, \sfrac{1}{3}}$,
    Alice's belief after observing a run of the system be $p = \OrdSet{0, \sfrac{1}{2}, \sfrac{1}{2}}$,
    and their goal be modelled by the identity gain function $\GainBayes$ \eqref{eq:gain-bayes}.
    From Bob's point of view, all three secret values are optimal guesses,
    as per Definition~\ref{def:strategy} with $\mathcal{W} = \mathcal{X}$.
    The number of optimal strategies available for Bob is infinite, 
    but we can circumvent this by limiting the strategies that Bob can choose.%
  }

  \Review{Let the number of decimal places $n$ assigned by Bob to the probability of each action be 1. 
    The set of strategies available to Bob is 
    $\Set{\OrdSet{1, 0, 0}, \OrdSet{0.9, 0.1, 0}, \ldots}$.
    There are 66 strategies in total. 
    Among these, 11 assign probability $0$ to $x_{0}$, 
    and so, according to Alice's knowledge, lead to a vulnerability of $\sfrac{1}{2}$,
    as per \eqref{eq:st-vg/any}.
    Then, there are 10 strategies assigning probability $0.1$ to $x_{0}$,
    resulting in a vulnerability of $\sfrac{9}{20}$ each.
    The sum of the vulnerabilities is $
    \left(11 \cdot \sfrac{1}{2}\right) +
    \left(10 \cdot \sfrac{9}{20}\right) +
    \cdots +
    \left(1 \cdot 0 \right)
    = 22$.
    The average over the 66 strategies 
    ---~when picked at random~---
    then yields a vulnerability of $\sfrac{22}{66} = \sfrac{1}{3}$,
    exactly the same vulnerability if Bob chooses the uniform strategy.%
  }
\end{Example}

\Review{The convergence seen in Example~\ref{ex:avg-st} occurrs with any $n \geq 0$. \ifarxiv%
    For more details, see Appendix~\ref{sec:appendix-proofs}.
  \fi
  We thus define \emph{strategy-based} measurements 
  as the expected case over optimal actions chosen based on the adversary's belief $q$,
  when they pick one of such actions at random,%
  \footnote{%
    \Review{Alternatively, we could consider the supremum over all strategies,
    assuming that the adversary always picks the best strategy.
    We argue that this might underestimate leakage, as there is
    nothing indicating to the adversary which strategy to choose.
    For instance, in Example~\ref{ex:avg-st} leakage would be 0.}%
  } 
  and when secrets follow a more refined knowledge $p$, the baseline:
}

\Review{\begin{definition}[Strategy-based measurements]
  \label{def:st-measures}
  Let $q : \Dist\mathcal{X}$ be the adversary's belief
  and $p : \Dist\mathcal{X}$ be the baseline.
  Then, for any gain function $\Gain : \mathcal{W} \times \mathcal{X} \to \Reals$,
  the \emph{strategy-based $g$-vulnerability} caused by $q$ with respect to the baseline $p$ is
  \begin{equation*}
    \StVg{p}{q} \EqDef \sum_{w \in \mathcal{W}} \St^{u}(q)_{w} \sum_{x \in \mathcal{X}} p_{x}\, \Gain(w, x) 
  \end{equation*}
  Dually, 
  for any loss function $\Loss : \mathcal{W} \times \mathcal{X} \to \Reals_{\geq 0} \cup \Set{\infty}$,
  \begin{equation*}
    \StUl{p}{q} \EqDef \sum_{w \in \mathcal{W}} \St[\Loss]^{u}(q)_{w} \sum_{x \in \mathcal{X}} p_{x}\, \Loss(w, x)
  \end{equation*}
\end{definition}}



Finally, we formalise the strategy-based dynamic leakage based on the
difference between the vulnerability of the secret to the 
adversaries with more (posterior) and less (prior) refined beliefs.
This measures the knowledge gain to the
adversary after making an observation, and is defined (additively) as: 

\begin{definition}[Strategy-based dynamic leakage]
  \label{def:st-leak}
  For any gain function $\Gain : \CalW \times \CalX \to \Reals$,
  prior knowledge $\Prior{\pi} : \Dist\CalX$, channel
  $\Chan{C} : \CalX \to \Dist\CalY$ and output $y \in \CalY$, 
  the (additive) \emph{strategy-based dynamic $g$-leakage} $\StLeak{\Posterior{\Prior{\pi}}{\Chan{C}}[X][y]}{\Prior{\pi}}$ is defined as
  \begin{equation*}
      \StVg
      {\Posterior{\Prior{\pi}}{\Chan{C}}[X][y]}
      {\Posterior{\Prior{\pi}}{\Chan{C}}[X][y]}
      -
      \StVg
      {\Posterior{\Prior{\pi}}{\Chan{C}}[X][y]}
      {\Prior{\pi}}
  \end{equation*}
  For loss functions, we define $\StLeak[\Loss]{\Posterior{\Prior{\pi}}{\Chan{C}}[X][y]}{\Prior{\pi}}$ as
  \begin{equation*}
      \StUl
      {\Posterior{\Prior{\pi}}{\Chan{C}}[X][y]}
      {\Prior{\pi}}
      -
      \StUl
      {\Posterior{\Prior{\pi}}{\Chan{C}}[X][y]}
      {\Posterior{\Prior{\pi}}{\Chan{C}}[X][y]}
  \end{equation*}
\end{definition}

\subsection{Motivating examples revisited}\label{sec:motivation:contd}

We now revisit the motivating examples from Section~\ref{sec:motivation} 
and show how our definition of dynamic leakage resolves the identified issues of monotonicity and DPI.
We start by revisiting Example~\ref{ex:motivation:shannon},
to demonstrate that, from the point of view of an attacker,
observing an output of a system should never reduce information,
and thus a rational adversary would always rely on their posterior knowledge.

\begin{ExampleContinued}{ex:motivation:shannon}{ex:case-studies:shannon}
  Recall the deterministic program 
  that takes as input a 2-bit secret from the set of possible secrets $\CalX = \OrdSet{\textnormal{00}, \textnormal{10}, \textnormal{11}}$,
  and outputs $a$ if $X = \textnormal{00}$ or else $b$.
  This can be modelled by the deterministic channel $\Chan{B}$ below,
  which, in conjunction with the adversary's prior knowledge $\Prior{\pi} = p_{X}$,
  gives the posterior knowledge $\Hyper{\Prior{\pi}}{\Chan{B}}$:
  \begin{equation*}
    \begin{array}{c@{\;\;}c@{}c@{\;}c}
      \Prior{\pi}
      &&& \\[.5ex]
      \textnormal{00} & \ChanLDelim{3.5} & \frac{7}{8} & \ChanRDelim{3.5} \\[.5ex]
      \textnormal{10} && \frac{1}{16} & \\[.5ex]
      \textnormal{11} && \frac{1}{16} &
    \end{array}
    \hspace{-.75em}\Hyper{}{}\hspace{.15em}
    \begin{array}{c@{\;\;}c@{}cc@{\;}c}
      \Chan{B}
      && a & b & \\[.5ex]
      \textnormal{00} & \ChanLDelim{3.5} & 1 & 0 & \ChanRDelim{3.5} \\[.5ex]
      \textnormal{10} && 0 & 1 & \\[.5ex]
      \textnormal{11} && 0 & 1 &
    \end{array}
    \hspace{-.75em}=\hspace{.15em}
    \begin{array}{c@{\;\;}c@{}cc@{\;}c}
      \Hyper{\Prior{\pi}}{\Chan{B}}
      && \frac{7}{8}\,a & \frac{1}{8}\,b & \\[.5ex]
      \textnormal{00} & \ChanLDelim{3.5} & 1 & 0 & \ChanRDelim{3.5} \\[.5ex]
      \textnormal{10} && 0 & \frac{1}{2} & \\[.5ex]
      \textnormal{11} && 0 & \frac{1}{2} &
    \end{array}
  \end{equation*}

  Now, the adversary's prior strategy with respect to the loss function $\LossShannon$~\eqref{eq:loss-shannon} 
  will be to guess that the real distribution from which the secret was drawn is exactly the prior distribution $\Prior{\pi}$. 
  Therefore, the adversary's (uniform) strategy \emph{a priori},
  i.e., $\St[\LossShannon]^{u}(\Prior{\pi})$,
  is a \emph{point distribution} $\PointDist{\pi}$, 
  noting that in general
  \begin{equation}%
    \label{eq:point-dist}%
    \PointDist{x}_{y} \EqDef 1 \textnormal{ if } x = y \textnormal{ else } 0
  \end{equation}

  Then, program $\Chan{B}$ was executed and produced an output $b$. 
  We can compute the strategy-based $\LossShannon$-uncertainty of the prior $\Prior{\pi}$
  taking the posterior $\Posterior{\Prior{\pi}}{\Chan{B}}[X][b]$ as the baseline to get
  \begin{align*}
    &\;\;
      \StUl[\LossShannon]
      {\Posterior{\Prior{\pi}}{\Chan{B}}[X][b]}
      {\Prior{\pi}}
    \\ =
    &\;\;
    \sum_{w \in \CalX} \St[\LossShannon]^{u}(\Prior{\pi})_{w} \sum_{x \in \CalX}
      \Posterior{\Prior{\pi}}{\Chan{B}}[x][b]\,
      \LossShannon(w, x)
      \Comment{Definition \ref{def:st-measures}}
    \\ =
    &\;\;
      \sum_{x \in \CalX}
      \Posterior{\Prior{\pi}}{\Chan{B}}[x][b]\,
      (- \log_{2} \pi_{x})
      \Comment{$\St[\LossShannon]^{u}(\Prior{\pi}) = \PointDist{\pi}$,
        defined in \eqref{eq:point-dist}; \\[.75ex]
      Definition of $\LossShannon$ in \eqref{eq:loss-shannon}}
    \\ =
    &\;\; - \left(0\, \log_{2} \frac{7}{8} +
      \frac{1}{2}\, \log_{2} \frac{1}{16} +
      \frac{1}{2}\, \log_{2} \frac{1}{16}\right)
    =
    4
    \Comment{}
  \end{align*}
  which is greater than the strategy-based posterior uncertainty $
  \StUl[\LossShannon]
  {\Posterior{\Prior{\pi}}{\Chan{B}}[X][b]}
  {\Posterior{\Prior{\pi}}{\Chan{B}}[X][b]} 
  =
  \Ul[\LossShannon](\Posterior{\Prior{\pi}}{\Chan{B}}[X][b]) = 1.
  $
  Notice that this time leakage (Definition~\ref{def:st-leak}) is $4 - 1 = 3$,
  indicating that uncertainty decreased.
  This captures the fact that the adversary learned the first bit of the secret,
  and can be interpreted as the Kullback-Leibler (KL) divergence of the adversary's prior knowledge and the baseline,
  which we explain later.

  To flesh out why our formulation works best in this example, recall that,
  in the traditional definition of dynamic leakage (Definition~\ref{def:dyn-leak}),
  the adversary's prior state of knowledge plays the role of both belief and baseline
  when measuring the prior vulnerability of the secret.
  This is not suited for a dynamic analysis,
  as it considers that when the adversary picked ``00'' as the most likely secret a priori,
  there was indeed a $\sfrac{7}{8}$ probability of this choice being the right one,
  but the more refined knowledge obtained after observing the execution of the system
  reveals that the secret ``00'' was never possible in the first place.
  Our measure captures that by assessing prior vulnerability as the success of the prior strategy,
  based on the adversary's prior belief, against the baseline (the posterior),
  which correctly reflects that the optimal prior choice was not going to be optimal
  in the light of the more refined baseline available for this run.
  Hence, the actual leakage is positive.
\end{ExampleContinued}

Recall from Remark~\ref{rmk:bound} that we defined uncertainty measures to be bounded.
This would be trivially true if we had restricted loss functions to $\Reals_{\geq 0}$,
but $\LossShannon$~\eqref{eq:loss-shannon} has range $[0, \infty]$.
Nevertheless, as the next result shows,
we can confirm that the strategy-based $\LossShannon$-uncertainty satisfies such a requirement.

\begin{restatable}[]{lemma}{restateBoundSTLShannon}%
  \label{lemma:bound-stl-shannon}

  For any prior knowledge $\Prior{\pi} : \Dist\CalX$ and channel $\Chan{C} : \CalX \to \Dist\CalY$,
  it follows that both the strategy-based dynamic prior $\LossShannon$-uncertainty
  $
  \StUl[\LossShannon]
  {\Posterior{\Prior{\pi}}{\Chan{C}}[X][y]}
  {\Prior{\pi}}
  $
  and the strategy-based dynamic posterior $\LossShannon$-uncertainty
  $
  \StUl[\LossShannon]
  {\Posterior{\Prior{\pi}}{\Chan{C}}[X][y]}
  {\Posterior{\Prior{\pi}}{\Chan{C}}[X][y]}
  $
  are finite and non-negative.
\end{restatable}

There is an interesting correspondence between the strategy-based $\LossShannon$-uncertainty
and the Kullback-Leibler (KL) divergence of two distributions $p : \Dist\CalX$ and $q: \Dist\CalX$,
defined as
\begin{equation}\label{eq:kl}
  D_{\mathit{KL}}(p \SubsOp q)
  \EqDef
  \sum_{x \in \CalX} p_{x} \log_{2} \frac{p_{x}}{q_{x}}
\end{equation}
Notice that the result obtained in Example~\ref{ex:case-studies:shannon} 
is exactly the KL divergence of the posterior $\Posterior{\Prior{\pi}}{\Chan{C}}[X][b]$ and prior $\Prior{\pi}$:
\begin{align*}
  &\;\;
    D_{\mathit{KL}}(\Posterior{\Prior{\pi}}{\Chan{C}}[X][b] \SubsOp \Prior{\pi})
  \\ =
  &\;\; 0\, \log_{2} \frac{0}{\sfrac{7}{8}} +
    \frac{1}{2}\, \log_{2} \frac{\sfrac{1}{2}}{\sfrac{1}{16}} +
    \frac{1}{2}\, \log_{2} \frac{\sfrac{1}{2}}{\sfrac{1}{16}}
  = 
  3
  \Comment{}
\end{align*}
In fact, as the next theorem states,
this equivalence is not particular to Example~\ref{ex:case-studies:shannon},
and this is what motivated the use of the $\SubsOp$ operator in the notation for strategy-based measures.

\begin{restatable}[]{theorem}{restateKL}\label{thm:kl}
  Let $q : \Dist\mathcal{X}$ be the adversary's belief and $p : \Dist\mathcal{X}$ be the baseline.
  Then, $\StLeak[\LossShannon]{p}{q} = D_{\mathit{KL}}(p \SubsOp q).$
\end{restatable}

Notice that, in principle, KL divergence could be infinite.
However, since we are not interested in generic distributions $p$ and $q$ 
but rather in a posterior distribution obtained from a prior and a channel,
it turns out that strategy-based $\LossShannon$-leakage \emph{is} bounded,
and this follows immediately from Lemma~\ref{lemma:bound-stl-shannon}.

Going back to Example~\ref{ex:case-studies:shannon},
Theorem~\ref{thm:kl} means that the value \qm{3} we obtained as the information leakage,
according to Definition~\ref{def:st-leak},
can be expressed in terms of the KL divergence of the prior and posterior states of knowledge.
Therefore, it can be interpreted according to Shannon's source-coding theorem:
it is the average number of extra bits needed to encode the secret 
when one believes that the secret was produced from a distribution $\Prior{\pi}$,
but in reality it was obtained from the baseline $\Posterior{\Prior{\pi}}{\Chan{C}}[X][b]$.
From the attacker's point of view,
this could be understood as the number of \qm{yes/no questions} they save by changing from their prior
to the posterior knowledge.\footnote{%
  Although, from a privacy perspective, 
  it is not immediately clear whether this kind of measure is well-suited,
  and prior work suggests otherwise~\cite{Foundations:Smith2009}.%
}

Next, we revisit Example~\ref{ex:motivation:query},
to show that the misleading conclusion of the preservation of the DPI,
previously obtained with the traditional definition of dynamic vulnerability as in \eqref{eq:motivation:bayes-vuln},
is now solved under our formalisation.
We note that this behaviour 
(a break in the DPI under a dynamic perspective, even though refinement holds in the static perspective)
can only be observed in analyses where there is a third system that yields the baseline 
(in Example~\ref{ex:motivation:query}, a system that reports the original answer, available to the data engineering team).
Later, in Section~\ref{sec:soundness-dyn}, we show that the DPI always holds pairwise.
This means that a single step of post-processing is safe, 
but (counter-intuitively) further steps might not be.

\begin{ExampleContinued}{ex:motivation:query}{ex:case-studies:query}
  The system $\Chan{P}$ can be modelled by a channel $\Chan{P} : \CalX \to \Dist\CalX$,
  where $\CalX = \OrdSet{\textsc{no}, \textsc{yes}}$.
  Similarly, the second post-processing step $\Chan{S}$ can be represented as a channel $\Chan{S} : \CalX \to \Dist\CalX$.
  From these, we construct the new system $\Chan{P}\Cascade\Chan{S}$,
  as per Definition~\ref{def:cascade}.
  These channels are as follows:
  \begin{equation*}
    \begin{array}{c@{}c@{}c@{\;\;}c@{}c}
      \Chan{P}
      && \textsc{no} & \textsc{yes} & \\[.5ex]
      \textsc{no} & \ChanLDelim{2.2} & \frac{2}{3} & \frac{1}{3} & \ChanRDelim{2.2} \\[.5ex]
      \textsc{yes} && \frac{1}{3} & \frac{2}{3} &
    \end{array}
    \hspace{-.9em}\Cascade\hspace{-.15em}
    \begin{array}{c@{}c@{}c@{\;\;}c@{}c}
      \Chan{S}
      && \textsc{no} & \textsc{yes} & \\[.5ex]
      \textsc{no} & \ChanLDelim{2.2} & 1 & 0 & \ChanRDelim{2.2} \\[.5ex]
      \textsc{yes} && \frac{1}{2} & \frac{1}{2} &
    \end{array}
    \hspace{-1.15em}=\hspace{-.25em}
    \begin{array}{c@{}c@{}c@{\;\;}c@{\;}c}
      \Chan{P}\Cascade\Chan{S}
      && \textsc{no} & \textsc{yes} & \\[.5ex]
      \textsc{no} & \ChanLDelim{2.2} & \frac{5}{6} & \frac{1}{6} & \ChanRDelim{2.2} \\[.5ex]
      \textsc{yes} && \frac{2}{3} & \frac{1}{3} &
    \end{array}
  \end{equation*}
  Then, from channels $\Chan{P}$ and $\Chan{P}\Cascade\Chan{S}$,
  and assuming a uniform prior $\Prior{\upsilon} = p_{X}$,
  we can construct $\TrueAdv$'s posterior knowledge $\Hyper{\Prior{\upsilon}}{\Chan{P}}$ 
  and $\NoiseAdv$'s posterior knowledge $\Hyper{\Prior{\upsilon}}{\Chan{P}\Cascade\Chan{S}}$ 
  as, respectively,
  \begin{equation*}
    \begin{array}{c@{}c@{}c@{\;\;}c@{}c}
      \Hyper{\Prior{\upsilon}}{\Chan{P}}
      && \frac{1}{2}\,\textsc{no} & \frac{1}{2}\,\textsc{yes} & \\[.5ex]
      \textsc{no} & \ChanLDelim{2.2} & \frac{2}{3} & \frac{1}{3} & \ChanRDelim{2.2} \\[.5ex]
      \textsc{yes} && \frac{1}{3} & \frac{2}{3} &
    \end{array}
    \hspace{-.75em}\textnormal{ and }\hspace{0em}
    \begin{array}{c@{}c@{}c@{\;\;}c@{}c}
      \Hyper{\Prior{\upsilon}}{\Chan{P}\Cascade\Chan{S}}
      && \frac{3}{4}\,\textsc{no} & \frac{1}{4}\,\textsc{yes} & \\[.5ex]
      \textsc{no} & \ChanLDelim{2.2} & \frac{5}{9} & \frac{1}{3} & \ChanRDelim{2.2} \\[.5ex]
      \textsc{yes} && \frac{4}{9} & \frac{2}{3} &
    \end{array}
  \end{equation*}

  Assuming that the secret answer was \textsc{no},
  the baseline we consider is that of the data engineering team,
  namely the point distribution $\PointDist{\textsc{no}}$,
  against which we evaluate both $\TrueAdv$'s and $\NoiseAdv$'s \emph{a posteriori} strategies. 
  Starting with $\TrueAdv$, recall that she observed \textsc{yes}.
  Therefore, her strategy w.r.t.\ the gain function $\GainBayes$~\eqref{eq:gain-bayes} is
  $\St[\GainBayes]^{u}(\Posterior{\Prior{\upsilon}}{\Chan{P}}[X][\textsc{yes}])
  = \PointDist{\textsc{yes}}
  = \OrdSet{0, 1},$ 
  meaning that she will guess \textsc{yes}.
  This gives a strategy-based vulnerability of
  \begin{align*}
    &\;\;
      \StVg[\GainBayes]
      {\PointDist{\textsc{no}}}
      {\Posterior{\Prior{\upsilon}}{\Chan{P}}[X][\textsc{yes}]}
    \\ =
    &\;\;
      \sum_{w \in \CalW}
      \St[\GainBayes]^{u}(\Posterior{\Prior{\upsilon}}{\Chan{P}}[X][\textsc{yes}])_{w}
      \sum_{x \in \CalX} \PointDist{\textsc{no}}_{x}\, \GainBayes(w, x)
      \Comment{Definition~\ref{def:st-measures}}
    \\ =
    &\;\;
      \sum_{w \in \CalW}
      \PointDist{\textsc{yes}}_{w}
      \sum_{x \in \CalX} \PointDist{\textsc{no}}_{x}\, \GainBayes(w, x)
      \Comment{$
        \St[\GainBayes]^{u}(\Posterior{\Prior{\upsilon}}{\Chan{P}}[X][\textsc{yes}])
      	=
      	\PointDist{\textsc{yes}}
      $}
    \\ =
    &\;\;
    \GainBayes(\textsc{yes}, \textsc{no})
    = 0
    \Comment{%
      Definition of $\PointDist{\textsc{no}}$ and $\PointDist{\textsc{yes}}$ in \eqref{eq:point-dist}; \\
      Definition of $\GainBayes$ in \eqref{eq:gain-bayes}%
    }
  \end{align*}

  For adversary $\NoiseAdv$, recall that he observed \textsc{no} 
  and thus his strategy according to his posterior knowledge $\Posterior{\Prior{\upsilon}}{\Chan{P}\Cascade\Chan{S}}[X][\textsc{no}]$ 
  will be to guess that the original answer was \textsc{no}. 
  That is,
  $\St[\GainBayes]^{u}(\Posterior{\Prior{\upsilon}}{\Chan{P}\Cascade\Chan{S}}[X][\textsc{no}])
  = \PointDist{\textsc{no}}
  = \OrdSet{1, 0}.$ 
  Following the same reasoning above, we get a vulnerability of
  \begin{equation*}
    \StVg[\GainBayes]
    {[\textsc{no}]}
    {\Posterior{\Prior{\upsilon}}{\Chan{P}\Cascade\Chan{S}}[X][\textsc{no}]}
    =
    \GainBayes(\textsc{no}, \textsc{no})
    = 1
  \end{equation*}
  which is larger than the vulnerability computed for $\TrueAdv$.
  This shows that the strategy-based formalisation correctly leads to the conclusion that,
  from the perspective of the engineering team, who know the true answer,
  it is not always safe to replace $\Chan{P}$ with $\Chan{P}\Cascade\Chan{S}$ 
  (albeit \emph{on average} it is, as $\Chan{P} \RefinedBy \Chan{P}\Cascade\Chan{S}$).
\end{ExampleContinued}

\subsection{Knowledge ordering and its limitations}\label{sec:knowledge-ord}

Notice that Definition~\ref{def:st-measures} of strategy-based $\Gain$-vulnerability 
(or, equivalently, strategy-based $\Loss$-uncertainty)
requires two probability distributions: $p \SubsOp q$.
The distribution $p$ is what we called the \emph{baseline},
and it represents the knowledge of the more-informed adversary.
However, we have not addressed how to identify who is the more-informed adversary.

We begin by noting that, given two distributions $p$ and $q$,
there is no information inherent to $p$ or $q$ 
that a rational adversary could rely on to decide which to use as a baseline.
As we shall see later in Lemma~\ref{lemma:q-leq-p},
an adversary that derives their strategy from the baseline always reaches an optimal gain,
implying that $\StLeak{p_{A}}{p_{B}} \geq 0$, but also $\StLeak{p_{B}}{p_{A}} \geq 0$.
Hence, $\StLeakFn$ on its own is meaningless in terms of whether information flow has indeed occurred.
This is not surprising, though, as this behaviour is observed in other well-known measures,
such as the Kullback-Leibler divergence.
In the case of $D_{\mathrm{KL}}(p\SubsOp q)$,
$p$ is viewed as the true probability distribution and $q$ as the model,
but looking at the distributions $p$ and $q$ without the context does not give any information,
and flipping the roles of $p$ and $q$ still leads to a non-negative KL divergence.

However, when there is a refinement $\Chan{C} \RefinedBy \Chan{D}$ (Definition~\ref{def:refinement}), 
we know that, on average, an adversary $\Adv^{\Chan{C}}$ who observes an output of channel $\Chan{C}$ 
gains more information than an adversary $\Adv^{\Chan{D}}$ with access only to $\Chan{D}$.
Thus, although in the dynamic case it is possible that the strategy of $\Adv^{\Chan{D}}$
is better than that of $\Adv^{\Chan{C}}$, 
there is no way for either adversary to know this 
and, hence, there is no incentive for each to act differently than they do.
As such, we argue that $\Adv^{\Chan{C}}$ can be deemed more informed.

\begin{remark}\label{rmk:mono-as-dpi}
  In the case where there is only one system under analysis,
  the more informed adversary is the one with the posterior knowledge,
  after observing an execution of the system (Alice, in the thought experiment).
  We note that this case can also be modelled as a refinement:
  given a prior $\Prior{\pi}$ and a channel $\Chan{C}$,
  we can recover $\Prior{\pi}$ (Bob's knowledge) as a post-processing of $\Chan{C}$ by $\ChanLeaksNothing$ via $
  \Hyper{\Prior{\pi}}{\Chan{C}\Cascade\ChanLeaksNothing} 
  = \Hyper{\Prior{\pi}}{\ChanLeaksNothing} 
  = [\Prior{\pi}],
  $
  where $\ChanLeaksNothing$ is the channel that leaks nothing, 
  i.e., a channel mapping every input to $\bot$, representing the lack of a useful observation.
\end{remark}

From this reasoning, we argue that a knowledge ordering is \emph{necessary} 
in order to draw sensible conclusions regarding dynamic leakage using our measure,
and that $\RefinedBy$ is an appropriate relation that can be used to define such an ordering.
We remark that the refinement defined for average-case static measures 
is strictly stronger than refinement on max-case measures,
and therefore constitutes the strongest order used in QIF~\cite{Max-CaseRef:Chatzikokolakis2019}.
Going back to the operation $p \SubsOp q$ in $\StVg{p}{q}$,
there must be a prior knowledge $\Prior{\pi} : \Dist\CalX$,
channels $\Chan{B} : \CalX \to \Dist\CalY$ and $\Chan{C} : \CalX \to \Dist\CalZ$,
and outputs $y \in \CalY$ and $z \in \CalZ$ such that
\begin{itemize}
  \item $\Chan{B} \RefinedBy \Chan{C}$, 
    and thus there exists $\Chan{R}$ such that $\Chan{C} = \Chan{B}\Cascade\Chan{R}$;
  \item $p = \Posterior{\Prior{\pi}}{\Chan{B}}[X][y]$ and $q = \Posterior{\Prior{\pi}}{\Chan{C}}[X][z]$; and
  \item $\Chan{R}_{y, z} > 0$, so that $z$ is the result of post-processing $y$.
\end{itemize}

\subsection{Towards an axiomatisation of dynamic leakage}
\label{sec:soundness-dyn}

We now turn our attention to three important axioms that hold under a dynamic perspective,
when comparing two adversaries that satisfy the knowledge ordering discussed in Section~\ref{sec:knowledge-ord}: 
single-step monotonicity, single-step data-processing inequality and non-interference.
To show them,
we first consider the following lemma that states that the strategy-based $g$-vulnerability of $p \SubsOp p$ 
(i.e., when the baseline is $p$ itself)
is just the traditional $g$-vulnerability of the distribution $p$:

\begin{restatable}[]{lemma}{restateApplyPP}\label{lemma:p-p}
  Let $p : \Dist\CalX$ be both the adversary's belief and the baseline.
  Then, for any gain function $\Gain : \CalW \times \CalX \to \Reals$, 
  \begin{equation*}
    \StVg{p}{p} = \Vg(p)
  \end{equation*}
  Dually, for any loss function $\Loss : \CalW \times \CalX \to \Reals_{\geq 0} \cup \{\infty\}$, 
  \begin{equation*}
    \StUl{p}{p} = \Ul(p)
  \end{equation*}
\end{restatable}

Moreover, vulnerability is maximum when the belief from which actions are chosen coincides with the baseline.
This means that, when evaluating two adversaries with beliefs represented by distributions $p$ and $q$,
the result is only meaningful if one knows that $p$ and $q$ satisfy the knowledge ordering in some direction, 
so that there is a well-defined baseline.

\begin{restatable}[]{lemma}{restateQleqP}\label{lemma:q-leq-p}
  Let $q : \Dist\mathcal{X}$ be the adversary's belief and $p : \Dist\mathcal{X}$ be the baseline.
  Then, for any gain function $g : \CalW \times \CalX \to \Reals$, 
  \begin{equation*}
    \StVg{p}{q} \leq \StVg{p}{p} = \Vg(p)
  \end{equation*}
  with equality holding when $q = p$.
  Dually, for any loss function $\Loss : \CalW \times \CalX \to \Reals_{\geq 0} \cup \{\infty\}$, it follows that
  \begin{equation*}
    \StUl{p}{q} \geq \StUl{p}{p} = \Ul(p)
  \end{equation*}
\end{restatable}

We are now ready to state the three axioms of single-step monotonicity, single-step DPI and non-interference,
under a dynamic perspective.
We begin with single-step monotonicity, which corresponds to the perspective of an attacker,
who starts with a prior knowledge $\Prior{\pi}$,
knows the implementation details of the system $\Chan{C}$ under analysis,
and can observe the result of an execution of system $\Chan{C}$,
but cannot observe any intermediate value flowing through the sequential steps that constitute $\Chan{C}$: 

\begin{restatable}[Single-step monotonicity]{theorem}{restateDynLeak}\label{thm:mono}
  Let $\Prior{\pi} : \Dist\mathcal{X}$ be the adversary's prior belief,
  and consider a channel $\Chan{C} : \mathcal{X} \to \Dist\mathcal{Y}$ producing $y \in \CalY$.
  Then, for any gain function $\Gain : \CalW \times \CalX \to \Reals$,
  \begin{equation*}
    \StVg
    {\Posterior{\Prior{\pi}}{\Chan{C}}[X][y]}
    {\Prior{\pi}}
    \leq
    \StVg
    {\Posterior{\Prior{\pi}}{\Chan{C}}[X][y]}
    {\Posterior{\Prior{\pi}}{\Chan{C}}[X][y]}
  \end{equation*}
  Dually, for any loss function
  $\Loss : \CalW \times \CalX \to \Reals_{\geq 0} \cup \{\infty\}$,
  \begin{equation*}
    \StUl
    {\Posterior{\Prior{\pi}}{\Chan{C}}[X][y]}
    {\Prior{\pi}}
    \geq
    \StUl
    {\Posterior{\Prior{\pi}}{\Chan{C}}[X][y]}
    {\Posterior{\Prior{\pi}}{\Chan{C}}[X][y]}
  \end{equation*}
\end{restatable}

The (single-step) data-processing inequality axiom for the dynamic perspective states that,
whenever comparing two adversaries with access to channels $\Chan{C}$ and $\Chan{D}$,
whose states of knowledge are known to satisfy the knowledge ordering 
---~i.e., $\Chan{C} \RefinedBy \Chan{D}$~---
it is \emph{rational} to replace channel $\Chan{C}$ with channel $\Chan{D}$. 
Equivalently, a rational adversary would always prefer to have the extra knowledge that comes with channel $\Chan{C}$:

\begin{restatable}[Single-step DPI]{theorem}{restateCompSystem}
  \label{thm:dpi}
  Let $\Prior{\pi} : \Dist\mathcal{X}$ be the adversary's prior belief, and consider two channels 
  $\Chan{C} : \CalX \to \Dist\CalY$ and $\Chan{D} : \CalX \to \Dist\CalZ$ such that $\Chan{C} \RefinedBy \Chan{D}$.
  Let $y \in \CalY$ and $z \in \CalZ$ such that $\Chan{R}_{y, z} > 0$.
  Then, for any gain function $\Gain : \CalW \times \CalX \to \Reals$, 
  \begin{equation*}
    \StVg
    {\Posterior{\Prior{\pi}}{\Chan{C}}[X][y]}
    {\Posterior{\Prior{\pi}}{\Chan{D}}[X][z]}
    \leq
    \StVg
    {\Posterior{\Prior{\pi}}{\Chan{C}}[X][y]}
    {\Posterior{\Prior{\pi}}{\Chan{C}}[X][y]}
  \end{equation*}
  Dually, for any loss function
  $\Loss : \CalW \times \CalX \to \Reals_{\geq 0} \cup \Set{\infty}$,
  \begin{equation*}
    \StUl
    {\Posterior{\Prior{\pi}}{\Chan{C}}[X][y]}
    {\Posterior{\Prior{\pi}}{\Chan{D}}[X][z]}
    \geq
    \StUl
    {\Posterior{\Prior{\pi}}{\Chan{C}}[X][y]}
    {\Posterior{\Prior{\pi}}{\Chan{C}}[X][y]}
  \end{equation*}
\end{restatable}

Recall from Remark~\ref{rmk:mono-as-dpi} that monotonicity can be viewed as a particular case of DPI,
by setting channel $\Chan{D} = \ChanLeaksNothing$.
With this in mind, notice that both single-step monotonicity and single-step DPI 
require that the baseline is exactly one of the posteriors obtained from channel $\Chan{C}$,
with $\Chan{C} \RefinedBy \Chan{D}$.
In contrast, whenever the baseline comes from a third system $\Chan{B} : \mathcal{X} \to \Dist\mathcal{O}$,
with $\Chan{B} \RefinedBy \Chan{C} \RefinedBy \Chan{D}$, $\Chan{B} \neq \Chan{C}$
and $\Chan{B} \neq \Chan{D}$ (as in Example~\ref{ex:motivation:query}),
the DPI is not necessarily satisfied;
that is, it might be that 
\begin{equation*}%
  \StVg
  {\Posterior{\Prior{\pi}}{\Chan{B}}[X][o]}
  {\Posterior{\Prior{\pi}}{\Chan{D}}[X][z]}
  >
  \StVg
  {\Posterior{\Prior{\pi}}{\Chan{B}}[X][o]}
  {\Posterior{\Prior{\pi}}{\Chan{C}}[X][y]}
\end{equation*}
We call this a \emph{multi-step} analysis,
as there are (at least) two dynamic steps: $o \to y$ and then $y \to z$.
We will return to multi-step analysis in Section~\ref{sec:multi-step},
to understand its implications.

Finally, the non-interference axiom states that
whenever the system corresponds to channel $\ChanLeaksNothing$, 
i.e., the channel that leaks nothing,
then no information leakage can occur.
This is true under a static perspective, as $
  \Vg\Hyper{\Prior{\pi}}{\ChanLeaksNothing} 
  =
  \Vg\PointDist{\Prior{\pi}} = \Vg(\Prior{\pi}) 
$, and it is trivial to see under a dynamic perspective,
as there is only one possible output $\bot$ from which the adversary learns nothing useful 
that could guide them in changing their strategy:

\begin{restatable}[Non-interference]{theorem}{restateNI}
  \label{thm:ni}
  Let $\Prior{\pi} : \Dist\mathcal{X}$ be the adversary's prior belief.
  Then, for any gain function $\Gain : \CalW \times \CalX \to \Reals$,
  \begin{equation*}
    \StVg
    {\Posterior{\Prior{\pi}}{\ChanLeaksNothing}[X][\bot]}
    {\Posterior{\Prior{\pi}}{\ChanLeaksNothing}[X][\bot]}
    =
    \StVg
    {\Posterior{\Prior{\pi}}{\ChanLeaksNothing}[X][\bot]}
    {\Prior{\pi}}
  \end{equation*}
  Dually, for any loss function
  $\Loss : \CalW \times \CalX \to \Reals_{\geq 0} \cup \Set{\infty}$,
  \begin{equation*}
    \StUl
    {\Posterior{\Prior{\pi}}{\ChanLeaksNothing}[X][\bot]}
    {\Posterior{\Prior{\pi}}{\ChanLeaksNothing}[X][\bot]}
    =
    \StUl
    {\Posterior{\Prior{\pi}}{\ChanLeaksNothing}[X][\bot]}
    {\Prior{\pi}}
  \end{equation*}
\end{restatable}

\subsection{Consistency with the static perspective}
\label{sec:soundness-static}

In addition to the properties verified in Section~\ref{sec:soundness-dyn} under a dynamic perspective, 
we require the strategy-based formalisation to be consistent with the results obtained under a static perspective.
We begin by showing that the proposed formalisation of dynamic $g$-vulnerability 
(dually, $\Loss$-uncertainty)
recovers an important property of the traditional definition: 
the expected posterior $\Gain$-vulnerability is the expectation
of the strategy-based dynamic posterior $g$-vulnerability over all posteriors:

\begin{restatable}[]{theorem}{restateRecoverExpPosterior}%
  \label{thm:recover-exp-posterior}
  Let $\Prior{\pi} : \Dist\mathcal{X}$ be the adversary's prior belief,
  and consider a channel $\Chan{C} : \mathcal{X} \to \Dist\mathcal{Y}$ modelling the system of interest.
  Then, for any gain function $\Gain : \CalW \times \CalX \to \Reals$, 
  \begin{equation*}
    \sum_{y \in \CalY} \Outer{\Prior{\pi}}{\Chan{C}}[y]\,
    \StVg
    {\Posterior{\Prior{\pi}}{\Chan{C}}[X][y]}
    {\Posterior{\Prior{\pi}}{\Chan{C}}[X][y]}
    = \Vg\Hyper{\Prior{\pi}}{\Chan{C}}
  \end{equation*}
  Dually, for any loss function $\Loss : \CalW \times \CalX \to
  \Reals_{\geq 0} \cup \{\infty\}$,
  \begin{equation*}
    \sum_{y \in \CalY} \Outer{\Prior{\pi}}{\Chan{C}}[y]\,
    \StUl
    {\Posterior{\Prior{\pi}}{\Chan{C}}[X][y]}
    {\Posterior{\Prior{\pi}}{\Chan{C}}[X][y]}
    = \Ul\Hyper{\Prior{\pi}}{\Chan{C}}
  \end{equation*}
\end{restatable}

An analagous property can be derived for the max-case 
(or min-case, when computing $\Loss$-uncertainty).
That is, max-case posterior $g$-vulnerability 
is the maximum strategy-based dynamic posterior $g$-vulnerability over all posteriors:

\begin{restatable}[]{theorem}{restateRecoverMaxPosterior}%
  \label{thm:recover-max-posterior}
  Let $\Prior{\pi} : \Dist\mathcal{X}$ be the adversary's prior belief,
  and consider a channel $\Chan{C} : \mathcal{X} \to \Dist\mathcal{Y}$ modelling the system of interest.
  Then, for any gain function $\Gain : \CalW \times \CalX \to \Reals$, 
  \begin{equation*}
    \max_{y \in \CalY}
    \StVg
    {\Posterior{\Prior{\pi}}{\Chan{C}}[X][y]}
    {\Posterior{\Prior{\pi}}{\Chan{C}}[X][y]}
    = \VgMax\Hyper{\Prior{\pi}}{\Chan{C}}
  \end{equation*}
  Dually, for any loss function $\Loss : \CalW \times \CalX \to
  \Reals_{\geq 0} \cup \{\infty\}$,
  \begin{equation*}
    \min_{y \in \CalY}
    \StUl
    {\Posterior{\Prior{\pi}}{\Chan{C}}[X][y]}
    {\Posterior{\Prior{\pi}}{\Chan{C}}[X][y]}
    = \UlMin\Hyper{\Prior{\pi}}{\Chan{C}}
  \end{equation*}
\end{restatable}

Theorems \ref{thm:recover-exp-posterior} and \ref{thm:recover-max-posterior} 
follow directly from Lemma \ref{lemma:p-p}, 
from which we get that $
  \StVg {\Posterior{\Prior{\pi}}{\Chan{C}}[X][y]}
  {\Posterior{\Prior{\pi}}{\Chan{C}}[X][y]}
$ is exactly the traditional $\Vg( \Posterior{\Prior{\pi}}{\Chan{C}}[X][y])$.
A perhaps more interesting property
is that we can reconstruct the traditional prior $g$-vulnerability 
by averaging the strategy-based dynamic prior $g$-vulnerability over each possible posterior.
To illustrate, recall the thought experiment from Section~\ref{sec:st-knowledge1}.
If Alice has more information than Bob 
(because she made an observation that Bob has not),
but she decides to follow Bob's strategy,
she will end up with the same vulnerability as Bob would have computed.
In the end this shows that knowledge does not give you an advantage,
unless it guides the design of strategies.

\begin{restatable}[]{theorem}{restateRecoverPrior}
  \label{thm:recover-prior}
  Let $\Prior{\pi} : \Dist\mathcal{X}$ be the adversary's prior belief,
  and consider a channel $\Chan{C} : \mathcal{X} \to \Dist\mathcal{Y}$ modelling the system of interest.
  Then, for any gain function $\Gain : \CalW \times \CalX \to \Reals$, 
  \begin{equation*}
    \sum_{y \in \CalY} \Outer{\Prior{\pi}}{\Chan{C}}[y]\,
    \StVg{\Posterior{\Prior{\pi}}{\Chan{C}}[X][y]}{\Prior{\pi}}
    = \Vg(\Prior{\pi})
  \end{equation*}
  Dually, for any loss function $\Loss : \CalW \times \CalX \to
  \Reals_{\geq 0} \cup \{\infty\}$,
  \begin{equation*}
    \sum_{y \in \CalY} \Outer{\Prior{\pi}}{\Chan{C}}[y]\,
    \StUl{\Posterior{\Prior{\pi}}{\Chan{C}}[X][y]}{\Prior{\pi}}
    = \Ul(\Prior{\pi})
  \end{equation*}
\end{restatable}

Alice's extra knowledge can come from two sources:
either she observed something that Bob has not 
or Bob's observation was the result of post-processing the output observed by Alice.%
\footnote{%
  If the output observed by Bob was not the result of post-processing Alice's observation,
  then Alice's and Bob's states of knowledge are not necessarily comparable.
  The measurement of dynamic leakage relies on the knowledge ordering (Section~\ref{sec:knowledge-ord}),
  which is only known from the context;
  there is nothing explicit in the distributions of Alice and Bob 
  which can always be used to determine who has more knowledge than the other.%
}
The next two theorems thus show that 
the extra knowledge due to the lack of post-processing is, again,
only useful if it guides the construction of strategies:

\begin{restatable}[]{theorem}{restateRecoverExpRefinement}%
  \label{thm:recover-exp-refinement}
  Let $\Prior{\pi} : \Dist\mathcal{X}$ be the adversary's prior belief, and consider two channels 
  $\Chan{C} : \CalX \to \Dist\CalY$ and $\Chan{D} : \CalX \to \Dist\CalZ$, $\Chan{C} \RefinedBy \Chan{D}$.
  For any gain function $\Gain : \CalW \times \CalX \to \Reals$, 
  the \emph{expected} posterior $\Gain$-vulnerability $\Vg\Hyper{\Prior{\pi}}{\Chan{D}}$ is recovered as
  \begin{equation*}
    \sum_{z \in \CalZ} \sum_{y \in \CalY}
    \Outer{\Prior{\pi}}{\Chan{C}}[y]\, \Chan{R}_{y, z}\,
    \StVg
    {\Posterior{\Prior{\pi}}{\Chan{C}}[X][y]}
    {\Posterior{\Prior{\pi}}{\Chan{D}}[X][z]}
  \end{equation*}
  Dually, for any loss function $\Loss : \CalW \times \CalX \to \Reals_{\geq 0} \cup \{\infty\}$,
  the \emph{expected} posterior $\Loss$-uncertainty $\Ul\Hyper{\Prior{\pi}}{\Chan{D}}$ is recovered as
  \begin{equation*}
    \sum_{z \in \CalZ} \sum_{y \in \CalY}
    \Outer{\Prior{\pi}}{\Chan{C}}[y]\, \Chan{R}_{y, z}\,
    \StUl
    {\Posterior{\Prior{\pi}}{\Chan{C}}[X][y]}
    {\Posterior{\Prior{\pi}}{\Chan{D}}[X][z]}
  \end{equation*}
\end{restatable}

\begin{restatable}[]{theorem}{restateRecoverMaxRefinement}%
  \label{thm:recover-max-refinement}
  Let $\Prior{\pi} : \Dist\mathcal{X}$ be the adversary's prior belief, and consider two channels 
  $\Chan{C} : \CalX \to \Dist\CalY$ and $\Chan{D} : \CalX \to \Dist\CalZ$, $\Chan{C} \RefinedBy \Chan{D}$.
  For any gain function $\Gain : \CalW \times \CalX \to \Reals$, 
  the \emph{max-case} posterior $\Gain$-vulnerability $\VgMax\Hyper{\Prior{\pi}}{\Chan{D}}$ is recovered as
  \begin{equation*}
      \max_{z \in \CalZ} \frac{1}{\Outer{\Prior{\pi}}{\Chan{D}}[z]}
      \sum_{y \in \CalY}
      \Outer{\Prior{\pi}}{\Chan{C}}[y]\, \Chan{R}_{y, z}\,
      \StVg
      {\Posterior{\Prior{\pi}}{\Chan{C}}[X][y]}
      {\Posterior{\Prior{\pi}}{\Chan{D}}[X][z]}
    \end{equation*}
    Dually, for any loss function $\Loss : \CalW \times \CalX \to \Reals_{\geq 0} \cup \{\infty\}$,
    the \emph{min-case} posterior $\Loss$-uncertainty $\UlMin\Hyper{\Prior{\pi}}{\Chan{D}}$ is recovered as
    \begin{equation*}
      \min_{z \in \CalZ} \frac{1}{\Outer{\Prior{\pi}}{\Chan{D}}[z]}
      \sum_{y \in \CalY}
      \Outer{\Prior{\pi}}{\Chan{C}}[y]\, \Chan{R}_{y, z}\,
      \StUl
      {\Posterior{\Prior{\pi}}{\Chan{C}}[X][y]}
      {\Posterior{\Prior{\pi}}{\Chan{D}}[X][z]}
    \end{equation*}
\end{restatable}

\section{Case study: privacy-preserving data releases}
\label{sec:case-studies}

In this section we show a dynamic scenario in which post-processing is safe, 
but the traditional definition of dynamic vulnerability
(Definition~\ref{def:post-dyn-measures})
fails to capture that.
This example is set in the context of privacy-preserving dataset releases.

To satisfy privacy constraints, datasets may undergo various transformations,
which we abstract into two major steps.
First, \emph{\textbf{d}e-identification}, in which obvious personal identifiers,
such as names or official ids,
are stripped from the original data.
This is modelled as a channel $\Chan{D} : \mathcal{X} \to \Dist\mathcal{Y}$
that maps the original datasets to de-identified datasets.
Second, mitigation,
where one or more \emph{privacy \textbf{m}echanisms} are employed to protect against different kinds of privacy attacks.
This is modelled as a channel $\Chan{M} : \mathcal{Y} \to \Dist\mathcal{Z}$
that maps de-identified datasets to sanitised datasets.
The full privacy-preserving data-release process can be described by the cascading $
  \Chan{D}\Cascade\Chan{M} : \CalX{} \to \Dist{}\CalZ{}
$ (Definition~\ref{def:cascade}), 
depicted as \qm{Data Release} in Figure~\ref{fig:data-release-model}.

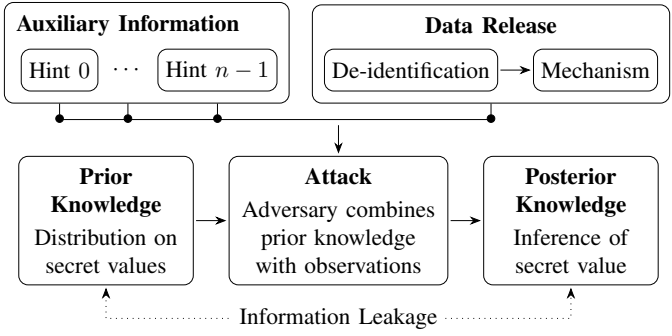
\begin{figure}[tb]
  \centering
  \resizebox{\columnwidth}{!}{\begin{tikzpicture}[
  arrow/.style = {shorten < = 1pt, shorten > = 1pt},
  inner box/.style = {draw, rounded corners, minimum height = 4ex},
  outer box/.style = {draw, rounded corners, inner sep = 5pt},
  title/.style = {inner sep = 0pt},
]
  \node[title] (auxInfo) at (0, 0) {\textbf{Auxiliary Information}};
  \node[below = 2ex of auxInfo] (hintDots) {$\cdots$};
  \node[inner box, left = .25em of hintDots] (fstHint) {Hint $0$};
  \node[inner box, right = .25em of hintDots] (lastHint) {Hint $n - 1$};
  \node[outer box, fit = (auxInfo) (hintDots) (fstHint) (lastHint)] 
    (boxAuxInfo) {};

  \coordinate[below = 3ex of fstHint] (fstHintOutput);
  \coordinate[below = 3ex of lastHint] (lastHintOutput);
  \coordinate (hintDotsOutput) at ($(hintDots |- lastHint.south) + (0, -3ex)$);

  \draw[-Circle] (fstHintOutput |- boxAuxInfo.south) to (fstHintOutput);
  \draw[-Circle] (hintDots |- boxAuxInfo.south) to (hintDotsOutput);
  \draw[-Circle] (lastHint |- boxAuxInfo.south) to (lastHintOutput);

  \node[inner box, right = 2em of lastHint] (deid) {De-identification};
  \node[inner box, right = 1.5em of deid]
    (mechanism) {Mechanism};
  \coordinate (midDataRelease) at ($(deid.west)!0.5!(mechanism.east)$);
  \node[title] (dataRelease) at (auxInfo -| midDataRelease)
    {\textbf{Data Release}};
  \node[outer box, fit = (dataRelease) (deid) (mechanism)] 
    (boxDataRelease) {};

  \draw[-Stealth, arrow] (deid) to (mechanism);

  \coordinate (dataReleaseOutput) at (lastHintOutput -| boxDataRelease);
  \draw[-Circle] (boxDataRelease) to (dataReleaseOutput);

  \node[
    outer box, below = 5ex of boxAuxInfo, xshift = -1.75em,
    text width=6.5em, text centered, minimum height=12ex
  ] (prior) {%
    \textbf{Prior Knowledge} \\[.5ex]
    Distribution on secret values
  };

  \node[
    outer box, right = 1.5em of prior,
    text width=8.4em, text centered, minimum height=12ex
  ] (attack) {%
      \textbf{Attack} \\[.5ex]
      Adversary combines prior knowledge with observations
  };

  \draw[] ($(fstHintOutput)+(0, .05)$) to ($(dataReleaseOutput)+(0, .05)$);
  \draw[-Stealth, arrow] (prior) to (attack);
  \draw[-Stealth, arrow] (dataReleaseOutput -| attack) to (attack);

  \node[
    outer box, right = 1.5em of attack,
    text width=6.5em, text centered, minimum height=12ex
  ] (posterior) {%
      \textbf{Posterior Knowledge} \\[.5ex]
      Inference of secret value
  };

  \draw[-Stealth, arrow] (attack) to (posterior);

  \draw[Stealth-Stealth, arrow, dotted] (prior |- posterior.south) 
    to +(0, -.5) 
    to node[fill = white, midway] {Information Leakage} ($(posterior.south) + (0, -.5)$) 
    to (posterior.south);
\end{tikzpicture}}%
  \caption{Attack model against privacy-preserving data releases.\label{fig:data-release-model}}
\end{figure}

We assume that the adversary has a prior knowledge about the original datasets.
Then, once the data is actually published, the adversary observes a sanitised dataset,
which leads them to update their knowledge.
We assume also that the adversary has access to some auxiliary information 
that can be used to further improve their knowledge,
often called \emph{quasi-identifiers} (QIDs)~\cite{Dalenius1986,Samarati1998,Sweeney2000,Narayanan2008}.
Formally, the adversary observes
\begin{enumerate*}[label=(\roman*)]
\item the sanitised dataset produced by the data-release system $\Chan{D}\Cascade\Chan{M}$; and

\item some QID values of one of the individuals,
  which can be viewed as the output of another channel $\Chan{H} : \CalX{} \to \Dist{}\CalQ{}$
  that produces a \emph{\textbf{h}int} from that user's record.
\end{enumerate*}

We allow multiple hints (see Figure~\ref{fig:data-release-model})
and assume these are independent of each other, and also of the output of $\Chan{D}\Cascade\Chan{M}$. 
This translates to the parallel composition
$\Chan{H}[$0$] \ParComp \cdots \ParComp \Chan{H}[$n\!\!-\!\!1$] \ParComp (\Chan{D}\Cascade\Chan{M})$,
defined for two channels $\Chan{C} : \CalX \to \Dist\CalY$ and $\Chan{D} : \CalX \to \Dist\CalZ$ as
\begin{equation}
  (\Chan{C} \ParComp \Chan{D})_{x, \Tuple{y, z}}
  \EqDef
  \Chan{C}_{x, y}\, \Chan{D}_{x, z}
  \quad\forall x \in \CalY \textnormal{ and } z \in \CalZ
\end{equation}
Figure~\ref{fig:data-release-model} then describes the attack model for data releases.

\begin{Example}{ex:case-studies:data-release}
  Consider an adversary $\NoiseAdv$ 
  who wants to determine which locations users visited, and how often, during a time period.
  Beforehand, he does not know anything.
  This is represented as a uniform prior over datasets,
  which are mappings from users to a multiset of locations visited by them:
  \begin{equation*}
    \begin{array}{c@{\;\;}c@{}c@{\;}c}
      \Prior{\upsilon}
      &&& \\[.5ex]
      x_{0} = \Map{
        \MapEntry{\textnormal{\Target}}{\Set{\textnormal{A, A}}},
        \MapEntry{\textnormal{Bob}}{\Set{\textnormal{B}}}
      } & \ChanLDelim{6} & \frac{1}{\Card{\CalX}} & \ChanRDelim{6} \\[.5ex]
      x_{1} = \Map{
        \MapEntry{\textnormal{\Target}}{\Set{\textnormal{B}}},
        \MapEntry{\textnormal{Bob}}{\Set{\textnormal{A, A}}}
      } && \frac{1}{\Card{\CalX}} & \\[.5ex]
      x_{2} = \Map{
        \MapEntry{\textnormal{\Target}}{\Set{\textnormal{A, B}}},
        \MapEntry{\textnormal{Bob}}{\Set{\textnormal{A}}}
      } && \frac{1}{\Card{\CalX}} & \\[.5ex]
      x_{3} = \Map{
        \MapEntry{\textnormal{\Target}}{\Set{\textnormal{A}}},
        \MapEntry{\textnormal{Bob}}{\Set{\textnormal{A, B}}}
      } && \frac{1}{\Card{\CalX}} & \\
      \vdots && \vdots &
    \end{array}
  \end{equation*}

  Now, suppose that $\NoiseAdv$ gains access to the de-identified database 
  $s = \Map{
    \MapEntry{0}{\Set{\textnormal{A, A}}},
    \MapEntry{1}{\Set{\textnormal{B}}}
  }$.
  Moreover, he knows that the database he received was perturbed using the mechanism
  \begin{center}
      \begin{tabular}{l}
      if location $=$ A then replace with (A $\IChoice[\sfrac{3}{4}]$ B) \\
      else replace with (A $\IChoice[\sfrac{1}{4}]$ B).
    \end{tabular}
  \end{center}
  This means that the probability of observing the sanitised dataset $s$ given, e.g.,
  that the secret is $x_{2}$, is $\sfrac{3}{4} \cdot \sfrac{1}{4} \cdot \sfrac{1}{4} = \sfrac{3}{64},$
  noting that the mechanism is applied onto each location in the dataset.  
  As such, when lifted to the domain of datasets,
  the mechanism can be represented by the following channel:
  \begin{equation*}
    \begin{array}{c@{\;}c@{}cc@{\;\;}c}
      \Chan{M}
      && s & \cdots & \\[.5ex]
      d_{0} = \Map{
         \MapEntry{1}{\Set{\textnormal{A, A}}},
         \MapEntry{2}{\Set{\textnormal{B}}}
      }
      & \ChanLDelim{3.8} & \frac{27}{64} & \cdots & \ChanRDelim{3.8} \\[.5ex]
      d_{1} = \Map{
         \MapEntry{1}{\Set{\textnormal{A, B}}},
         \MapEntry{2}{\Set{\textnormal{A}}}
      }
      && \frac{3}{64} & \cdots & \\
      \vdots && \vdots & \vdots &
    \end{array}
  \end{equation*}
  where $s = \Map{
    \MapEntry{0}{\Set{\textnormal{A, A}}},
    \MapEntry{1}{\Set{\textnormal{B}}}
  }$ is the sanitised dataset observed by $\NoiseAdv$ 
  and $d_{i}$ is one of the possible de-identified datasets produced by channel $\Chan{D}$.
  The system $\Chan{D}\Cascade\Chan{M}$ is then
  \begin{equation*}
    \begin{array}{c@{\;}c@{}c@{\;\;}c@{\;\;}c@{\;}c}
      \Chan{D}
      && d_{0} & d_{1} & \cdots & \\[.5ex]
      x_{0} & \ChanLDelim{6} & 1 & 0 & 0 & \ChanRDelim{6} \\[.5ex]
      x_{1} && 1 & 0 & 0 & \\[.5ex]
      x_{2} && 0 & 1 & 0 & \\[.5ex]
      x_{3} && 0 & 1 & 0 & \\
      \vdots && \vdots & \vdots & \vdots &
    \end{array}
    \hspace{-1em}
    \begin{array}{c@{\;}c@{}c@{\;\;}c@{\;\;}c}
      \Chan{M}
      && s & \cdots & \\[.5ex]
      d_{0} & \ChanLDelim{6} & \frac{27}{64} & \cdots & \ChanRDelim{6} \\[.5ex]
      d_{1} && \frac{3}{64} & \cdots & \\[.5ex]
      \vdots && \vdots & \vdots & \\[.5ex]
      &&&& \\
      &&&&
    \end{array}
    \hspace{-1em}=\hspace{-.5em}
    \begin{array}{c@{\;}c@{}c@{\;\;}c@{\;\;}c}
      \Chan{D}\Cascade\Chan{M}
      && s & \cdots & \\[.5ex]
      x_{0} & \ChanLDelim{6} & \frac{27}{64} & \cdots & \ChanRDelim{6} \\[.5ex]
      x_{1} && \frac{27}{64} & \cdots & \\[.5ex]
      x_{2} && \frac{3}{64} & \cdots & \\[.5ex]
      x_{3} && \frac{3}{64} & \cdots & \\
      \vdots && \vdots & \vdots &
    \end{array}
  \end{equation*}

  Suppose also that $\NoiseAdv$ observes the hint that \Target\ visited location A,
  and a histogram $k = \Map{ \MapEntry{\textnormal{A}}{2}, \MapEntry{\textnormal{B}}{1} }$ 
  with the number of occurrences of each location over the entire dataset.
  Let $\Chan{H}[Hist] : \CalX \to \Dist\CalK$ model the system that maps datasets to histograms,
  and let $\Chan{H}[\Target]$ reveal one of \Target's locations,
  defined in terms of the frequency of locations in her record.
  For example, $\Chan{H}[\Target]_{x_{2}, \textnormal{A}} = \sfrac{1}{2}$,
  as location A corresponds to half of the locations visited by \Target{} in $x_{2}$.
  Then, $\NoiseAdv$ has the composition
  \begin{equation*}
    \begin{array}{c@{}c@{}c@{\;\;}c@{}c}
      \Chan{H}[Hist]
      && k & \cdots & \\[.5ex]
      x_{0} & \ChanLDelim{6} & 1 & 0 & \ChanRDelim{6} \\[.5ex]
      x_{1} && 1 & 0 & \\[.5ex]
      x_{2} && 1 & 0 & \\[.5ex]
      x_{3} && 1 & 0 & \\
      \vdots && \raisebox{.5ex}{0} & \vdots &
    \end{array}
    \hspace{-1em}\ParComp{}\hspace{-.75em}
    \begin{array}{c@{}c@{}c@{\;\;}c@{\;}c}
      \Chan{H}[\Target]
      && \textnormal{A} & \textnormal{B} & \\[.5ex]
      x_{0} & \ChanLDelim{6} & 1 & 0 & \ChanRDelim{6} \\[.5ex]
      x_{1} && 0 & 1 & \\[.5ex]
      x_{2} && \frac{1}{2} & \frac{1}{2} & \\[.5ex]
      x_{3} && 1 & 0 & \\
      \vdots && \vdots & \vdots &
    \end{array}
    \hspace{-1.15em}=\hspace{-1.30em}
    \begin{array}{c@{\;\;}c@{}c@{\;}c@{\;}c@{\;\;}c}
      && \Tuple{k, \textnormal{A}} & \Tuple{k,\textnormal{B}} & \cdots & \\[.5ex]
      & \ChanLDelim{6} & 1 & 0 & \cdots & \ChanRDelim{6} \\[.5ex]
      && 0 & 1 & \cdots & \\[.5ex]
      && \frac{1}{2} & \frac{1}{2} & \cdots & \\[.5ex]
      && 1 & 0 & \cdots & \\
      && \raisebox{.5ex}{0} & \raisebox{.5ex}{0} & \vdots &
    \end{array}
  \end{equation*}

  Let $
    \Chan{S} =
    \Chan{H}[Hist] \ParComp
    \Chan{H}[Alex] \ParComp (\Chan{D}\Cascade\Chan{M}).
  $
  By combining his uniform prior $\upsilon$ with the two components above,
  $\NoiseAdv$ can construct the posterior knowledge $
    \Posterior
    {\Prior{\upsilon}}
    {\Chan{S}}[X][\Tuple{k, \textnormal{A}, s}]
    =
    \OrdSet{\sfrac{6}{7}, 0, \sfrac{1}{21}, \sfrac{2}{21}, \ldots}.
  $
  From this, given the observation $\Tuple{k, \textnormal{A}, s}$,
  we can conclude that $\NoiseAdv$'s best action is to guess $x_{0}$, which,
  following the traditional definition of dynamic Bayes vulnerability (Definition~\ref{def:post-dyn-measures}),
  gives him a chance of success of
  $
    \Vg[\GainBayes]
    (\Posterior{\Prior{\upsilon}}{\Chan{S}}[X][\Tuple{k, \textnormal{A}, s}])
    =
    \sfrac{6}{7}.
  $

  Now, consider another adversary $\TrueAdv$ who observes the original (de-identified) data $
    d_{1} = \Map{
      \MapEntry{0}{\Set{\textnormal{A, B}}},
      \MapEntry{1}{\Set{\textnormal{A}}}
    }
  $ before being sanitised onto dataset $s$,
  in addition to the histogram $k$ (which can be implicitly derived from $d_{1}$)
  and the hint that \Target\ visited location A.
  Then, she can construct the following posterior knowledge, where $
    \Chan{P} =
    \Chan{H}[Hist] \ParComp
    \Chan{H}[Alex] \ParComp \Chan{D}
  $: $
    \Posterior{\Prior{\upsilon}}{\Chan{P}}[X][\Tuple{k, \textnormal{A}, d_{1}}]
    =
    \OrdSet{0, 0, \sfrac{1}{3}, \sfrac{2}{3}, \ldots}.
  $
  Therefore, given the observation $\Tuple{k, \textnormal{A}, d_{1}}$,
  $\TrueAdv$'s best action is to guess $x_{3}$,
  which gives her a chance of success of $
    \Vg[\GainBayes]
    (\Posterior{\Prior{\upsilon}}{\Chan{P}}[X][\Tuple{k, \textnormal{A}, d_{1}}])
    =
    \sfrac{2}{3}.
  $

  Notice that $
    \Vg[\GainBayes]
    (\Posterior{\Prior{\upsilon}}{\Chan{S}}[X][\Tuple{k, \textnormal{A}, s}])
    >
    \Vg[\GainBayes]
    (\Posterior{\Prior{\upsilon}}{\Chan{P}}[X][\Tuple{k, \textnormal{A}, d_{1}}])
  $, suggesting that $\NoiseAdv$ has a better chance of guessing the correct secret.
  Nonetheless, his posterior knowledge is skewed towards the dataset $x_{0}$,
  which contains the wrong data: 
  we assumed that $\TrueAdv$ observed the real data,
  and there is no record \Set{A, A} in $x_{3}$.
  Hence, $\NoiseAdv$ would actually make a wrong guess;
  that is, his chance of success \emph{in practice} is zero.

  Now, let us see what happens when we employ the proposed strategy-based formalisation.
  Formally, $\NoiseAdv$'s strategy with respect to gain function $\GainBayes$ is a point distribution $\PointDist{x_{0}}$; 
  i.e., $
    \St[\GainBayes]^{u}(\Posterior{\Prior{\upsilon}}{\Chan{S}}
    [X][\Tuple{k, \textnormal{A}, s}])
    =
    \PointDist{x_{0}}
    =
    \OrdSet{1, 0, 0, 0, \ldots}.
  $
  In contrast, the strategy of $\TrueAdv$ is a point distribution $\PointDist{x_{3}}$; i.e., $
    \St[\GainBayes]^{u}(\Posterior{\Prior{\upsilon}}{\Chan{P}}
    [X][\Tuple{k, \textnormal{A}, d_{1}}])
    =
    \PointDist{x_{3}}
    =
    \OrdSet{0, 0, 0, 1, \ldots}.
  $
  Given the two strategies above,
  and taking the posterior knowledge of adversary $\TrueAdv$ as the baseline for the evaluation of such strategies
  ($\Chan{D} \RefinedBy \Chan{D}\Cascade\Chan{M}$ implies that $\Chan{P} \RefinedBy \Chan{S}$~\cite[Thm 5]{Americo2018}),
  we compute that $\NoiseAdv$'s chance of determining the secret dataset is
  \begin{align*}
    &\;\;
      \StVg[\GainBayes]
      {\Posterior{\Prior{\upsilon}}{\Chan{P}}[X][\Tuple{k, \textnormal{A}, d_{1}}]}
      {\Posterior{\Prior{\upsilon}}{\Chan{S}}[X][\Tuple{k, \textnormal{A}, s}]}
    \\ =
    &\;\;
      \sum_{w \in \CalW}
      \St[\GainBayes]^{\Review{u}}(\Posterior{\Prior{\upsilon}}{\Chan{S}}[X][%
      \Tuple{k, \textnormal{A}, s}])_{w} \sum_{x \in \CalX}
      \Posterior{\Prior{\upsilon}}{\Chan{P}}[x][\Tuple{k, \textnormal{A}, d_{1}}]\,
      \GainBayes(w, x)
    \\ =
    &\;\;
      \sum_{w \in \CalW}
      \PointDist{x_{0}}
      \sum_{x \in \CalX}
      \Posterior{\Prior{\upsilon}}{\Chan{P}}[x][\Tuple{k, \textnormal{A}, d_{1}}]\,
      \GainBayes(w, x)
    \\ =
    &\;\;
      \Posterior{\Prior{\upsilon}}{\Chan{P}}
      [x_{0}][\Tuple{k, \textnormal{A}, d_{1}}]\,
      \Comment{Definition of $\PointDist{x_{0}}$ in \eqref{eq:point-dist}; \\
      Definition of $\GainBayes$ in \eqref{eq:gain-bayes}}
    =
    0
  \end{align*}
  which is smaller than $\TrueAdv$'s chance of success, as expected:
  \begin{align*}
    &\;\;
      \StVg[\GainBayes]
      {\Posterior{\Prior{\upsilon}}{\Chan{P}}[X][\Tuple{k, \textnormal{A}, d_{1}}]}
      {\Posterior{\Prior{\upsilon}}{\Chan{P}}[X][\Tuple{k, \textnormal{A}, d_{1}}]}
    \\ =
    &\;\;
      \Vg[\GainBayes](\Posterior{\Prior{\upsilon}}{\Chan{P}}
      [X][\Tuple{k, \textnormal{A}, d_{1}}])
    = 
    \sfrac{2}{3}\qedhere
  \end{align*}
\end{Example}

As with Example~\ref{ex:motivation:query}, the example above 
shows how the traditional definition of information-flow measures leads to incorrect conclusions 
when comparing specific executions of systems (that is, dynamic scenarios).
And, as the example illustrates, this issue is solved by our formalisation.

\section{Multi-step analysis}%
\label{sec:multi-step}

In Section~\ref{sec:soundness-dyn}
we proved that the strategy-based formalisation of dynamic leakage 
satisfies relaxed versions of  the axioms of monotonicity and of the data-processing inequality.
More specifically,
these properties are satisfied whenever the analysis involves a single $\mathrm{input} \to \mathrm{output}$ dynamic step.
Figure~\ref{fig:single-step-diagram} depicts this kind of analysis:
the upper diagram represents a dynamic analysis 
with respect to some output $y \in \CalY$ from channel $\Chan{C} = \Chan{B}\Cascade\Chan{R}$,
obtained from a sequence of steps $\Dist\CalX \to \Dist\mathcal{O} \to y$,
whereas the lower diagram represents a dynamic analysis 
of some output $z \in \CalZ$ from channel $\Chan{D} = \Chan{C}\Cascade\Chan{S}$,
obtained from a sequence $\Dist\CalX \to \Dist\mathcal{O} \to y \to z$.
Notice that, in the lower diagram, 
there is only one step that is entirely dynamic (the last step $y \to z$).
This is the kind of single-step analysis covered by Theorem~\ref{thm:dpi}.

A comparison between the result of two single-step analyses,
before and after the execution of a post-processing step
(i.e., comparing the two diagrams from Figure~\ref{fig:single-step-diagram}),
corresponds to the question
\qm{Taking into account every possible path from the input that, going
  through system $\Chan{C}$, could lead to an output $y \in \CalY$, is
it safe to post-process $y$ with $\Chan{S}$?}
and the conclusion, given Theorem~\ref{thm:dpi}, is that it is safe \emph{on average}. 
Here,
even though $\Chan{D} = \Chan{B}\Cascade\Chan{R}\Cascade\Chan{S}$ involves a first post-processing step $\Chan{R}$,
we are still considering every path through $\Chan{R}$.
To illustrate that, recall Example~\ref{ex:motivation:query}
(pages~\pageref{ex:motivation:query}~and~\pageref{ex:case-studies:query}): 
instead of focusing on a particular secret answer \textsc{no},
a single-step analysis would consider every secret answer 
that could be mapped by the first post-processing step $\Chan{P}$ to \textsc{yes};
if this was the kind of analysis we were interested on,
then we could conclude that it is safe (on average) to replace $\Chan{P}$ with $\Chan{P}\Cascade\Chan{S}$.%

On the other hand,
if we also focus on a concrete output $o \in \mathcal{O}$ from $\Chan{B}$,
so that we are comparing the result of two analyses with semantics 
$\Dist\CalX \to o \to y$ vs.\ $\Dist\CalX \to o \to y \to z$,
then the comparison now involves two dynamic steps,
from $o \to y$ and then $y \to z$,
which is \emph{not} covered by Theorem~\ref{thm:dpi}.
This is what we call a \emph{multi-step} analysis,
which Figure~\ref{fig:multi-step-diagram} depicts, 
and it is exactly the kind of analysis described in Example~\ref{ex:case-studies:query},
by focusing on a concrete original answer \textsc{no},
instead of taking into account every possible original answer. 

Hence, a comparison between two multi-step analyses can be interpreted as:
\qm{Taking into account every possible path from the input that, going
  through system $\Chan{B}$, could lead to an output $o \in
  \mathcal{O}$, which was then mapped to $y \in \mathcal{Y}$ by system
$\Chan{R}$, is it safe to post-process $y$ with $\Chan{S}$?}.
In this case the conclusion is that it is not necessarily safe,
as Example~\ref{ex:case-studies:query} illustrates,
and thus a stronger version of the data-processing inequality cannot be guaranteed in general.
We reinforce, nonetheless, that the conclusions that we obtain from single- and multi-step analysis are not comparable,
as their semantics are not compatible, and therefore one result does not invalidate the other.

\begin{figure}[tb]
  \centering%
  \begin{subfigure}{\columnwidth}
  \begin{tikzpicture}
  \node[] (C/input) at (0, 0) {$\Dist\mathcal{X}$};
  \node[draw, rounded corners, inner sep = 6pt, right = 2em of C/input]
  (C/baseline) {$\Chan{B}$};

  \draw[-Stealth, shorten < = 2pt, shorten > = 2pt]
  ($(C/input.north east |- C/baseline.north west) + (0, -1ex)$)
  to ($(C/baseline.north west) + (0, -1ex)$);
  
  \draw[-Stealth, shorten < = 2pt, shorten > = 2pt]
  (C/input) to (C/baseline);

  \draw[-Stealth, shorten < = 2pt, shorten > = 2pt]
  ($(C/input.south east |- C/baseline.south west) + (0, +1ex)$)
  to ($(C/baseline.south west) + (0, +1ex)$);

  \node[draw, rounded corners, inner sep = 6pt, right = 5em of C/baseline]
  (C/C) {$\Chan{R}$};

  \draw[-Stealth, shorten < = 2pt, shorten > = 2pt]
  ($(C/baseline.north east |- C/C.north west) + (0, -1ex)$)
  to node[midway, fill=white] {\phantom{$\Dist\mathcal{O}$}}
  ($(C/C.north west) + (0, -1ex)$);
  
  \draw[-Stealth, shorten < = 2pt, shorten > = 2pt]
  ($(C/baseline.south east |- C/C.south west) + (0, +1ex)$)
  to node[midway, fill=white] {\phantom{$\Dist\mathcal{O}$}}
  ($(C/C.south west) + (0, +1ex)$);

  \draw[-Stealth, shorten < = 2pt, shorten > = 2pt]
  (C/baseline) to node[midway, fill=white] {$\Dist\mathcal{O}$} (C/C);

  \node[inner sep = 6pt, right = 5em of C/C] (C/output)
  {$\phantom{\Chan{R}}$};
  
  \draw[-Stealth, shorten < = 2pt, shorten > = 2pt]
  (C/C) to node[midway, fill=white] {$y$} (C/output);

  \node[below = 3ex of C/input] (D/input) {$\Dist\mathcal{X}$};
  \node[draw, rounded corners, inner sep = 6pt, right = 2em of D/input]
  (D/baseline) {$\Chan{B}$};

  \draw[-Stealth, shorten < = 2pt, shorten > = 2pt]
  ($(D/input.north east |- D/baseline.north west) + (0, -1ex)$)
  to ($(D/baseline.north west) + (0, -1ex)$);
  
  \draw[-Stealth, shorten < = 2pt, shorten > = 2pt]
  (D/input) to (D/baseline);

  \draw[-Stealth, shorten < = 2pt, shorten > = 2pt]
  ($(D/input.south east |- D/baseline.south west) + (0, +1ex)$)
  to ($(D/baseline.south west) + (0, +1ex)$);

  \node[draw, rounded corners, inner sep = 6pt, right = 5em of D/baseline]
  (D/C) {$\Chan{R}$};

  \draw[-Stealth, shorten < = 2pt, shorten > = 2pt]
  ($(D/baseline.north east |- D/C.north west) + (0, -1ex)$)
  to node[midway, fill=white] {\phantom{$\Dist\mathcal{O}$}}
  ($(D/C.north west) + (0, -1ex)$);
  
  \draw[-Stealth, shorten < = 2pt, shorten > = 2pt]
  ($(D/baseline.south east |- D/C.south west) + (0, +1ex)$)
  to node[midway, fill=white] {\phantom{$\Dist\mathcal{O}$}}
  ($(D/C.south west) + (0, +1ex)$);

  \draw[-Stealth, shorten < = 2pt, shorten > = 2pt]
  (D/baseline) to node[midway, fill=white] {$\Dist\mathcal{O}$} (D/C);

  \node[draw, rounded corners, inner sep = 6pt, right = 5em of D/C]
  (D/D) {$\Chan{S}$};
  
  \draw[-Stealth, shorten < = 2pt, shorten > = 2pt]
  (D/C) to node[midway, fill=white] {$y$} (D/D);
  
  \node[inner sep = 6pt, right = 5em of D/D] (D/output)
  {$\phantom{\Chan{R}}$};

  \draw[-Stealth, shorten < = 2pt, shorten > = 2pt]
  (D/D) to node[midway, fill=white] {$z$} (D/output);

  \draw[dotted, color = HighlightColor, very thick]
  ($(C/output.north west) + (-1em, +1ex)$) to ($(D/D.south west) + (-1em, -1ex)$);
\end{tikzpicture}%
  \caption{Single-step analysis.\label{fig:single-step-diagram}}%
  \hspace*{\fill}%
  \end{subfigure}

  \begin{subfigure}{\columnwidth}
  \begin{tikzpicture}
  \node[] (C/input) at (0, 0) {$\Dist\mathcal{X}$};
  \node[draw, rounded corners, inner sep = 6pt, right = 2em of C/input]
  (C/baseline) {$\Chan{B}$};

  \draw[-Stealth, shorten < = 2pt, shorten > = 2pt]
  ($(C/input.north east |- C/baseline.north west) + (0, -1ex)$)
  to ($(C/baseline.north west) + (0, -1ex)$);
  
  \draw[-Stealth, shorten < = 2pt, shorten > = 2pt]
  (C/input) to (C/baseline);

  \draw[-Stealth, shorten < = 2pt, shorten > = 2pt]
  ($(C/input.south east |- C/baseline.south west) + (0, +1ex)$)
  to ($(C/baseline.south west) + (0, +1ex)$);

  \node[draw, rounded corners, inner sep = 6pt, right = 5em of C/baseline]
  (C/C) {$\Chan{R}$};

  \draw[-Stealth, shorten < = 2pt, shorten > = 2pt]
  (C/baseline) to node[midway, fill=white] {$o$} (C/C);

  \node[inner sep = 6pt, right = 5em of C/C] (C/output)
  {$\phantom{\Chan{R}}$};
  
  \draw[-Stealth, shorten < = 2pt, shorten > = 2pt]
  (C/C) to node[midway, fill=white] {$y$} (C/output);

  \node[below = 3ex of C/input] (D/input) {$\Dist\mathcal{X}$};
  \node[draw, rounded corners, inner sep = 6pt, right = 2em of D/input]
  (D/baseline) {$\Chan{B}$};

  \draw[-Stealth, shorten < = 2pt, shorten > = 2pt]
  ($(D/input.north east |- D/baseline.north west) + (0, -1ex)$)
  to ($(D/baseline.north west) + (0, -1ex)$);
  
  \draw[-Stealth, shorten < = 2pt, shorten > = 2pt]
  (D/input) to (D/baseline);

  \draw[-Stealth, shorten < = 2pt, shorten > = 2pt]
  ($(D/input.south east |- D/baseline.south west) + (0, +1ex)$)
  to ($(D/baseline.south west) + (0, +1ex)$);

  \node[draw, rounded corners, inner sep = 6pt, right = 5em of D/baseline]
  (D/C) {$\Chan{R}$};

  \draw[-Stealth, shorten < = 2pt, shorten > = 2pt]
  (D/baseline) to node[midway, fill=white] {$o$} (D/C);

  \node[draw, rounded corners, inner sep = 6pt, right = 5em of D/C]
  (D/D) {$\Chan{S}$};
  
  \draw[-Stealth, shorten < = 2pt, shorten > = 2pt]
  (D/C) to node[midway, fill=white] {$y$} (D/D);
  
  \node[inner sep = 6pt, right = 5em of D/D] (D/output)
  {$\phantom{\Chan{S}}$};

  \draw[-Stealth, shorten < = 2pt, shorten > = 2pt]
  (D/D) to node[midway, fill=white] {$z$} (D/output);

  \draw[dotted, color = HighlightColor, very thick]
  ($(C/output.north west) + (-1em, +1ex)$) to ($(D/D.south west) + (-1em, -1ex)$);
\end{tikzpicture}%
  \caption{Multi-step analysis.\label{fig:multi-step-diagram}}%
  \hspace*{\fill}%
  \end{subfigure}
  \caption{%
      The semantics of single- and multi-step analysis
      with respect to adversaries $\Adv^{\Chan{C}}$ (upper diagram in each figure),
      who observes an output $y \in \mathcal{Y}$ from system $\Chan{C} = \Chan{B}\Cascade\Chan{R}$,
      and $\Adv^{\Chan{D}}$ (lower diagram),
      who observes $z \in \mathcal{Z}$ from system $\Chan{D} = \Chan{C} \Cascade\Chan{S}$.%
      \label{fig:multi-vs-single-step}%
  }
\end{figure}
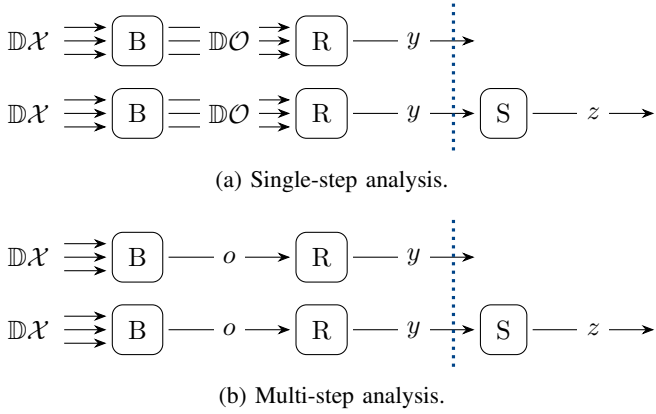

A similar reasoning follows when assessing the information that leaks from a particular system, 
setting the post-processing channel $\Chan{S} = \ChanLeaksNothing$ (see Remark~\ref{rmk:mono-as-dpi}).
In this case a single-step analysis corresponds to $\StLeak{\Posterior{\Prior{\pi}}{\Chan{C}}[X][y]}{\Prior{\pi}}$,
as defined in Definition~\ref{def:st-leak}.
This is the usual analysis we conduct,
evaluating the threat to the secret before and after the adversary observes an execution of the system,
from the adversary's point of view.
This value is always non-negative, as stated in Theorem~\ref{thm:mono},
and indicates that a rational adversary would always change their strategy once they observe an execution of the system.%

Nevertheless, a multi-step analysis is still meaningful \emph{from the point of view of a data analyst} 
who has access to an intermediate result $o \in \mathcal{O}$ that is hidden to the adversary.
Let channel $\Chan{B} : \mathcal{X} \to \Dist\mathcal{O}$ model the system that produces the intermediate result, 
so that $\Posterior{\Prior{\pi}}{\Chan{B}}[X][o]$ is the baseline.
Then, we can define $\Gain$- and $\Loss$-leakage more generally as, respectively,%
\begin{align}
  &
  \StVg
  {\Posterior{\Prior{\pi}}{\Chan{B}}[X][o]}
  {\Posterior{\Prior{\pi}}{\Chan{C}}[X][y]}
  -
  \StVg
  {\Posterior{\Prior{\pi}}{\Chan{B}}[X][o]}
  {\Prior{\pi}}
  \label{eq:extended-g-leak}%
  \\
  &
  \StUl
  {\Posterior{\Prior{\pi}}{\Chan{B}}[X][o]}
  {\Prior{\pi}}
  -
  \StUl
  {\Posterior{\Prior{\pi}}{\Chan{B}}[X][o]}
  {\Posterior{\Prior{\pi}}{\Chan{C}}[X][y]}
  \label{eq:extended-l-leak}
\end{align}

This is \emph{not} covered by Theorem~\ref{thm:mono},
since here the baseline does not come from the channel that is available to the adversary (channel $\Chan{C}$).
Therefore, it might be that, under this type of analysis,
we observe negative information leakage.
Such a result would mean that the adversary would perform better 
by constructing their strategy from their prior knowledge 
instead of using their posterior knowledge for that 
(at least, in this particular execution of the system).
Yet, as already discussed, the adversary does not have access to the intermediate result $o \in \mathcal{O}$,
and thus can only compute $\StLeak{\Posterior{\Prior{\pi}}{\Chan{C}}[X][y]}{\Prior{\pi}}$,
which is always non-negative 
and would thus lead them to update their strategy based on their posterior.
Hence, we argue that negative information leakage indicates a privacy improvement, 
as it corresponds to an \qm{opportunity loss} for the adversary.%

\CiteAuthor{DynLeak:Bielova2016}'s Example 2~\cite{DynLeak:Bielova2016},
using the Belief Tracking framework of \CiteAuthor{Belief:Clarkson2009}~\cite{Belief:Clarkson2009}, 
illustrates a case in which negative leakage can be observed.
We note that \CiteAuthor{Belief:Clarkson2009}'s framework occurs as a particular case of our formalisation.
\CiteAuthor{Belief:Clarkson2009} defines information leakage with respect to a particular secret input $x^{*}$ as follows:
\begin{equation}%
  \label{eq:belief-tracking/defn}%
  D_{\mathit{KL}}(\PointDist{x^{*}} \SubsOp \Prior{\pi}) -
  D_{\mathit{KL}}(\PointDist{x^{*}} \SubsOp \Posterior{\Prior{\pi}}{\Chan{C}}[X][y])
\end{equation}
But, recall from Theorem~\ref{thm:kl} that $D_{\mathit{KL}}(p \SubsOp q) = \StLeak[\LossShannon]{p}{q}$,
and also notice that both terms \Review{in \eqref{eq:belief-tracking/defn}} involve subtracting 
$\StUl[\LossShannon]{\PointDist{x^{*}}}{\PointDist{x^{*}}} = H(\PointDist{x^{*}}) = 0$.
Therefore, we get that \CiteAuthor{Belief:Clarkson2009}'s definition of information leakage
is exactly our general definition given in \eqref{eq:extended-l-leak},
with $\Posterior{\Prior{\pi}}{\Chan{B}}[X][o] = \PointDist{x^{*}}$.

To further clarify the notion of a negative leakage,
we now provide another example that illustrates this behaviour:

\begin{Example}{ex:negative-leakage}
  Consider the scenario of a doctor (an \qm{adversary}) 
  who is trying to diagnose a patient's disease (the \qm{secret}).
  Suppose that, based on the symptoms described by the patient,
  the doctor knows that there are three possible diseases,
  but one of them is significantly more probable than the others.
  In this case the prior knowledge of the doctor can be represented, for instance,
  by the probability distribution $\Prior{\pi} = \OrdSet{\sfrac{9}{10}, \sfrac{1}{20}, \sfrac{1}{20}}$ 
  ranging over diseases $x_{0}, x_{1}, x_{2}$.
  Based on this, the doctor would initially guess that the patient's disease is $x_{0}$.
  That is, $\St[\GainBayes]^{u}(\Prior{\pi}) = \PointDist{x_{0}}$,
  where the gain function is $\GainBayes$, as defined in~\eqref{eq:gain-bayes}.

  Then, the doctor asks for a medical test to confirm their guess.
  The result of the test is either positive (\textsc{p}),
  which confirms that the patient has the disease $x_{0}$,
  or negative (\textsc{n}), which rules out $x_{0}$ as the possible disease, with 1\% chance of error.
  The medical test can be represented by a channel $\Chan{C} : \CalX \to \Dist\OrdSet{\textsc{p}, \textsc{n}}$.
  From that, and given the doctor's prior knowledge $\Prior{\pi}$,
  their posterior knowledge $\Hyper{\Prior{\pi}}{\Chan{C}}$ would be
  \begin{equation*}%
    \begin{array}{c@{\;}c@{}c@{\;}c}
      \Prior{\pi}
      &&& \\[.5ex]
      x_{0} & \ChanLDelim{3.4} & \frac{9}{10} & \ChanRDelim{3.4} \\[.5ex]
      x_{1} && \frac{1}{20} & \\[.5ex]
      x_{2} && \frac{1}{20} &
    \end{array}
    \hspace{-.85em}\Hyper{}{}\hspace{.05em}
    \begin{array}{c@{\;}c@{}c@{\;}c@{\;}c}
      \Chan{C}
      && \textsc{p} & \textsc{n} & \\[.5ex]
      x_{0} & \ChanLDelim{3.4} & \frac{99}{100} & \frac{1}{100} & \ChanRDelim{3.4} \\[.5ex]
      x_{1} && \frac{1}{100} & \frac{99}{100} & \\[.5ex]
      x_{2} && \frac{1}{100} & \frac{99}{100} & 
    \end{array}
    \hspace{-1em}=\hspace{-.75em}
    \begin{array}{c@{}c@{}c@{\;\;}c@{\;}c}
      \Hyper{\Prior{\pi}}{\Chan{C}}
      && \frac{223}{250}\,\textsc{p} & \frac{27}{250}\,\textsc{n} & \\[.5ex]
      x_{0} & \ChanLDelim{3.4} & \frac{891}{892} & \frac{1}{12} & \ChanRDelim{3.4} \\[.5ex]
      x_{1} && \frac{1}{1784} & \frac{11}{24} & \\[.5ex]
      x_{2} && \frac{1}{1784} & \frac{11}{24} & 
    \end{array}
  \end{equation*}

  If the output is \textsc{n}, the doctor's strategy a posteriori would be 
  $\St[\GainBayes]^{u}(\Posterior{\Prior{\pi}}{\Chan{C}}[X][\textsc{n}]) = \OrdSet{0, \sfrac{1}{2}, \sfrac{1}{2}}$,
  since disease $x_{0}$ is now less likely than the others.
  Hence, the doctor would compute
  \begin{equation*}
    \StLeak{\Posterior{\Prior{\pi}}{\Chan{C}}[X][\textsc{n}]}{\Prior{\pi}}
    = \sfrac{11}{24} - \sfrac{1}{12} = \sfrac{3}{8}
  \end{equation*}
  This indicates that the doctor should change their strategy to 
  $\St[\GainBayes]^{u}(\Posterior{\Prior{\pi}}{\Chan{C}}[X][\textsc{n}])$.
  However, if we (as the \qm{data analyst}) knew that the disease is $x_{0}$ (\qm{intermediate result}),
  which can be represented by a point distribution $\PointDist{x_{0}}$,
  we would notice that
  \begin{equation*}%
    \StVg
    {\PointDist{x_{0}}}
    {\Posterior{\Prior{\pi}}{\Chan{C}}[X][\textsc{n}]}
    -
    \StVg
    {\PointDist{x_{0}}}
    {\Prior{\pi}}
    = 
    0 - 1 = -1
  \qedhere\end{equation*}
\end{Example}

Example~\ref{ex:negative-leakage} above shows that,
under the dynamic scenario where the disease is $x_{0}$ and the result of the medical test was a false negative,
the doctor's initial guess would be correct.
But, because the doctor is a \emph{rational} \qm{adversary},
the rational choice is to change their strategy once they observed the result of the medical test,
and this choice leads them to an incorrect conclusion. 
This behaviour is captured by the negative leakage.

We reinforce that a negative leakage under a multi-step analysis is not an incorrect result.
As shown in Example~\ref{ex:negative-leakage},
it represents a rational adversary who, after observing an execution of the system,
updated their prior strategy to a posterior strategy that,
for a particular execution path,
performs worse than the strategy they had a priori,
and this behaviour is only observable through the lens of the data analyst.

\section{Related Work}\label{sec:rw}

The lack of a suitable definition of dynamic leakage is well known in the literature,
and it is not particular to QIF.
In~\cite{DynLeak:Bielova2016} \CiteAuthor{DynLeak:Bielova2016} describes different measures of information leakage,
and points out the need for a new measure that evaluates information leakage with respect to one concrete output,
regardless of the input.
The approach that gets closer to such a measure is that of Belief Tracking~\cite{Belief:Clarkson2009},
which quantifies the amount of information flow caused by changes in the accuracy of the adversary's belief,
noting that the adversary's \qm{belief} does not necessarily correspond to the \qm{reality}.

The drawback of Belief Tracking is that, in principle,
it requires one to choose a reality against which we want to quantify information leakage 
(that is, to choose a secret input).
Therefore, this does not fill the gap described in~\cite{DynLeak:Bielova2016}.
In contrast, our strategy-based formalisation proposed in Section~\ref{sec:formalisation}
does evaluate information leakage with respect to one concrete output without requiring a fixed secret input.
Yet, Belief Tracking is similar to our proposal, 
in the sense that one has to choose a baseline,
but it goes to the extreme of forcing this baseline to be a point distribution.

As shown in Section~\ref{sec:multi-step}, Belief Tracking can be seen as a multi-step analysis,
as its definition coincides with \eqref{eq:extended-l-leak},
which extends our formalisation of dynamic leakage to multiple dynamic steps.
Although not explicitly defined in~\cite{Belief:Clarkson2009},
it is possible to modify the Belief Tracking framework to single-step analysis 
---~thus obtaining an instance of the information-leakage measure from Definition~\ref{def:st-leak}~---
by taking the expectation over each possible baseline $\PointDist{x}$, weighted by $\Posterior{\Prior{\pi}}{\Chan{C}}[x][y]$.

\CiteAuthor{SecretGen:Alvim2014}~\cite{SecretGen:Alvim2014} have investigated the use of ``strategies'' in QIF,
under ``dynamic secrets''.
However similar, their terminology refers to secrets that change over time (``dynamic''),
and in a scenario in which the adversary's prior is modelled by a hyper-distribution on secrets,
that is, a distribution on ``strategies'' for generating secrets.
Hence, there is little overlap between their contributions and those of this paper.

Privacy analyses
such as those of \CiteAuthor{Alvim2022}~\cite{Alvim2022} and \CiteAuthor{OpenBanking:Soares2025}~\cite{OpenBanking:Soares2025} 
can be viewed as a practical application of this dynamic perspective.
They can be formalised following the framework presented in Section~\ref{sec:case-studies} for Example~\ref{ex:case-studies:data-release}.
More precisely, they follow a sort of mixed approach:
they are dynamic with respect to the dataset observed by the attacker,
but static in regard to the hints that the adversary could obtain.

\section{Conclusion}\label{sec:conclusion}

In this paper we formalised a strategy-based definition of information leakage 
that solves the inconsistencies found when considering dynamic scenarios.
More precisely, our definition preserves the properties of single-step monotonicity 
(i.e., by observing the output of a system, an adversary cannot lose information about the input)
and of single-step data-processing inequality 
(i.e., one step of post-processing of the output of a system cannot create information; counter-intuitively,
additional post-processing might \qm{create} information).  
This approach builds on the idea raised in \cite{Belief:Clarkson2009} and \cite{SecretGen:Alvim2014} 
(albeit with different purposes)
that the adversary's belief does not necessarily correspond to the reality,
and thus a definition of information leakage should take this distinction into account.

We proved that our formulation of strategy-based $g$-leakage preserves fundamental information-theoretic properties 
enjoyed by the static definition of leakage in the QIF framework.
More precisely, 
we showed that the expectation of dynamic vulnerability corresponds to the traditional definition of expected vulnerability 
(equivalently, uncertainty).
A similar property holds for the max-case.
We also showed that our approach generalises that of Belief Tracking~\cite{Belief:Clarkson2009},
as (i) it is parameterised by a gain (or loss) function modelling the attacker's goal;
(ii) when instantiated with the loss function $\LossShannon$, it coincides with the Kullback-Leibler divergence,
the measure adopted in Belief Tracking (Theorem~\ref{thm:kl});
and (iii) it does not require one to choose a secret input as the reality.

We still require that a baseline is picked so that we can evaluate strategies against it,
but the choice of a baseline is flexible, in the sense that one should choose the best baseline available 
(as long as it satisfies the knowledge ordering).
In the cases where the goal is to quantify the information leakage of a system,
the obvious choice is the adversary's posterior knowledge for one concrete output,
whereas when comparing systems under a refinement relation,
the baseline would be the least post-processed implementation available.

The fact that it is possible that the DPI does not hold in some dynamic scenarios 
(as illustrated in Example~\ref{ex:case-studies:query})
gives room to the possibility of a stronger notion of refinement 
(perhaps too strong to be satisfied in practice).
We leave this for future work,
as well as studying other interesting properties that have already been demonstrated for the static perspective,
such as the behaviour of $g$-leakage measures under shifting and scaling of gain functions~\cite[Theorem 5.13]{QIF-Book:Alvim2020}.

\ifanonymous%
\else%
  \section*{Acknowledgment}

  \LuigiSays[TODO]{Do not include in the version submitted for major revision, 
  this is here just to check the space...}

  Luigi~D. C.~Soares was supported by CNPq (Brazil), grant number 140359/2022-2,
  and by the International Cotutelle Macquarie University Research Excellence Scholarship.
\fi
\bibliographystyle{IEEEtranN}
\bibliography{IEEEabrv,references}

\ifarxiv%
  \appendices
  \counterwithin{lemma}{section}
  \counterwithin{theorem}{section}
  \counterwithin{example}{section}

  \section{Proofs omitted from the main body}
\label{sec:appendix-proofs}

\subsection{Proofs omitted from Section~\ref{sec:st-leakage}}

In Section~\ref{sec:formalisation} we claimed that
the expected vulnerability caused by an adversary who selects one of the possible optimal strategies at random
converges to the vulnerability caused by an adversary who picks the uniform strategy.
This was illustrated in Example~\ref{ex:avg-st},
by limiting the strategies available to the adversary
to those that only assign to the optimal actions
probabilities with exactly $n = 1$ decimal place.
In this case there were 66 optimal strategies available to the adversary,
and the average over the vulnerability caused by each strategy,
as defined in \eqref{eq:st-vg/any},
was $\sfrac{1}{3}$, exactly the same as the vulnerability caused by the uniform strategy.
Then, we affirmed that this phenomenon is actually observed for any $n \geq 0$ 
that we choose as the number of decimal places allowed in the definition of strategies.
We now formalise and prove this statement:

\begin{theorem}
  \label{thm:avg-st}
  Let $q : \Dist\mathcal{X}$ be the adversary's belief about the secret
  and $\Gain : \mathcal{W} \times \mathcal{X} \to \Reals$ be the gain function modelling the adversary's goal.
  Furthermore, let $\St^{n}(q)$, $n \geq 0$, be the (finite) set of optimal strategies available to the adversary, 
  such that the probability assigned by any strategy in $\St^{n}(q)$ to any action $w \in \mathcal{W}$
  is restricted to exactly $n$ decimal places.
  Finally, let $\St^{u}(q)$ be the uniform strategy.
  Then, assuming that the set of optimal actions with respect to $q$ is finite,
  \begin{align}
    &\;\;
    \lim_{n \to \infty}
    \sum_{s \in \St^{n}(q)} \frac{1}{\Card{\St^{n}(q)}}
    \sum_{w \in \mathcal{W}} s_{w} \sum_{x \in \mathcal{X}} p_{x}\, \Gain(w, x)
    \\ =
    &\;\;
    \sum_{w \in \mathcal{W}} \St^{u}(q)_{w} \sum_{x \in \mathcal{X}} p_{x}\, \Gain(w, x)
    \Comment{}
  \end{align}
  Dually, for any loss function $\Loss : \mathcal{W} \times \mathcal{X} \to \Reals_{\geq 0} \cup \{\infty\}$,
  \begin{align}
    &\;\;
    \lim_{n \to \infty}
    \sum_{s \in \St[\Loss]^{n}(q)} \frac{1}{\Card{\St[\Loss]^{n}(q)}}
    \sum_{w \in \mathcal{W}} s_{w} \sum_{x \in \mathcal{X}} p_{x}\, \Loss(w, x)
    \\ = 
    &\;\;
    \sum_{w \in \mathcal{W}} \St[\Loss]^{u}(q)_{w} \sum_{x \in \mathcal{X}} p_{x}\, \Loss(w, x)
    \Comment{}
  \end{align}
\end{theorem}

\begin{proof}
  Recall from Definition~\ref{def:strategy} 
  that an (optimal) adversarial strategy $s \in \Dist\mathcal{W}$,
  with respect to a state of knowledge $q : \Dist\mathcal{X}$ (the adversary's belief) 
  and a gain function $g : \mathcal{X} \times \mathcal{Y} \to \Reals$,
  is a probability distribution on the set of actions $\mathcal{W}$, such that
  \begin{equation}
    s_{w^*} > 0 \implies w^{*} \in \argmax_{w \in \mathcal{W}}
    \sum_{x \in \mathcal{X}} q_{x}\,\Gain(w, x)
  \end{equation}
  and that $\St(q)$ is the set of all strategies satisfying the above.

  Let $\mathcal{W}^{*} = \argmax_{w \in \mathcal{W}} \sum_{x \in \mathcal{X}} q_{x}\,\Gain(w, x)$ be the 
  set of optimal actions (finite, as per hypothesis) 
  with respect to the state of knowledge $q$ and gain function $g$,
  such that the uniform strategy over all optimal actions is defined as follows:
  \begin{equation}\label{eq:uniform-st}
    \St^{u}(q)_{w} \EqDef \frac{1}{\Card{\mathcal{W}^{*}}} 
      \textnormal{ if } w \in \mathcal{W}^{*}
    \textnormal{ else } 0
  \end{equation}

  Fix an arbitrary $n \geq 0$ for the number of decimal places allowed in the definition of strategies.
  Then, consider the following matrix, 
  where each column is one of the (finitely many) strategies $s^{j} \in \St^{n}(q)$
  and each row is labelled by one of the (finitely many) optimal actions $w^{*}_{i} \in \mathcal{W}^{*}$:
  \begin{equation*}
    \begin{array}{c@{\;} c@{} cccc@{\;\;} c}
      && s^{0} & s^{1} & \cdots & s^{\Card{\St^{n}(q)} - 1} & \\
      w^{*}_{0} 
      & \ChanLDelim{6.5} 
      & s^{0}_{w^{*}_{0}} 
      & s^{1}_{w^{*}_{0}} 
      & \cdots 
      & s^{\Card{\St^{n}(q)} - 1}_{w^{*}_{0}} 
      & \ChanRDelim{6.5}
      \\[1ex]
      w^{*}_{1} 
      && s^{0}_{w^{*}_{1}} 
      & s^{1}_{w^{*}_{1}} 
      & \cdots 
      & s^{\Card{\St^{n}(q)} - 1}_{w^{*}_{1}} 
      & 
      \\[1ex]
      \vdots 
      && \vdots 
      & \vdots 
      & \vdots 
      & \vdots 
      & \\[1ex]
      w^{*}_{\Card{\mathcal{W}^{*}} - 1} 
      && s^{0}_{w^{*}_{\Card{\mathcal{W}^{*}} - 1}} 
      & s^{1}_{w^{*}_{\Card{\mathcal{W}^{*}} - 1}} 
      & \cdots 
      & s^{\Card{\St^{n}(q)} - 1}_{w^{*}_{\Card{\mathcal{W}^{*}} - 1}} 
      & 
    \end{array}
  \end{equation*}

  Notice that, in the matrix above, the two following facts must hold:
  \begin{enumerate*}[label=(\roman*)]
    \item each column must add up to 1, since it corresponds to a probability distribution on optimal actions, 
      so the sum of all cells in the matrix must equal the number $\Card{\St^{n}(q)}$ of optimal strategies; and

    \item each row contains \emph{all} probabilities that could be assigned by a strategy to the corresponding action 
      and therefore, by symmetry, the sum of each row must be the same, i.e.,
      it must be the sum $\Card{\St^{n}(q)}$ of all cells in the matrix divided by the number $\Card{\mathcal{W}^{*}}$ of rows.
  \end{enumerate*}
  Putting these two facts together, we conclude that the sum of each row $i$ must be:
  \begin{equation}\label{eq:avg-st/sum-row}
    \sum_{j = 0}^{\Card{\St^{n}(q)} - 1} s^{j}_{w^{*}_{i}} 
    = 
    \frac{\Card{\St^{n}(q)}}{\Card{\mathcal{W}^{*}}}
  \end{equation}

  Equipped with \eqref{eq:avg-st/sum-row}, we can then reason as follows:
  \begin{align*}
    &\;\;
    \lim_{n \to \infty} \sum_{s \in \St^{n}(q)} 
    \frac{1}{\Card{\St^{n}(q)}}
    \sum_{w \in \mathcal{W}} s_{w} \sum_{x \in \mathcal{X}} p_{x}\,g(w, x)
    \\
    =&\;\;
    \lim_{n \to \infty}
    \sum_{w^{*} \in \mathcal{W}^{*}} \sum_{x \in \mathcal{X}} p_{x}\,g(w^{*}, x)
    \\ &\;\;
    \left(\frac{1}{\Card{\St^{n}(q)}}
      \sum_{j = 0}^{\Card{\St^{n}(q)} - 1}
    s^{j}_{w^{*}}\right)
    \Comment{Rearrange terms; \\ Rewrite indices}
    \\
    =&\;\;
    \lim_{n \to \infty}
    \sum_{w^{*} \in \mathcal{W}^{*}} \sum_{x \in \mathcal{X}} p_{x}\,g(w^{*}, x)
    \left(\frac{\Card{\St^{n}(q)}}{\Card{\St^{n}(q)}\, \Card{\mathcal{W}^{*}}}\right)
    \Comment{See \eqref{eq:avg-st/sum-row}}
    \\
    =&\;\;
    \lim_{n \to \infty}
    \sum_{w^{*} \in \mathcal{W}^{*}} \sum_{x \in \mathcal{X}} p_{x}\,g(w^{*}, x)\,
    \frac{1}{\Card{\mathcal{W}^{*}}}
    \\
    =&\;\; 
    \sum_{w^{*} \in \mathcal{W}^{*}} \frac{1}{\Card{\mathcal{W}^{*}}} \sum_{x \in \mathcal{X}} p_{x}\,g(w^{*}, x)
    \Comment{Rearrange terms; \\ No depedency on $n$}
    \\
    =&\;\;
    \sum_{w \in \mathcal{W}} \St^{u}(q)_{w} \sum_{x \in \mathcal{X}} p_{x}\,g(w^{*}, x)
    \Comment{See \eqref{eq:uniform-st}}
  \end{align*}
\end{proof}

\subsection{Proofs omitted from Section~\ref{sec:motivation:contd}}
Lemma~\ref{lemma:bound-stl-shannon} can be seen as a particular case
of a more general lemma, stated as follows:

\begin{lemma}
  \label{lemma:bound-shannon-general}
  Consider two systems modelled by channels
  $\Chan{C} : \mathcal{X} \to \Dist\mathcal{Y}$ and
  $\Chan{D} : \mathcal{X} \to \Dist\CalZ$ such that
  $\Chan{C} \RefinedBy \Chan{D}$, and recall that
  $\Chan{D} = \Chan{C}\Cascade\Chan{R}$.  It follows that, for any prior knowledge
  $\Prior{\pi} : \Dist\mathcal{X}$, and outputs $y \in \mathcal{Y}$ and
  $z \in \CalZ$ such that $\Chan{R}_{y, z} > 0$,
  $
  \StUl[\LossShannon]
  {\Posterior{\Prior{\pi}}{\Chan{C}}[X][y]}
  {\Posterior{\Prior{\pi}}{\Chan{D}}[X][z]}
  $
  is finite and non-negative.
\end{lemma}

\begin{proof}
  To show that, we just need to demonstrate that if
  $\Posterior{\Prior{\pi}}{\Chan{D}}[x][z] = 0$, then
  $\Posterior{\Prior{\pi}}{\Chan{C}}[x][y] = 0$ as well, so that
  whenever
  $\LossShannon(\Posterior{\Prior{\pi}}{\Chan{D}}[x][z], x) = \infty$,
  $\Posterior{\Prior{\pi}}{\Chan{C}}[x][y] = 0$ and thus
  $\Posterior{\Prior{\pi}}{\Chan{C}}[x][y]\,
  \LossShannon(\Posterior{\Prior{\pi}}{\Chan{D}}[x][z], x) = 0$.
  
  If $\Posterior{\Prior{\pi}}{\Chan{D}}[x][z] = 0$, we
  have (from Defn~\ref{def:refinement}, Eqn~\ref{eq:posterior} and Eqn~\ref{eq:joint})
  \begin{equation*}
    \Posterior{\Prior{\pi}}{\Chan{D}}[x][z]
    =
    \frac
    {\Prior{\pi}_{x}\, (\Chan{C}\Cascade\Chan{R})_{x,z}}
    {\Outer{\Prior{\pi}}{\Chan{C}\Cascade\Chan{R}}[z]}
    =
    \frac
    {\Prior{\pi}_{x}\, \sum_{y'} \Chan{C}_{x, y'}\, \Chan{R}_{y', z}}
    {\Outer{\Prior{\pi}}{\Chan{C}\Cascade\Chan{R}}[z]}
    =
    0
  \end{equation*}
  and we know that $\Chan{R}_{y, z} > 0$. Therefore, either
  $\Prior{\pi}_{x} = 0$ or $\Chan{C}_{x, y} = 0$, and both cases imply
  that $\Posterior{\Prior{\pi}}{\Chan{C}}[x][z] = 0$.
\end{proof}

\restateBoundSTLShannon*%
\begin{proof}
  This follows directly from Lemma~\ref{lemma:bound-shannon-general},
  setting $\Chan{R} = \ChanLeaksNothing$ so that
  $\Posterior{\Prior{\pi}}{\Chan{D}}[X][z]
  =
  \Posterior{\Prior{\pi}}{\ChanLeaksNothing}[X][\bot]
  =
  \Prior{\pi}$.
\end{proof}

\restateKL*%
\begin{proof}
  Using Gibbs' inequality, 
  we can show that there is only one optimal action 
  that minimises the expected loss with respect to the loss function $\LossShannon$,
  as defined in \eqref{eq:loss-shannon}.
  That is, for any states of knowledge $\pi : \Dist\mathcal{X}$ and $\delta : \Dist\mathcal{X}$,
  it follows that
  \begin{equation}%
    \label{eq:gibbs}
    \sum_{x \in \mathcal{X}} \delta_{x}\,\LossShannon(\delta, x)
    <
    \sum_{x \in \mathcal{X}} \delta_{x}\,\LossShannon(\pi, x)
    \quad\forall \pi \neq \delta
  \end{equation}

  With this, we can reason as follows:
  \begin{align*}
    &\;\;
    \StLeak[\LossShannon]{p}{q}
    \\ =
    &\;\;
    \StUl[\LossShannon]{p}{q} - \Ul[\LossShannon](p)
    \Comment{Definition~\ref{def:st-leak}}
    \\ =
    &
    \left(\sum_{w \in \mathcal{W}} \St[\LossShannon]^{u}(q)_{w} \sum_{x \in \mathcal{X}} p_{x}\, \LossShannon(w, x)\right) -
    \Ul[\LossShannon](p)
    \Comment{Definition~\ref{def:st-measures}}
    \\ =
    &
    \left(\sum_{x \in \mathcal{X}} p_{x}\, \LossShannon(q, x)\right) -
    \Ul[\LossShannon](p)
    \Comment{%
      From \eqref{eq:gibbs} and Definition~\ref{def:strategy}, \\
      $\St[\LossShannon]^{u}(q) = \PointDist{q}$, \\
      with \PointDist{q} as defined in \eqref{eq:point-dist}%
    }
    \\ =
    &
    \left(\sum_{x \in \mathcal{X}} p_{x}\, \LossShannon(q, x)\right) -
    \left(\min_{w \in \mathcal{W}} \sum_{x \in \mathcal{X}} p_{x}\, \LossShannon(w, x)\right)
    \Comment{See~\eqref{eq:ul-dist}}
    \\ =
    &
    \left(\sum_{x \in \mathcal{X}} p_{x}\, \LossShannon(q, x)\right) -
    \left(\sum_{x \in \mathcal{X}} p_{x}\, \LossShannon(p, x)\right)
    \Comment{From \eqref{eq:gibbs}, \\ $\argmin_{w} = p$}
    \\ =
    &\;\;
    \sum_{x \in \mathcal{X}} p_{x}\, (\LossShannon(q, x) - \LossShannon(p, x))
    \Comment{Merge sum}
    \\ =
    &\;\;
    \sum_{x \in \mathcal{X}} p_{x}\, (- \log_{2} q_{x} + \log_{2} p_{x})
    \Comment{From~\eqref{eq:loss-shannon}}
    \\ =
    &\;\;
    \sum_{x \in \mathcal{X}} p_{x} \log_{2} \frac{p_{x}}{q_{x}}
    \Comment{Log difference}
    \\ =
    &\;\;
    D_{\mathit{KL}}(p \SubsOp q)
    \Comment{From~\eqref{eq:kl}}
  \end{align*}
\end{proof}

\subsection{Proofs omitted from Section~\ref{sec:soundness-dyn}}

\restateApplyPP*%
\begin{proof}
  Let $\Supp(p)$ be the support of a probability distribution $p$.
  Then, for any state of knowledge $p : \Dist\mathcal{X}$,
  any strategy $s \in \St(p)$
  and any gain function $\Gain : \mathcal{W} \times \mathcal{X} \to \Reals$, we have
  \begin{align*}
    &\;\;
    \sum_{w \in \mathcal{W}} s_{w}
    \sum_{x \in \mathcal{X}} p_{x}\, g(w, x)
    \Comment{From \eqref{eq:st-vg/any}}
    \\ =
    &\;\;
    \sum_{w^{*} \in\, \Supp(s)}
    s_{w^{*}} \sum_{x \in \mathcal{X}} p_{x}\, g(w^{*}, x)
    \\ =
    &\;\;
    \sum_{w^{*} \in\, \Supp(s)} s_{w^{*}} 
      \max_{w \in \mathcal{W}}
      \sum_{x \in \mathcal{X}} p_{x}\, g(w, x)
      \tag{$\star$}\label{eq:proof:p-p:max-step}
    \\ =
    &\;\;
    \sum_{w^{*} \in\, \Supp(s)} s_{w^{*}}\,
    \Vg(p)
    \Comment{From~\eqref{eq:vg-dist}}
    \\ =
    &\;\;
    \Vg(p)
    \sum_{w^{*} \in\, \Supp(s)} s_{w^{*}}\,
    \Comment{Rearrange terms}
    \\ =
    &\;\;
    \Vg(p)
    \Comment{$s$ is a probility distribution; \\
    $\forall w \notin \Supp(s)\colon s_{w} = 0$}
  \end{align*}
  noting that in \eqref{eq:proof:p-p:max-step} 
  every action $w^{*} \in \Supp(s)$ maximises $\sum_{x \in \mathcal{X}} p_{x}\, g(w^{*}, x)$.
  The above holds for any (optimal) strategy $s \in \St(p)$, 
  which includes the uniform strategy adopted in Defintion~\ref{def:st-measures}.
  The proof for the second part (uncertainty) can be obtained in the same way,
  replacing $\max$ with $\min$.
\end{proof}

\restateQleqP*%
\begin{proof}
  The equality $\StVg{p}{p} = \Vg(p)$ (right-hand side) follows from Lemma~\ref{lemma:p-p}.
  Then, for any $s \in \St(q)$, we have
  \begin{align*}
    &\;\;
    \sum_{w \in \mathcal{W}} s_{w} \sum_{x \in \mathcal{X}} p_{x}\, g(w, x)
    \Comment{From \eqref{eq:st-vg/any}}
    \\ \leq
    &\;\;
    \sum_{w \in \mathcal{W}} s_{w} \max_{w' \in \mathcal{W}} \sum_{x \in \mathcal{X}} p_{x}\, g(w', x)
    \\ =
    &\;\;
    \sum_{w \in \mathcal{W}} s_{w}\, \Vg(p)
    \Comment{From~\eqref{eq:vg-dist}}
    \\ =
    &\;\;
    \Vg(p) \sum_{w \in \mathcal{W}} s_{w}
    \Comment{Rearrange terms}
    \\ =
    &\;\;
    \Vg(p)
    \Comment{$s$ is a probability distribution}
  \end{align*}
  Since the above holds for any strategy $s \in \St(q)$,
  it also holds for the uniform strategy $\St^{u}(q)$ adopted in Definition~\ref{def:st-measures}.
  The proof for the second part ($\Loss$-uncertainty) can be obtained in the same way,
  replacing $\max$ with $\min$ and $\leq$ with $\geq$.
\end{proof}

\restateDynLeak*%
\begin{proof}
  This follows directly from Lemma~\ref{lemma:q-leq-p}, 
  by setting the baseline $p$ as $\Posterior{\Prior{\pi}}{\Chan{C}}[X][y]$
  and the adversary's belief $q$ as $\Prior{\pi}$.
\end{proof}

\restateCompSystem*%
\begin{proof}
  This follows directly from Lemma~\ref{lemma:q-leq-p}, 
  by setting the baseline $p$ as $\Posterior{\Prior{\pi}}{\Chan{C}}[X][y]$ 
  and the belief $q$ as $\Posterior{\Prior{\pi}}{\Chan{D}}[X][z]$.
\end{proof}

\subsection{Proofs omitted from Section~\ref{sec:soundness-static}}

\restateRecoverExpPosterior*%
\begin{proof}
  This follows immediately from Lemma~\ref{lemma:p-p}.
\end{proof}

\restateRecoverMaxPosterior*%
\begin{proof}
  This follows immediately from Lemma~\ref{lemma:p-p}.
\end{proof}

\restateRecoverPrior*%
\begin{proof}
  For any observation $y \in \mathcal{Y}$, 
  any gain function $\Gain : \mathcal{W} \times \mathcal{X} \to \Reals$
  and any strategy $s \in \St(\Prior{\pi})$,
  it follows that
  \begin{align*}
    &\;\;
    \sum_{y \in \mathcal{Y}} \Outer{\Prior{\pi}}{\Chan{C}}[y]
    \sum_{w \in \mathcal{W}} s_{w}
    \sum_{x \in \mathcal{X}} \Posterior{\Prior{\pi}}{\Chan{C}}[x][y]\, g(w, x)
    \Comment{From \eqref{eq:st-vg/any}}
    \\ =
    &\;\;
    \sum_{w \in \mathcal{W}} s_{w}
    \\
    &\;\;
    \sum_{y \in \mathcal{Y}}
    \sum_{x \in \mathcal{X}} \Outer{\Prior{\pi}}{\Chan{C}}[y]\,
    \Posterior{\Prior{\pi}}{\Chan{C}}[x][y]\, g(w, x)
    \Comment{Rearrange terms}
    \\ =
    &\;\;
    \sum_{w \in \mathcal{W}} s_{w}
    \\
    &\;\;
    \sum_{y \in \mathcal{Y}}
    \sum_{x \in \mathcal{X}} \Joint{\Prior{\pi}}{\Chan{C}}[x, y]\, g(w, x)
    \Comment{From posterior \eqref{eq:posterior} to joint}
    \\ =
    &\;\;
    \sum_{w \in \mathcal{W}} s_{w}
    \\
    &\;\;
    \sum_{y \in \mathcal{Y}}
    \sum_{x \in \mathcal{X}} \Prior{\pi}_{x}\, \Chan{C}_{x, y}\, g(w, x)
    \Comment{From joint \eqref{eq:joint} to prior \& channel}
    \\ =
    &\;\;
    \sum_{w \in \mathcal{W}} s_{w}
    \sum_{x \in \mathcal{X}} \Prior{\pi}_{x}\, g(w, x)
    \sum_{y \in \mathcal{Y}} \Chan{C}_{x, y}
    \Comment{Rearrange terms}
    \\ =
    &\;\;
    \sum_{w \in \mathcal{W}} s_{w}
    \sum_{x \in \mathcal{X}} \Prior{\pi}_{x}\, g(w, x)
    \Comment{Row $x$ is a probability distribution}
    \\ =
    &\;\;
    \Vg(\Prior{\pi})
    \Comment{Proof of Lemma~\ref{lemma:p-p}}
  \end{align*}
  Since the above holds for any strategy $s \in \St(\pi)$,
  it also holds for the uniform strategy $\St^{u}(\pi)$ adopted in Definition~\ref{def:st-measures}.
  The proof for the second part ($\Loss$-uncertainty) can be obtained in a similar way,
  and thus is omitted.
\end{proof}

To prove Theorems~\ref{thm:recover-exp-refinement} and
\ref{thm:recover-max-refinement}, we first show the
\qm{squashing} effect for a fixed output $z \in \CalZ$ from channel $\Chan{D}$:

\begin{lemma}
  \label{lemma:recover-dyn-refinement}
  Let $\Prior{\pi} : \Dist\mathcal{X}$ be the adversary's prior belief,
  and consider two channels $\Chan{C} : \mathcal{X} \to \Dist\mathcal{Y}$ and $\Chan{D} : \mathcal{X} \to \Dist\CalZ$,
  $\Chan{C} \RefinedBy \Chan{D}$,
  so that $\Chan{D} = \Chan{C}\Cascade\Chan{R}$.
  Moreover, let $z \in \mathcal{Z}$ be an output from $\Chan{D}$.
  For any gain function $\Gain : \mathcal{W} \times \mathcal{X} \to \Reals$,
  the \emph{dynamic} posterior $\Gain$-vulnerability $\Vg(\Posterior{\Prior{\pi}}{\Chan{D}}[X][z])$ is recovered as
  \begin{equation*}
    \frac{1}{\Outer{\Prior{\pi}}{\Chan{D}}[z]} \sum_{y \in \mathcal{Y}}
    \Outer{\Prior{\pi}}{\Chan{C}}[y]\, \Chan{R}_{y, z}\,
    \StVg
    {\Posterior{\Prior{\pi}}{\Chan{C}}[X][y]}
    {\Posterior{\Prior{\pi}}{\Chan{D}}[X][z]}
  \end{equation*}
  Dually, for any loss function $\Loss : \mathcal{W} \times \mathcal{X} \to \Reals_{\geq 0} \cup \{\infty\}$,
  the \emph{dynamic} posterior $\Loss$-uncertainty $\Ul(\Posterior{\Prior{\pi}}{\Chan{D}}[X][z])$ equals to
  \begin{equation*}
    \frac{1}{\Outer{\Prior{\pi}}{\Chan{D}}[z]} \sum_{y \in \mathcal{Y}}
    \Outer{\Prior{\pi}}{\Chan{C}}[y]\, \Chan{R}_{y, z}\,
    \StUl
    {\Posterior{\Prior{\pi}}{\Chan{C}}[X][y]}
    {\Posterior{\Prior{\pi}}{\Chan{D}}[X][z]}
  \end{equation*}
\end{lemma}

\begin{proof}
  For any observation $z \in \CalZ$,
  any gain function $\Gain : \mathcal{W} \times \mathcal{X} \to \Reals$
  and any strategy $s \in \St(\Posterior{\Prior{\pi}}{\Chan{D}}[X][z])$,
  it follows that
  \begin{align*}
    &\;\;
    \frac{1}{\Outer{\Prior{\pi}}{\Chan{D}}[z]}
    \sum_{y \in \mathcal{Y}} \Outer{\Prior{\pi}}{\Chan{C}}[y]\, \Chan{R}_{y, z}
    \\
    &\;\;
    \sum_{w \in \mathcal{W}}
    s_{w}
    \sum_{x \in \mathcal{X}}
    \Posterior{\Prior{\pi}}{\Chan{C}}[x][y]\, g(w, x)
    \Comment{From \eqref{eq:st-vg/any}}
    \\ =
    &\;\;
    \frac{1}{\Outer{\Prior{\pi}}{\Chan{D}}[z]}
    \sum_{y \in \mathcal{Y}}  
    \sum_{w \in \mathcal{W}}
    s_{w}
    \\
    &\;\;
    \sum_{x \in \mathcal{X}}
    \Outer{\Prior{\pi}}{\Chan{C}}[y]\,
    \Posterior{\Prior{\pi}}{\Chan{C}}[x][y]\,
    \Chan{R}_{y, z}\, g(w, x)
    \Comment{Rearrange terms}
    \\ =
    &\;\;
    \frac{1}{\Outer{\Prior{\pi}}{\Chan{D}}[z]}
    \sum_{y \in \mathcal{Y}}  
    \sum_{w \in \mathcal{W}}
    s_{w}
    \\
    &\;\;
    \sum_{x \in \mathcal{X}}
    \Joint{\Prior{\pi}}{\Chan{C}}[x, y]\,
    \Chan{R}_{y, z}\, g(w, x)
    \Comment{From posterior \eqref{eq:posterior} to joint}
    \\ =
    &\;\;
    \frac{1}{\Outer{\Prior{\pi}}{\Chan{D}}[z]}
    \sum_{y \in \mathcal{Y}} \sum_{w \in \mathcal{W}} s_{w}
    \\
    &\;\;
    \sum_{x \in \mathcal{X}} \Prior{\pi}_{x}\, \Chan{C}_{x, y}\, \Chan{R}_{y, z}\, g(w, x)
    \Comment{From joint \eqref{eq:joint} to prior \& channel}
    \\ =
    &\;\;
    \frac{1}{\Outer{\Prior{\pi}}{\Chan{D}}[z]} \sum_{w \in \mathcal{W}} s_{w}
    \\
    &\;\;
    \sum_{x \in \mathcal{X}} g(w, x)\, \Prior{\pi}_{x} \sum_{y \in \mathcal{Y}} \Chan{C}_{x, y}\, \Chan{R}_{y, z}
    \Comment{Rearrange terms}
    \\ =
    &\;\;
    \frac{1}{\Outer{\Prior{\pi}}{\Chan{D}}[z]} \sum_{w \in \mathcal{W}} s_{w}
    \\
    &\;\;
    \sum_{x \in \mathcal{X}} g(w, x)\, \Prior{\pi}_{x}\, \Chan{D}_{x, z}
    \Comment{Matrix multiplication $\Chan{C}\Chan{R}$; \\
    $\Chan{D} = \Chan{C}\Cascade\Chan{R}$ (Definition~\ref{def:cascade})}
    \\ =
    &\;\;
    \sum_{w \in \mathcal{W}} s_{w} \sum_{x \in \mathcal{X}} g(w, x)\, \frac
    {\Joint{\Prior{\pi}}{\Chan{D}}[x, z]}
    {\Outer{\Prior{\pi}}{\Chan{D}}[z]}
    \Comment{From~\eqref{eq:joint}; \\ Rearrange terms}
    \\ =
    &\;\;
    \sum_{w \in \mathcal{W}} s_{w} \sum_{x \in \mathcal{X}} g(w, x)\, \Posterior{\Prior{\pi}}{\Chan{D}}[x][z]
    \Comment{From~\eqref{eq:posterior}}
  \end{align*}
  Since the above holds for any $s \in \St(\Posterior{\Prior{\pi}}{\Chan{D}}[X][z])$,
  it also holds for the uniform strategy $\St^{u}(\Posterior{\Prior{\pi}}{\Chan{D}}[X][z])$,
  which from Defn~\ref{def:st-measures} gives 
    $\StVg
    {\Posterior{\Prior{\pi}}{\Chan{D}}[X][z]}
    {\Posterior{\Prior{\pi}}{\Chan{D}}[X][z]}
  $ and by Lemma~\ref{lemma:p-p} is $\Vg(\Posterior{\Prior{\pi}}{\Chan{D}}[X][z])$.
  The proof for the second part ($\Loss$-uncertainty) can be obtained in a similar way,
  and thus is omitted.
\end{proof}

\restateRecoverExpRefinement*%
\begin{proof}
  This follows almost immediately from Lemma~\ref{lemma:recover-dyn-refinement}:
  \begin{align*}
    &\;\;
    \sum_{z \in \CalZ} \sum_{y \in \mathcal{Y}}
    \Outer{\Prior{\pi}}{\Chan{C}}[y]\, \Chan{R}_{y, z}\,
    \StVg
    {\Posterior{\Prior{\pi}}{\Chan{C}}[X][y]}
    {\Posterior{\Prior{\pi}}{\Chan{D}}[X][z]}
    \\ =
    &\;\;
    \sum_{z \in \CalZ} \Outer{\Prior{\pi}}{\Chan{D}}[z]\,
    \\
    &\;\;
    \biggl(
      \frac{1}{\Outer{\Prior{\pi}}{\Chan{D}}[z]}
      \sum_{y \in \mathcal{Y}} \Outer{\Prior{\pi}}{\Chan{C}}[y]\,
      \Chan{R}_{y, z}\,
      \StVg
      {\Posterior{\Prior{\pi}}{\Chan{C}}[X][y]}
      {\Posterior{\Prior{\pi}}{\Chan{D}}[X][z]}
    \biggr)
    \\ =
    &\;\;
    \sum_{z \in \CalZ} \Outer{\Prior{\pi}}{\Chan{D}}[z]\,
    \Vg(\Posterior{\Prior{\pi}}{\Chan{D}}[X][z])
    \Comment{Lemma~\ref{lemma:recover-dyn-refinement}} 
    \\ =
    &\;\;
    \Vg\Hyper{\Prior{\pi}}{\Chan{D}}
    \Comment{Definition~\ref{def:post-static-measures}}
  \end{align*}
  The proof for the second part
  ($\Loss$-uncertainty) can be obtained in a similar way, and thus is omitted.
\end{proof}

\restateRecoverMaxRefinement*%
\begin{proof}
  This follows immediately from Lemma~\ref{lemma:recover-dyn-refinement}:
  \begin{align*}
    &\;\;
    \max_{z \in \CalZ} \frac{1}{\Outer{\Prior{\pi}}{\Chan{D}}[z]}
    \\
    &\;\;
    \sum_{y \in \mathcal{Y}}
    \Outer{\Prior{\pi}}{\Chan{C}}[y]\, \Chan{R}_{y, z}\,
    \StVg
    {\Posterior{\Prior{\pi}}{\Chan{C}}[X][y]}
    {\Posterior{\Prior{\pi}}{\Chan{D}}[X][z]}
    \\ =
    &\;\;
    \max_{z \in \CalZ} \Vg(\Posterior{\Prior{\pi}}{\Chan{D}}[X][z])
    \Comment{Lemma~\ref{lemma:recover-dyn-refinement}} 
    \\ =
    &\;\;
    \VgMax\Hyper{\Prior{\pi}}{\Chan{D}}
    \Comment{Definition~\ref{def:post-static-measures}}
  \end{align*}
  The proof for the second part
  ($\Loss$-uncertainty) can be obtained in a similar way, and thus is omitted.
\end{proof}

\fi
\end{document}